\documentclass[reqno,11pt]{amsart}
\usepackage[utf8]{inputenc}
\usepackage{graphicx}
\usepackage{slashed}
\usepackage{amscd}
\usepackage{amssymb}
\usepackage{esint}
\usepackage{enumitem}
\usepackage{comment}
\usepackage{amsmath}
\usepackage{dsfont}
\usepackage[mathscr]{eucal}
\ifx\makepics\undefined
\usepackage[linktocpage]{hyperref}
\else
\usepackage[pspdf=-dALLOWPSTRANSPARENCY]{auto-pst-pdf} 
\fi
\usepackage{epsfig}
\usepackage{pstricks}
\usepackage{pst-all}
\usepackage{pst-all}

\numberwithin{equation}{section}
\allowdisplaybreaks[1]
\definecolor{labelkey}{gray}{.65}

\title[The $\L$-Calculus for Causal Variational Principles]{The $\L$-Calculus for Causal Variational Principles: An Exterior Differential Calculus on Non-Smooth Spaces}

\author[F.\ Finster]{Felix Finster}
\address{Fakult\"at f\"ur Mathematik \\ Universit\"at Regensburg \\ D-93040 Regensburg \\ Germany}
\email{finster@ur.de, fvandertop@gmail.com}

\author[N.\ Kamran]{Niky Kamran}
\address{Department of Mathematics and Statistics \\ McGill University \\ Montr{\'e}al \\ Canada}
\email{niky.kamran@mcgill.ca}

\author[F.\ van der Top]{Frank van der Top \\ \\ August 2026}

\newtheorem{Def}{Definition}[section]
\newtheorem{Thm}[Def]{Theorem}
\newtheorem{Prp}[Def]{Proposition}
\newtheorem{Lemma}[Def]{Lemma}
\newtheorem{Remark}[Def]{Remark}
\newtheorem{Corollary}[Def]{Corollary}
\newtheorem{Example}[Def]{Example}

\newcommand{\Thanks}{\vspace*{.5em} \noindent \thanks}
\newcommand{\beq}{\begin{equation}}
\newcommand{\eeq}{\end{equation}}
\newcommand{\Proof}{\begin{proof}}
	\newcommand{\QED}{\end{proof} \noindent}
\newcommand{\QEDrem}{\ \hfill $\Diamond$}

\newcommand{\la}{\langle}
\newcommand{\ra}{\rangle}

\newcommand{\C}{\mathbb{C}}
\newcommand{\R}{\mathbb{R}}
\newcommand{\1}{\mbox{\rm 1 \hspace{-1.05 em} 1}}
\newcommand{\Z}{\mathbb{Z}}
\newcommand{\N}{\mathbb{N}}

\renewcommand{\L}{{\mathcal{L}}}
\newcommand{\Sact}{{\mathcal{S}}}

\newcommand\B{{\mathscr{B}}}

\renewcommand{\H}{\mathscr{H}}

\DeclareMathOperator{\dist}{dist}

\newcommand{\F}{{\mathscr{F}}}

\newcommand{\K}{{\mathscr{K}}}

\newcommand{\D}{\mathscr{D}}

\DeclareMathOperator{\sign}{sign}
\DeclareMathOperator{\supp}{supp}

\newcommand{\s}{\mathfrak{s}}

\newcommand{\bitem}{\begin{itemize}[leftmargin=2.5em]}
\newcommand{\eitem}{\end{itemize}}
\renewcommand{\div}{{\rm{div}}\,}
\newcommand{\loc}{\text{\rm{loc}}}

\newcommand{\G}{{\mathscr{G}}}
\newcommand{\x}{\mathbf{x}}
\newcommand{\y}{\mathbf{y}}

\newcommand{\h}{\mathfrak{h}}

\newcommand{\0}{{\mathbf{0}}}

\newcommand{\sshape}{\text{\rm{star}}}

\usepackage{tikz-cd}
\usepackage{float}

\DeclareFontFamily{OT1}{rsfso}{}
\DeclareFontShape{OT1}{rsfso}{m}{n}{ <-7> rsfso5 <7-10> rsfso7 <10-> rsfso10}{}
\DeclareMathAlphabet{\mycal}{OT1}{rsfso}{m}{n}

\newcommand\Felix[1]{}

\begin{document}
\maketitle
\begin{abstract}
A differential calculus for causal variational principles is developed, which generalizes the exterior calculus of differential forms and some of the associated differential topological structures to non-smooth spaces. Our calculus includes the exterior derivative, de Rham cohomology, glueing constructions (restrictions and extensions of differential forms, Mayer-Vietoris sequence), a Künneth formula and Poincar{\'e}'s lemma. Moreover, we prove versions of Stokes' theorem and the Gauß divergence theorem. The constructions and results are illustrated by several examples.
\end{abstract}

\tableofcontents

\section{Introduction} \label{secintro}
The theory of {\em{causal fermion systems}} is a recent approach to fundamental physics
(for an introduction to the physical background and applications as well as
the mathematical context, we refer the interested reader to the
review~\cite{review}, the textbooks~\cite{cfs, intro} or the 
website~\cite{cfsweblink}).
In this approach, spacetime and all structures therein are encoded in a
measure~$\rho$ on a set of operators on a Hilbert space.
The physical equations are formulated via a variational
principle for the measure~$\rho$, the so-called causal action principle.
{\em{Causal variational principles}} evolved as a mathematical generalization
of the causal action principle~\cite{continuum, jet, noncompact}
(an introduction to the causal action principle and causal variational principles
can be found for example in~\cite[Chapters~5 and~6]{intro}).
From the perspective of geometry and topology, causal variational principles
provide a general framework for describing and analyzing non-smooth
spaces and spacetimes. In this context, geometric structures such as connections and curvature have been introduced and 
studied in~\cite{lqg}. Moreover, spinors on singular spaces and topological obstructions for the existence of generalized spin structures were analyzed
in~\cite{topology}. In the present paper, we proceed in a somewhat different
direction and develop an {\em{exterior differential calculus}} on non-smooth spaces
based on the structures of causal variational principles by obtaining a generalization of differential forms and the exterior derivative
to our non-smooth setting. Our results include notions of de Rham cohomology,
generalized glueing constructions (restrictions and extensions of
differential forms, the Mayer-Vietoris sequence), a Künneth formula
and a version of Poincar{\'e}'s lemma.
We also generalize surface integrals to our setting (so-called surface layer
integrals) and prove corresponding versions of Stokes' theorem and the Gau{\ss}
divergence theorem. Our constructions are illustrated by various examples.

We now explain a few ingredients and ideas behind our constructions
(an introduction to causal variational principles will be provided in Section~\ref{secprelim}).
In the setting of causal variational principles, the system under consideration
(which can be thought of as a physical system in space or spacetime) is described
by a Radon measure~$\tilde{\rho}$ on a smooth manifold~$\F$.
The underlying {\em{space}} (or {\em{spacetime}}) denoted by~$\tilde{M}$
is defined as the support of this measure,
\[ 
\tilde{M} := \supp \tilde \rho \subset \F \:. \]
It is by definition a closed subset of~$\F$, but it does not need to be a
submanifold. Indeed, in typical and interesting situations, $\tilde M$ will {\em{not}}
be smooth; instead, it can be a singular or even discrete space
(for simple examples see~\cite[Section~6]{intro}). In order to keep the setting
in the introduction as simple as possible, we consider in the following explanations the example of a discrete space; see Figure~\ref{figdiscrete}.
\begin{figure}[tb]
\psset{xunit=.35pt,yunit=.35pt,runit=.35pt}
\begin{pspicture}(598.5941537,127.59481115)
{
\newrgbcolor{curcolor}{0 0 0}
\pscustom[linestyle=none,fillstyle=solid,fillcolor=curcolor]
{
\newpath
\moveto(50.59844906,59.89968657)
\curveto(51.32864497,57.94419428)(50.31631516,55.82072147)(48.33734716,55.15677732)
\curveto(46.35837915,54.49283317)(44.16216787,55.53983981)(43.43197196,57.4953321)
\curveto(42.70177606,59.4508244)(43.71410587,61.5742972)(45.69307387,62.23824136)
\curveto(47.67204188,62.90218551)(49.86825316,61.85517887)(50.59844906,59.89968657)
\closepath
}
}
{
\newrgbcolor{curcolor}{0 0 0}
\pscustom[linewidth=1.72947494,linecolor=curcolor]
{
\newpath
\moveto(50.59844906,59.89968657)
\curveto(51.32864497,57.94419428)(50.31631516,55.82072147)(48.33734716,55.15677732)
\curveto(46.35837915,54.49283317)(44.16216787,55.53983981)(43.43197196,57.4953321)
\curveto(42.70177606,59.4508244)(43.71410587,61.5742972)(45.69307387,62.23824136)
\curveto(47.67204188,62.90218551)(49.86825316,61.85517887)(50.59844906,59.89968657)
\closepath
}
}
{
\newrgbcolor{curcolor}{0 0 0}
\pscustom[linestyle=none,fillstyle=solid,fillcolor=curcolor]
{
\newpath
\moveto(77.04301771,80.81195297)
\curveto(77.77321362,78.85646067)(76.76088381,76.73298787)(74.7819158,76.06904371)
\curveto(72.8029478,75.40509956)(70.60673651,76.4521062)(69.87654061,78.4075985)
\curveto(69.1463447,80.36309079)(70.15867451,82.4865636)(72.13764252,83.15050775)
\curveto(74.11661052,83.8144519)(76.31282181,82.76744526)(77.04301771,80.81195297)
\closepath
}
}
{
\newrgbcolor{curcolor}{0 0 0}
\pscustom[linewidth=1.72947494,linecolor=curcolor]
{
\newpath
\moveto(77.04301771,80.81195297)
\curveto(77.77321362,78.85646067)(76.76088381,76.73298787)(74.7819158,76.06904371)
\curveto(72.8029478,75.40509956)(70.60673651,76.4521062)(69.87654061,78.4075985)
\curveto(69.1463447,80.36309079)(70.15867451,82.4865636)(72.13764252,83.15050775)
\curveto(74.11661052,83.8144519)(76.31282181,82.76744526)(77.04301771,80.81195297)
\closepath
}
}
{
\newrgbcolor{curcolor}{0 0 0}
\pscustom[linestyle=none,fillstyle=solid,fillcolor=curcolor]
{
\newpath
\moveto(118.2726809,86.70813173)
\curveto(119.00287681,84.75263943)(117.990547,82.62916663)(116.011579,81.96522247)
\curveto(114.03261099,81.30127832)(111.83639971,82.34828496)(111.1062038,84.30377726)
\curveto(110.3760079,86.25926955)(111.3883377,88.38274236)(113.36730571,89.04668651)
\curveto(115.34627372,89.71063066)(117.542485,88.66362402)(118.2726809,86.70813173)
\closepath
}
}
{
\newrgbcolor{curcolor}{0 0 0}
\pscustom[linewidth=1.72947494,linecolor=curcolor]
{
\newpath
\moveto(118.2726809,86.70813173)
\curveto(119.00287681,84.75263943)(117.990547,82.62916663)(116.011579,81.96522247)
\curveto(114.03261099,81.30127832)(111.83639971,82.34828496)(111.1062038,84.30377726)
\curveto(110.3760079,86.25926955)(111.3883377,88.38274236)(113.36730571,89.04668651)
\curveto(115.34627372,89.71063066)(117.542485,88.66362402)(118.2726809,86.70813173)
\closepath
}
}
{
\newrgbcolor{curcolor}{0 0 0}
\pscustom[linestyle=none,fillstyle=solid,fillcolor=curcolor]
{
\newpath
\moveto(153.43325555,103.49030862)
\curveto(154.16345145,101.53481632)(153.15112165,99.41134352)(151.17215364,98.74739936)
\curveto(149.19318563,98.08345521)(146.99697435,99.13046185)(146.26677845,101.08595415)
\curveto(145.53658254,103.04144644)(146.54891235,105.16491925)(148.52788036,105.8288634)
\curveto(150.50684836,106.49280755)(152.70305964,105.44580091)(153.43325555,103.49030862)
\closepath
}
}
{
\newrgbcolor{curcolor}{0 0 0}
\pscustom[linewidth=1.72947494,linecolor=curcolor]
{
\newpath
\moveto(153.43325555,103.49030862)
\curveto(154.16345145,101.53481632)(153.15112165,99.41134352)(151.17215364,98.74739936)
\curveto(149.19318563,98.08345521)(146.99697435,99.13046185)(146.26677845,101.08595415)
\curveto(145.53658254,103.04144644)(146.54891235,105.16491925)(148.52788036,105.8288634)
\curveto(150.50684836,106.49280755)(152.70305964,105.44580091)(153.43325555,103.49030862)
\closepath
}
}
{
\newrgbcolor{curcolor}{0 0 0}
\pscustom[linestyle=none,fillstyle=solid,fillcolor=curcolor]
{
\newpath
\moveto(190.64193345,98.00163609)
\curveto(191.37212936,96.04614379)(190.35979955,93.92267099)(188.38083155,93.25872684)
\curveto(186.40186354,92.59478268)(184.20565226,93.64178932)(183.47545635,95.59728162)
\curveto(182.74526045,97.55277392)(183.75759025,99.67624672)(185.73655826,100.34019088)
\curveto(187.71552627,101.00413503)(189.91173755,99.95712839)(190.64193345,98.00163609)
\closepath
}
}
{
\newrgbcolor{curcolor}{0 0 0}
\pscustom[linewidth=1.72947494,linecolor=curcolor]
{
\newpath
\moveto(190.64193345,98.00163609)
\curveto(191.37212936,96.04614379)(190.35979955,93.92267099)(188.38083155,93.25872684)
\curveto(186.40186354,92.59478268)(184.20565226,93.64178932)(183.47545635,95.59728162)
\curveto(182.74526045,97.55277392)(183.75759025,99.67624672)(185.73655826,100.34019088)
\curveto(187.71552627,101.00413503)(189.91173755,99.95712839)(190.64193345,98.00163609)
\closepath
}
}
{
\newrgbcolor{curcolor}{0 0 0}
\pscustom[linestyle=none,fillstyle=solid,fillcolor=curcolor]
{
\newpath
\moveto(257.41376743,99.14321022)
\curveto(258.14396334,97.18771793)(257.13163353,95.06424512)(255.15266552,94.40030097)
\curveto(253.17369752,93.73635682)(250.97748623,94.78336346)(250.24729033,96.73885575)
\curveto(249.51709442,98.69434805)(250.52942423,100.81782085)(252.50839224,101.48176501)
\curveto(254.48736024,102.14570916)(256.68357153,101.09870252)(257.41376743,99.14321022)
\closepath
}
}
{
\newrgbcolor{curcolor}{0 0 0}
\pscustom[linewidth=1.72947494,linecolor=curcolor]
{
\newpath
\moveto(257.41376743,99.14321022)
\curveto(258.14396334,97.18771793)(257.13163353,95.06424512)(255.15266552,94.40030097)
\curveto(253.17369752,93.73635682)(250.97748623,94.78336346)(250.24729033,96.73885575)
\curveto(249.51709442,98.69434805)(250.52942423,100.81782085)(252.50839224,101.48176501)
\curveto(254.48736024,102.14570916)(256.68357153,101.09870252)(257.41376743,99.14321022)
\closepath
}
}
{
\newrgbcolor{curcolor}{0 0 0}
\pscustom[linestyle=none,fillstyle=solid,fillcolor=curcolor]
{
\newpath
\moveto(226.6732125,98.95567244)
\curveto(227.40340841,97.00018015)(226.3910786,94.87670734)(224.41211059,94.21276319)
\curveto(222.43314259,93.54881904)(220.2369313,94.59582568)(219.5067354,96.55131797)
\curveto(218.77653949,98.50681027)(219.7888693,100.63028307)(221.76783731,101.29422723)
\curveto(223.74680531,101.95817138)(225.9430166,100.91116474)(226.6732125,98.95567244)
\closepath
}
}
{
\newrgbcolor{curcolor}{0 0 0}
\pscustom[linewidth=1.72947494,linecolor=curcolor]
{
\newpath
\moveto(226.6732125,98.95567244)
\curveto(227.40340841,97.00018015)(226.3910786,94.87670734)(224.41211059,94.21276319)
\curveto(222.43314259,93.54881904)(220.2369313,94.59582568)(219.5067354,96.55131797)
\curveto(218.77653949,98.50681027)(219.7888693,100.63028307)(221.76783731,101.29422723)
\curveto(223.74680531,101.95817138)(225.9430166,100.91116474)(226.6732125,98.95567244)
\closepath
}
}
{
\newrgbcolor{curcolor}{0 0 0}
\pscustom[linestyle=none,fillstyle=solid,fillcolor=curcolor]
{
\newpath
\moveto(290.99856429,85.26668274)
\curveto(291.72876019,83.31119044)(290.71643038,81.18771764)(288.73746238,80.52377348)
\curveto(286.75849437,79.85982933)(284.56228309,80.90683597)(283.83208718,82.86232827)
\curveto(283.10189128,84.81782056)(284.11422109,86.94129337)(286.09318909,87.60523752)
\curveto(288.0721571,88.26918167)(290.26836838,87.22217503)(290.99856429,85.26668274)
\closepath
}
}
{
\newrgbcolor{curcolor}{0 0 0}
\pscustom[linewidth=1.72947494,linecolor=curcolor]
{
\newpath
\moveto(290.99856429,85.26668274)
\curveto(291.72876019,83.31119044)(290.71643038,81.18771764)(288.73746238,80.52377348)
\curveto(286.75849437,79.85982933)(284.56228309,80.90683597)(283.83208718,82.86232827)
\curveto(283.10189128,84.81782056)(284.11422109,86.94129337)(286.09318909,87.60523752)
\curveto(288.0721571,88.26918167)(290.26836838,87.22217503)(290.99856429,85.26668274)
\closepath
}
}
{
\newrgbcolor{curcolor}{0 0 0}
\pscustom[linestyle=none,fillstyle=solid,fillcolor=curcolor]
{
\newpath
\moveto(326.10410526,64.11890662)
\curveto(326.83430116,62.16341432)(325.82197135,60.03994152)(323.84300335,59.37599737)
\curveto(321.86403534,58.71205321)(319.66782406,59.75905985)(318.93762815,61.71455215)
\curveto(318.20743225,63.67004445)(319.21976206,65.79351725)(321.19873006,66.4574614)
\curveto(323.17769807,67.12140556)(325.37390935,66.07439892)(326.10410526,64.11890662)
\closepath
}
}
{
\newrgbcolor{curcolor}{0 0 0}
\pscustom[linewidth=1.72947494,linecolor=curcolor]
{
\newpath
\moveto(326.10410526,64.11890662)
\curveto(326.83430116,62.16341432)(325.82197135,60.03994152)(323.84300335,59.37599737)
\curveto(321.86403534,58.71205321)(319.66782406,59.75905985)(318.93762815,61.71455215)
\curveto(318.20743225,63.67004445)(319.21976206,65.79351725)(321.19873006,66.4574614)
\curveto(323.17769807,67.12140556)(325.37390935,66.07439892)(326.10410526,64.11890662)
\closepath
}
}
{
\newrgbcolor{curcolor}{0 0 0}
\pscustom[linestyle=none,fillstyle=solid,fillcolor=curcolor]
{
\newpath
\moveto(367.85028065,33.79506196)
\curveto(368.58047655,31.83956966)(367.56814675,29.71609686)(365.58917874,29.0521527)
\curveto(363.61021074,28.38820855)(361.41399945,29.43521519)(360.68380355,31.39070749)
\curveto(359.95360764,33.34619978)(360.96593745,35.46967259)(362.94490546,36.13361674)
\curveto(364.92387346,36.79756089)(367.12008475,35.75055425)(367.85028065,33.79506196)
\closepath
}
}
{
\newrgbcolor{curcolor}{0 0 0}
\pscustom[linewidth=1.72947494,linecolor=curcolor]
{
\newpath
\moveto(367.85028065,33.79506196)
\curveto(368.58047655,31.83956966)(367.56814675,29.71609686)(365.58917874,29.0521527)
\curveto(363.61021074,28.38820855)(361.41399945,29.43521519)(360.68380355,31.39070749)
\curveto(359.95360764,33.34619978)(360.96593745,35.46967259)(362.94490546,36.13361674)
\curveto(364.92387346,36.79756089)(367.12008475,35.75055425)(367.85028065,33.79506196)
\closepath
}
}
{
\newrgbcolor{curcolor}{0 0 0}
\pscustom[linestyle=none,fillstyle=solid,fillcolor=curcolor]
{
\newpath
\moveto(420.55815087,9.19215182)
\curveto(421.28834677,7.23665952)(420.27601696,5.11318672)(418.29704896,4.44924257)
\curveto(416.31808095,3.78529841)(414.12186967,4.83230505)(413.39167376,6.78779735)
\curveto(412.66147786,8.74328965)(413.67380767,10.86676245)(415.65277567,11.5307066)
\curveto(417.63174368,12.19465076)(419.82795496,11.14764412)(420.55815087,9.19215182)
\closepath
}
}
{
\newrgbcolor{curcolor}{0 0 0}
\pscustom[linewidth=1.72947494,linecolor=curcolor]
{
\newpath
\moveto(420.55815087,9.19215182)
\curveto(421.28834677,7.23665952)(420.27601696,5.11318672)(418.29704896,4.44924257)
\curveto(416.31808095,3.78529841)(414.12186967,4.83230505)(413.39167376,6.78779735)
\curveto(412.66147786,8.74328965)(413.67380767,10.86676245)(415.65277567,11.5307066)
\curveto(417.63174368,12.19465076)(419.82795496,11.14764412)(420.55815087,9.19215182)
\closepath
}
}
{
\newrgbcolor{curcolor}{0 0 0}
\pscustom[linestyle=none,fillstyle=solid,fillcolor=curcolor]
{
\newpath
\moveto(526.03158984,30.21421232)
\curveto(526.76178574,28.25872003)(525.74945593,26.13524722)(523.77048793,25.47130307)
\curveto(521.79151992,24.80735892)(519.59530864,25.85436556)(518.86511273,27.80985785)
\curveto(518.13491683,29.76535015)(519.14724664,31.88882295)(521.12621464,32.55276711)
\curveto(523.10518265,33.21671126)(525.30139393,32.16970462)(526.03158984,30.21421232)
\closepath
}
}
{
\newrgbcolor{curcolor}{0 0 0}
\pscustom[linewidth=1.72947494,linecolor=curcolor]
{
\newpath
\moveto(526.03158984,30.21421232)
\curveto(526.76178574,28.25872003)(525.74945593,26.13524722)(523.77048793,25.47130307)
\curveto(521.79151992,24.80735892)(519.59530864,25.85436556)(518.86511273,27.80985785)
\curveto(518.13491683,29.76535015)(519.14724664,31.88882295)(521.12621464,32.55276711)
\curveto(523.10518265,33.21671126)(525.30139393,32.16970462)(526.03158984,30.21421232)
\closepath
}
}
{
\newrgbcolor{curcolor}{0 0 0}
\pscustom[linestyle=none,fillstyle=solid,fillcolor=curcolor]
{
\newpath
\moveto(478.44235143,21.42853669)
\curveto(479.17254733,19.4730444)(478.16021752,17.34957159)(476.18124952,16.68562744)
\curveto(474.20228151,16.02168329)(472.00607023,17.06868993)(471.27587432,19.02418222)
\curveto(470.54567842,20.97967452)(471.55800823,23.10314732)(473.53697623,23.76709148)
\curveto(475.51594424,24.43103563)(477.71215552,23.38402899)(478.44235143,21.42853669)
\closepath
}
}
{
\newrgbcolor{curcolor}{0 0 0}
\pscustom[linewidth=1.72947494,linecolor=curcolor]
{
\newpath
\moveto(478.44235143,21.42853669)
\curveto(479.17254733,19.4730444)(478.16021752,17.34957159)(476.18124952,16.68562744)
\curveto(474.20228151,16.02168329)(472.00607023,17.06868993)(471.27587432,19.02418222)
\curveto(470.54567842,20.97967452)(471.55800823,23.10314732)(473.53697623,23.76709148)
\curveto(475.51594424,24.43103563)(477.71215552,23.38402899)(478.44235143,21.42853669)
\closepath
}
}
{
\newrgbcolor{curcolor}{0 0 0}
\pscustom[linestyle=none,fillstyle=solid,fillcolor=curcolor]
{
\newpath
\moveto(556.34953854,62.52885091)
\curveto(557.07973444,60.57335862)(556.06740464,58.44988581)(554.08843663,57.78594166)
\curveto(552.10946862,57.12199751)(549.91325734,58.16900415)(549.18306144,60.12449645)
\curveto(548.45286553,62.07998874)(549.46519534,64.20346155)(551.44416334,64.8674057)
\curveto(553.42313135,65.53134985)(555.61934263,64.48434321)(556.34953854,62.52885091)
\closepath
}
}
{
\newrgbcolor{curcolor}{0 0 0}
\pscustom[linewidth=1.72947494,linecolor=curcolor]
{
\newpath
\moveto(556.34953854,62.52885091)
\curveto(557.07973444,60.57335862)(556.06740464,58.44988581)(554.08843663,57.78594166)
\curveto(552.10946862,57.12199751)(549.91325734,58.16900415)(549.18306144,60.12449645)
\curveto(548.45286553,62.07998874)(549.46519534,64.20346155)(551.44416334,64.8674057)
\curveto(553.42313135,65.53134985)(555.61934263,64.48434321)(556.34953854,62.52885091)
\closepath
}
}
{
\newrgbcolor{curcolor}{0 0 0}
\pscustom[linestyle=none,fillstyle=solid,fillcolor=curcolor]
{
\newpath
\moveto(563.82957234,102.58803672)
\curveto(564.55976825,100.63254442)(563.54743844,98.50907162)(561.56847044,97.84512747)
\curveto(559.58950243,97.18118331)(557.39329115,98.22818995)(556.66309524,100.18368225)
\curveto(555.93289934,102.13917455)(556.94522914,104.26264735)(558.92419715,104.9265915)
\curveto(560.90316516,105.59053566)(563.09937644,104.54352902)(563.82957234,102.58803672)
\closepath
}
}
{
\newrgbcolor{curcolor}{0 0 0}
\pscustom[linewidth=1.72947494,linecolor=curcolor]
{
\newpath
\moveto(563.82957234,102.58803672)
\curveto(564.55976825,100.63254442)(563.54743844,98.50907162)(561.56847044,97.84512747)
\curveto(559.58950243,97.18118331)(557.39329115,98.22818995)(556.66309524,100.18368225)
\curveto(555.93289934,102.13917455)(556.94522914,104.26264735)(558.92419715,104.9265915)
\curveto(560.90316516,105.59053566)(563.09937644,104.54352902)(563.82957234,102.58803672)
\closepath
}
}
{
\rput[bl](600,130){$\F$}
\rput[bl](0,25){$\tilde{M}$}
}
\end{pspicture}
\caption{A discrete space~$\tilde{M}$ embedded in~$\F$.}
\label{figdiscrete}
\end{figure}
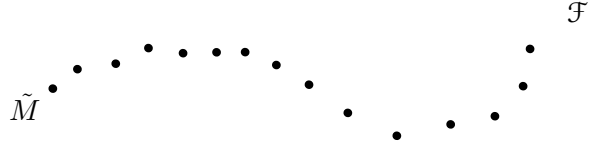%
The space~$\tilde{M}$ carries the topology induced by~$\F$.
Studying~$\tilde{M}$ intrinsically as a topological space is not of interest, because
the discrete topology is trivial (every subset of~$\tilde{M}$ is open).
Therefore, the interesting structures come from the embedding of~$\tilde{M}$
in~$\F$, as we now explain step-by-step.
First, on~$\F$ we are given the {\em{Lagrangian}} 
\[ \L \::\: \F \times \F \rightarrow \R^+_0 \:. \]
In the present paper, we always assume that~$\L$ is smooth in both arguments
(the {\em{smooth setting}}; for details see Section~\ref{seccvp}).
Similar to a convolution kernel, the Lagrangian can be used as a {\em{smoothing kernel}}
for mollifying functions on~$\tilde{M}$. For example, given a compactly supported
function~$f : \tilde{M} \rightarrow \R$, we obtain a smooth function~$\phi$ on~$\F$ by
\beq \label{intintro}
\phi(x) := \int_{\tilde{M}} \L(x,y)\: f(y)\: d\tilde{\rho}(y) \;\in\; C^\infty(\F, \R) \:.
\eeq
We note for clarity that, in our example of a discrete space, the integral
can be written as a sum over~$\tilde{M}$ with suitable weights~$c(y) \geq 0$,
\[ \phi(x) = \sum_{y \in \tilde{M}} \L(x,y)\: f(y)\: c(y) \:; \]
working with integrals as in~\eqref{intintro} has the advantage that our formulas
immediately extend to continuous measures
(where by ``continuous'' we mean non-atomic measures; see~\cite{johnson}).

Mollifying by taking the Lagrangian as a smoothing kernel is one of our basic tools.
It is very useful that the resulting smooth functions on~$\F$ can be differentiated.
However, the structures discussed so far are not yet sufficient for developing a
differential calculus on~$\tilde{M}$. 
One difficulty is that standard notions of differential calculus
(like tangent vectors to equivalence classes of curves or derivations acting on functions on~$\tilde{M}$)
cannot be used in an obvious way.
What is missing is a generalization of the concept
of a tangent space at a point~$p \in \tilde{M}$.
As a solution to this problem, we introduce the notion of an
{\em{osculating vacuum}}~$M_p$ at~$p \in \tilde{M}$. The idea is to ``attach'' to~$p$
a measure~$\rho_p$ which ``describes the vacuum'' and ``approximates~$\tilde{M}$
as well as possible''; see Figure~\ref{figosculate-discrete}.
\begin{figure}[tb]
\psset{xunit=.45pt,yunit=.45pt,runit=.45pt}
\begin{pspicture}(598.5941537,127.59481115)
{
\newrgbcolor{curcolor}{0 0 0}
\pscustom[linestyle=none,fillstyle=solid,fillcolor=curcolor]
{
\newpath
\moveto(50.59844906,59.89968657)
\curveto(51.32864497,57.94419428)(50.31631516,55.82072147)(48.33734716,55.15677732)
\curveto(46.35837915,54.49283317)(44.16216787,55.53983981)(43.43197196,57.4953321)
\curveto(42.70177606,59.4508244)(43.71410587,61.5742972)(45.69307387,62.23824136)
\curveto(47.67204188,62.90218551)(49.86825316,61.85517887)(50.59844906,59.89968657)
\closepath
}
}
{
\newrgbcolor{curcolor}{0 0 0}
\pscustom[linewidth=1.72947494,linecolor=curcolor]
{
\newpath
\moveto(50.59844906,59.89968657)
\curveto(51.32864497,57.94419428)(50.31631516,55.82072147)(48.33734716,55.15677732)
\curveto(46.35837915,54.49283317)(44.16216787,55.53983981)(43.43197196,57.4953321)
\curveto(42.70177606,59.4508244)(43.71410587,61.5742972)(45.69307387,62.23824136)
\curveto(47.67204188,62.90218551)(49.86825316,61.85517887)(50.59844906,59.89968657)
\closepath
}
}
{
\newrgbcolor{curcolor}{0 0 0}
\pscustom[linestyle=none,fillstyle=solid,fillcolor=curcolor]
{
\newpath
\moveto(77.04301771,80.81195297)
\curveto(77.77321362,78.85646067)(76.76088381,76.73298787)(74.7819158,76.06904371)
\curveto(72.8029478,75.40509956)(70.60673651,76.4521062)(69.87654061,78.4075985)
\curveto(69.1463447,80.36309079)(70.15867451,82.4865636)(72.13764252,83.15050775)
\curveto(74.11661052,83.8144519)(76.31282181,82.76744526)(77.04301771,80.81195297)
\closepath
}
}
{
\newrgbcolor{curcolor}{0 0 0}
\pscustom[linewidth=1.72947494,linecolor=curcolor]
{
\newpath
\moveto(77.04301771,80.81195297)
\curveto(77.77321362,78.85646067)(76.76088381,76.73298787)(74.7819158,76.06904371)
\curveto(72.8029478,75.40509956)(70.60673651,76.4521062)(69.87654061,78.4075985)
\curveto(69.1463447,80.36309079)(70.15867451,82.4865636)(72.13764252,83.15050775)
\curveto(74.11661052,83.8144519)(76.31282181,82.76744526)(77.04301771,80.81195297)
\closepath
}
}
{
\newrgbcolor{curcolor}{0 0 0}
\pscustom[linestyle=none,fillstyle=solid,fillcolor=curcolor]
{
\newpath
\moveto(118.2726809,86.70813173)
\curveto(119.00287681,84.75263943)(117.990547,82.62916663)(116.011579,81.96522247)
\curveto(114.03261099,81.30127832)(111.83639971,82.34828496)(111.1062038,84.30377726)
\curveto(110.3760079,86.25926955)(111.3883377,88.38274236)(113.36730571,89.04668651)
\curveto(115.34627372,89.71063066)(117.542485,88.66362402)(118.2726809,86.70813173)
\closepath
}
}
{
\newrgbcolor{curcolor}{0 0 0}
\pscustom[linewidth=1.72947494,linecolor=curcolor]
{
\newpath
\moveto(118.2726809,86.70813173)
\curveto(119.00287681,84.75263943)(117.990547,82.62916663)(116.011579,81.96522247)
\curveto(114.03261099,81.30127832)(111.83639971,82.34828496)(111.1062038,84.30377726)
\curveto(110.3760079,86.25926955)(111.3883377,88.38274236)(113.36730571,89.04668651)
\curveto(115.34627372,89.71063066)(117.542485,88.66362402)(118.2726809,86.70813173)
\closepath
}
}
{
\newrgbcolor{curcolor}{0 0 0}
\pscustom[linestyle=none,fillstyle=solid,fillcolor=curcolor]
{
\newpath
\moveto(153.43325555,103.49030862)
\curveto(154.16345145,101.53481632)(153.15112165,99.41134352)(151.17215364,98.74739936)
\curveto(149.19318563,98.08345521)(146.99697435,99.13046185)(146.26677845,101.08595415)
\curveto(145.53658254,103.04144644)(146.54891235,105.16491925)(148.52788036,105.8288634)
\curveto(150.50684836,106.49280755)(152.70305964,105.44580091)(153.43325555,103.49030862)
\closepath
}
}
{
\newrgbcolor{curcolor}{0 0 0}
\pscustom[linewidth=1.72947494,linecolor=curcolor]
{
\newpath
\moveto(153.43325555,103.49030862)
\curveto(154.16345145,101.53481632)(153.15112165,99.41134352)(151.17215364,98.74739936)
\curveto(149.19318563,98.08345521)(146.99697435,99.13046185)(146.26677845,101.08595415)
\curveto(145.53658254,103.04144644)(146.54891235,105.16491925)(148.52788036,105.8288634)
\curveto(150.50684836,106.49280755)(152.70305964,105.44580091)(153.43325555,103.49030862)
\closepath
}
}
{
\newrgbcolor{curcolor}{0 0 0}
\pscustom[linestyle=none,fillstyle=solid,fillcolor=curcolor]
{
\newpath
\moveto(190.64193345,98.00163609)
\curveto(191.37212936,96.04614379)(190.35979955,93.92267099)(188.38083155,93.25872684)
\curveto(186.40186354,92.59478268)(184.20565226,93.64178932)(183.47545635,95.59728162)
\curveto(182.74526045,97.55277392)(183.75759025,99.67624672)(185.73655826,100.34019088)
\curveto(187.71552627,101.00413503)(189.91173755,99.95712839)(190.64193345,98.00163609)
\closepath
}
}
{
\newrgbcolor{curcolor}{0 0 0}
\pscustom[linewidth=1.72947494,linecolor=curcolor]
{
\newpath
\moveto(190.64193345,98.00163609)
\curveto(191.37212936,96.04614379)(190.35979955,93.92267099)(188.38083155,93.25872684)
\curveto(186.40186354,92.59478268)(184.20565226,93.64178932)(183.47545635,95.59728162)
\curveto(182.74526045,97.55277392)(183.75759025,99.67624672)(185.73655826,100.34019088)
\curveto(187.71552627,101.00413503)(189.91173755,99.95712839)(190.64193345,98.00163609)
\closepath
}
}
{
\newrgbcolor{curcolor}{0 0 0}
\pscustom[linestyle=none,fillstyle=solid,fillcolor=curcolor]
{
\newpath
\moveto(257.41376743,99.14321022)
\curveto(258.14396334,97.18771793)(257.13163353,95.06424512)(255.15266552,94.40030097)
\curveto(253.17369752,93.73635682)(250.97748623,94.78336346)(250.24729033,96.73885575)
\curveto(249.51709442,98.69434805)(250.52942423,100.81782085)(252.50839224,101.48176501)
\curveto(254.48736024,102.14570916)(256.68357153,101.09870252)(257.41376743,99.14321022)
\closepath
}
}
{
\newrgbcolor{curcolor}{0 0 0}
\pscustom[linewidth=1.72947494,linecolor=curcolor]
{
\newpath
\moveto(257.41376743,99.14321022)
\curveto(258.14396334,97.18771793)(257.13163353,95.06424512)(255.15266552,94.40030097)
\curveto(253.17369752,93.73635682)(250.97748623,94.78336346)(250.24729033,96.73885575)
\curveto(249.51709442,98.69434805)(250.52942423,100.81782085)(252.50839224,101.48176501)
\curveto(254.48736024,102.14570916)(256.68357153,101.09870252)(257.41376743,99.14321022)
\closepath
}
}
{
\newrgbcolor{curcolor}{0 0 0}
\pscustom[linestyle=none,fillstyle=solid,fillcolor=curcolor]
{
\newpath
\moveto(226.6732125,98.95567244)
\curveto(227.40340841,97.00018015)(226.3910786,94.87670734)(224.41211059,94.21276319)
\curveto(222.43314259,93.54881904)(220.2369313,94.59582568)(219.5067354,96.55131797)
\curveto(218.77653949,98.50681027)(219.7888693,100.63028307)(221.76783731,101.29422723)
\curveto(223.74680531,101.95817138)(225.9430166,100.91116474)(226.6732125,98.95567244)
\closepath
}
}
{
\newrgbcolor{curcolor}{0 0 0}
\pscustom[linewidth=1.72947494,linecolor=curcolor]
{
\newpath
\moveto(226.6732125,98.95567244)
\curveto(227.40340841,97.00018015)(226.3910786,94.87670734)(224.41211059,94.21276319)
\curveto(222.43314259,93.54881904)(220.2369313,94.59582568)(219.5067354,96.55131797)
\curveto(218.77653949,98.50681027)(219.7888693,100.63028307)(221.76783731,101.29422723)
\curveto(223.74680531,101.95817138)(225.9430166,100.91116474)(226.6732125,98.95567244)
\closepath
}
}
{
\newrgbcolor{curcolor}{0 0 0}
\pscustom[linestyle=none,fillstyle=solid,fillcolor=curcolor]
{
\newpath
\moveto(290.99856429,85.26668274)
\curveto(291.72876019,83.31119044)(290.71643038,81.18771764)(288.73746238,80.52377348)
\curveto(286.75849437,79.85982933)(284.56228309,80.90683597)(283.83208718,82.86232827)
\curveto(283.10189128,84.81782056)(284.11422109,86.94129337)(286.09318909,87.60523752)
\curveto(288.0721571,88.26918167)(290.26836838,87.22217503)(290.99856429,85.26668274)
\closepath
}
}
{
\newrgbcolor{curcolor}{0 0 0}
\pscustom[linewidth=1.72947494,linecolor=curcolor]
{
\newpath
\moveto(290.99856429,85.26668274)
\curveto(291.72876019,83.31119044)(290.71643038,81.18771764)(288.73746238,80.52377348)
\curveto(286.75849437,79.85982933)(284.56228309,80.90683597)(283.83208718,82.86232827)
\curveto(283.10189128,84.81782056)(284.11422109,86.94129337)(286.09318909,87.60523752)
\curveto(288.0721571,88.26918167)(290.26836838,87.22217503)(290.99856429,85.26668274)
\closepath
}
}
{
\newrgbcolor{curcolor}{0 0 0}
\pscustom[linestyle=none,fillstyle=solid,fillcolor=curcolor]
{
\newpath
\moveto(326.10410526,64.11890662)
\curveto(326.83430116,62.16341432)(325.82197135,60.03994152)(323.84300335,59.37599737)
\curveto(321.86403534,58.71205321)(319.66782406,59.75905985)(318.93762815,61.71455215)
\curveto(318.20743225,63.67004445)(319.21976206,65.79351725)(321.19873006,66.4574614)
\curveto(323.17769807,67.12140556)(325.37390935,66.07439892)(326.10410526,64.11890662)
\closepath
}
}
{
\newrgbcolor{curcolor}{0 0 0}
\pscustom[linewidth=1.72947494,linecolor=curcolor]
{
\newpath
\moveto(326.10410526,64.11890662)
\curveto(326.83430116,62.16341432)(325.82197135,60.03994152)(323.84300335,59.37599737)
\curveto(321.86403534,58.71205321)(319.66782406,59.75905985)(318.93762815,61.71455215)
\curveto(318.20743225,63.67004445)(319.21976206,65.79351725)(321.19873006,66.4574614)
\curveto(323.17769807,67.12140556)(325.37390935,66.07439892)(326.10410526,64.11890662)
\closepath
}
}
{
\newrgbcolor{curcolor}{0 0 0}
\pscustom[linestyle=none,fillstyle=solid,fillcolor=curcolor]
{
\newpath
\moveto(367.85028065,33.79506196)
\curveto(368.58047655,31.83956966)(367.56814675,29.71609686)(365.58917874,29.0521527)
\curveto(363.61021074,28.38820855)(361.41399945,29.43521519)(360.68380355,31.39070749)
\curveto(359.95360764,33.34619978)(360.96593745,35.46967259)(362.94490546,36.13361674)
\curveto(364.92387346,36.79756089)(367.12008475,35.75055425)(367.85028065,33.79506196)
\closepath
}
}
{
\newrgbcolor{curcolor}{0 0 0}
\pscustom[linewidth=1.72947494,linecolor=curcolor]
{
\newpath
\moveto(367.85028065,33.79506196)
\curveto(368.58047655,31.83956966)(367.56814675,29.71609686)(365.58917874,29.0521527)
\curveto(363.61021074,28.38820855)(361.41399945,29.43521519)(360.68380355,31.39070749)
\curveto(359.95360764,33.34619978)(360.96593745,35.46967259)(362.94490546,36.13361674)
\curveto(364.92387346,36.79756089)(367.12008475,35.75055425)(367.85028065,33.79506196)
\closepath
}
}
{
\newrgbcolor{curcolor}{0 0 0}
\pscustom[linestyle=none,fillstyle=solid,fillcolor=curcolor]
{
\newpath
\moveto(420.55815087,9.19215182)
\curveto(421.28834677,7.23665952)(420.27601696,5.11318672)(418.29704896,4.44924257)
\curveto(416.31808095,3.78529841)(414.12186967,4.83230505)(413.39167376,6.78779735)
\curveto(412.66147786,8.74328965)(413.67380767,10.86676245)(415.65277567,11.5307066)
\curveto(417.63174368,12.19465076)(419.82795496,11.14764412)(420.55815087,9.19215182)
\closepath
}
}
{
\newrgbcolor{curcolor}{0 0 0}
\pscustom[linewidth=1.72947494,linecolor=curcolor]
{
\newpath
\moveto(420.55815087,9.19215182)
\curveto(421.28834677,7.23665952)(420.27601696,5.11318672)(418.29704896,4.44924257)
\curveto(416.31808095,3.78529841)(414.12186967,4.83230505)(413.39167376,6.78779735)
\curveto(412.66147786,8.74328965)(413.67380767,10.86676245)(415.65277567,11.5307066)
\curveto(417.63174368,12.19465076)(419.82795496,11.14764412)(420.55815087,9.19215182)
\closepath
}
}
{
\newrgbcolor{curcolor}{0 0 0}
\pscustom[linestyle=none,fillstyle=solid,fillcolor=curcolor]
{
\newpath
\moveto(526.03158984,30.21421232)
\curveto(526.76178574,28.25872003)(525.74945593,26.13524722)(523.77048793,25.47130307)
\curveto(521.79151992,24.80735892)(519.59530864,25.85436556)(518.86511273,27.80985785)
\curveto(518.13491683,29.76535015)(519.14724664,31.88882295)(521.12621464,32.55276711)
\curveto(523.10518265,33.21671126)(525.30139393,32.16970462)(526.03158984,30.21421232)
\closepath
}
}
{
\newrgbcolor{curcolor}{0 0 0}
\pscustom[linewidth=1.72947494,linecolor=curcolor]
{
\newpath
\moveto(526.03158984,30.21421232)
\curveto(526.76178574,28.25872003)(525.74945593,26.13524722)(523.77048793,25.47130307)
\curveto(521.79151992,24.80735892)(519.59530864,25.85436556)(518.86511273,27.80985785)
\curveto(518.13491683,29.76535015)(519.14724664,31.88882295)(521.12621464,32.55276711)
\curveto(523.10518265,33.21671126)(525.30139393,32.16970462)(526.03158984,30.21421232)
\closepath
}
}
{
\newrgbcolor{curcolor}{0 0 0}
\pscustom[linestyle=none,fillstyle=solid,fillcolor=curcolor]
{
\newpath
\moveto(478.44235143,21.42853669)
\curveto(479.17254733,19.4730444)(478.16021752,17.34957159)(476.18124952,16.68562744)
\curveto(474.20228151,16.02168329)(472.00607023,17.06868993)(471.27587432,19.02418222)
\curveto(470.54567842,20.97967452)(471.55800823,23.10314732)(473.53697623,23.76709148)
\curveto(475.51594424,24.43103563)(477.71215552,23.38402899)(478.44235143,21.42853669)
\closepath
}
}
{
\newrgbcolor{curcolor}{0 0 0}
\pscustom[linewidth=1.72947494,linecolor=curcolor]
{
\newpath
\moveto(478.44235143,21.42853669)
\curveto(479.17254733,19.4730444)(478.16021752,17.34957159)(476.18124952,16.68562744)
\curveto(474.20228151,16.02168329)(472.00607023,17.06868993)(471.27587432,19.02418222)
\curveto(470.54567842,20.97967452)(471.55800823,23.10314732)(473.53697623,23.76709148)
\curveto(475.51594424,24.43103563)(477.71215552,23.38402899)(478.44235143,21.42853669)
\closepath
}
}
{
\newrgbcolor{curcolor}{0 0 0}
\pscustom[linestyle=none,fillstyle=solid,fillcolor=curcolor]
{
\newpath
\moveto(556.34953854,62.52885091)
\curveto(557.07973444,60.57335862)(556.06740464,58.44988581)(554.08843663,57.78594166)
\curveto(552.10946862,57.12199751)(549.91325734,58.16900415)(549.18306144,60.12449645)
\curveto(548.45286553,62.07998874)(549.46519534,64.20346155)(551.44416334,64.8674057)
\curveto(553.42313135,65.53134985)(555.61934263,64.48434321)(556.34953854,62.52885091)
\closepath
}
}
{
\newrgbcolor{curcolor}{0 0 0}
\pscustom[linewidth=1.72947494,linecolor=curcolor]
{
\newpath
\moveto(556.34953854,62.52885091)
\curveto(557.07973444,60.57335862)(556.06740464,58.44988581)(554.08843663,57.78594166)
\curveto(552.10946862,57.12199751)(549.91325734,58.16900415)(549.18306144,60.12449645)
\curveto(548.45286553,62.07998874)(549.46519534,64.20346155)(551.44416334,64.8674057)
\curveto(553.42313135,65.53134985)(555.61934263,64.48434321)(556.34953854,62.52885091)
\closepath
}
}
{
\newrgbcolor{curcolor}{0 0 0}
\pscustom[linestyle=none,fillstyle=solid,fillcolor=curcolor]
{
\newpath
\moveto(563.82957234,102.58803672)
\curveto(564.55976825,100.63254442)(563.54743844,98.50907162)(561.56847044,97.84512747)
\curveto(559.58950243,97.18118331)(557.39329115,98.22818995)(556.66309524,100.18368225)
\curveto(555.93289934,102.13917455)(556.94522914,104.26264735)(558.92419715,104.9265915)
\curveto(560.90316516,105.59053566)(563.09937644,104.54352902)(563.82957234,102.58803672)
\closepath
}
}
{
\newrgbcolor{curcolor}{0 0 0}
\pscustom[linewidth=1.72947494,linecolor=curcolor]
{
\newpath
\moveto(563.82957234,102.58803672)
\curveto(564.55976825,100.63254442)(563.54743844,98.50907162)(561.56847044,97.84512747)
\curveto(559.58950243,97.18118331)(557.39329115,98.22818995)(556.66309524,100.18368225)
\curveto(555.93289934,102.13917455)(556.94522914,104.26264735)(558.92419715,104.9265915)
\curveto(560.90316516,105.59053566)(563.09937644,104.54352902)(563.82957234,102.58803672)
\closepath
}
}
{
\newrgbcolor{curcolor}{0 0 0}
\pscustom[linewidth=3.02362209,linecolor=curcolor]
{
\newpath
\moveto(0.17175307,85.3538776)
\lineto(356.44963654,126.09279288)
}
}
{
\newrgbcolor{curcolor}{0 0 0}
\pscustom[linewidth=3.02362209,linecolor=curcolor]
{
\newpath
\moveto(250.12209638,13.82094784)
\lineto(598.54081134,1.5108754)
}
\rput[bl](600,130){$\F$}
\rput[bl](5,25){$\tilde{M}$}
\rput[bl](145,77){$p$}
\rput[bl](412,19){$q$}
\rput[bl](365,112){$M_p$}
\rput[bl](210,5){$M_q$}
}
\end{pspicture}
\caption{Osculating vacua of a discrete space.}
\label{figosculate-discrete}
\end{figure}
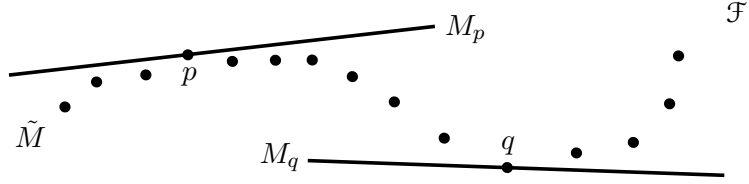%
The words in quotation marks in the last sentence have not yet been defined;
it will be the main task of Section~\ref{secosculate} to give these words a precise 
mathematical meaning.
For the purpose of the introduction, it suffices to note that~$M_p$ is
a smooth submanifold of~$\F$, which also carries a vector space structure
(corresponding to the usual vector space structure of the vacuum space or
spacetime).
Moreover, $M_p$ is the support of a measure~$\rho_p$ which is absolutely continuous with respect to the Lebesgue measure on~$M_p$.

Having chosen an osculating vacuum at each point~$p \in \tilde{M}$,
a differential calculus can be introduced by pursuing the following strategy.
Interpreting the vectors in~$M_p$ as ``tangential directions'', we can
form corresponding directional derivatives as weak derivatives,
using the Lagrangian for testing. To explain the idea, let us consider
directional derivatives of functions on~$\tilde{M}$. To this end, we define a vector field~$v$ as a
mapping which to every point~$y \in \tilde M$ associates a vector~$v(y) \in M_y$
in the corresponding osculating vacuum. Since the Lagrangian is smooth,
we can form the directional derivative~$D_{2,v} \L(x,y)$ (where the index two
means that the derivative acts on the second argument of the Lagrangian).
Considering for simplicity again a compactly supported
function~$f : \tilde{M} \rightarrow \R$, we interpret the expression
\beq \label{weakderivative}
-\int_{M_p} D_{2,v} \L(x,y)\: f(y)\: d\rho_p(y)
\eeq
as the weak derivative of~$f$ in the direction~$v$.
In this way, one can generalize certain aspects of the standard differential
calculus on manifolds to the non-smooth setting.
As we shall explain step by step in Sections~\ref{secdirder}--\ref{secLcalc},
our setup also allows for introducing differential forms and the exterior calculus,
and gives rise to suitable generalizations of cohomological constructions
(de Rham cohomology, Poincar{\'e} Lemma, Mayer-Vietoris sequence and
Künneth formula).

In order to generalize Stokes' theorem to our setting,
we need two more ingredients:
First, one needs to make use of the {\em{Euler-Lagrange equations}} (EL equations) which hold
for critical measures. This means that Stokes' theorem does not hold for
any measure~$\tilde{\rho}$, but it applies only if~$\tilde{\rho}$ is a critical point
of the causal variational principle. The necessary analytic setup will be introduced
in Section~\ref{seccvp} of the preliminaries.
The second ingredient is that of {\em{surface layer integrals}}.
In formulating Stokes' theorem one needs the notion of a boundary integral
or, more generally, integrals over subsets of a given co-dimension.
As becomes clear already in the example of a discrete space,
an integral over a hypersurface cannot be introduced in an obvious way.
The way out is to work instead of integrals over hypersurfaces
with double integrals over suitably chosen subsets of~$\tilde{M}$.
This will be explained further in Section~\ref{secOSI} of the preliminaries.
Based on the EL equations and working with surface layer integrals,
we succeed in proving versions of the {\em{Gau{\ss} divergence theorem}}
(Theorem~\ref{thmgauss}) and of {\em{Stokes' theorem}} (Theorem~\ref{thmstokes}
for top forms and Theorem~\ref{thmstokesgen} for integrals of higher co-dimension).

We point out that the exterior calculus on non-smooth 
spaces that is developed in this paper is conceptually different from the ones found in other approaches existing in the literature (see for example~\cite{hydon-mansfield, white-hydon, arnold-hu, arnold-dn, ferri-gagliardi, casali, gurau-ryan, desbrun})
in that it is anchored around causal variational principles, their critical points, osculating vacua and mollification by the Lagrangian taken as a smoothing kernel.

The paper is organized as follows. Section~\ref{secprelim} makes the paper
self-contained by providing the necessary background on causal variational principles.
In Section~\ref{secosculate} we introduce and discuss the notion of 
osculating vacua.
In Section~\ref{secdirder} vector fields and directional derivatives are defined.
We proceed in Section~\ref{secgauss} with a
generalization of the Gau{\ss} divergence theorem.
 In Section~\ref{secconnection} it is explained how the Lagrangian induces
specific coordinate systems which can be used to describe~$\tilde{M}$ locally.
We also construct a canonical mapping~$\nabla^\L_{q,p} : M_p \rightarrow M_q$
which can be interpreted as a parallel transport or connection between the points.
In Section~\ref{secLcalc} it is explained step-by-step how the idea of weak directional
derivatives~\eqref{weakderivative} can be used to develop an exterior
differential calculus leading to versions for non-smooth spaces of the de Rham cohomology, the Mayer-Vietoris sequence, the Künneth formula, the Stokes theorem and the Poincar{\'e} lemma. Section~\ref{secsoft} is devoted to surface layer integrals
of higher co-dimension and a corresponding Stokes theorem.
In Section~\ref{sectensor} a general tensor calculus associated to the
osculating vacua is introduced. In Section~\ref{secex},
the abstract constructions and results are illustrated by various examples:
Spaces consisting of one point (Section~\ref{sec:onepoint}) and two points (Section~\ref{sec:Two-points}),
the discrete line~$\Z$ (Section~\ref{sec:discline}), the discrete circle~$\Z / 4 \Z$
(Section~\ref{sec:disc-circ}), a $k$-dimensional lattice (Section~\ref{seclattice})
and an $\L$-star-shaped subset thereof (Section~\ref{sec-lattice-subset}).

We conclude the introduction with two remarks. We first point out that,
throughout the paper, we do {\em{not}} assume a metric
on the osculating vacua. Correspondingly, we cannot introduce a Hodge star
or a Hodge Laplacian. Also, we do not discuss the properties of the
connection~$\nabla^\L$ and the corresponding notions of
curvature. All the geometric aspects will be worked out separately
(see~\cite{gauss, nonsmooth}).
Here we shall focus on developing the $\L$-calculus as an exterior calculus
in connection with differential topology.

We finally remark that, in order to keep the setting as simple as possible,
we do not endow the spaces of differential forms with a topology or a norm.
Also, in the main part of the paper we do not
impose that the differential forms should have compact support or 
satisfy decay conditions at infinity.
As we shall see in the examples, doing so may well affect the cohomology
and properties of the differential calculus (see in particular the lattice examples
in Sections~\ref{sec:discline} and~\ref{seclattice}).
One reason why we do not enter the details is that the resulting extensions of our
methods and results should be tailored to the specific applications in mind.

\section{Preliminaries} \label{secprelim}
\subsection{Causal Variational Principles} \label{seccvp}
We consider causal variational principles in the smooth and non-compact setting
(see for example~\cite[Chapter~6]{intro}). Thus we let~$\F$ be a smooth
(possibly non-compact) manifold. The {\em{Lagrangian}}~$\L$ is a given smooth function
\beq \label{Lsmooth}
\L \in C^\infty(\F \times \F, \R^+_0) \:.
\eeq
Moreover, we assume that~$\L$ has the following properties:
\bitem
\item[{\rm{(i)}}] $\L$ is symmetric: $\L(x,y) = \L(y,x)$ for all~$x,y \in \F$.\label{Cond1}
\item[{\rm{(ii)}}] $\L$ is strictly positive on the diagonal: $\L(x,x)>0$ for all~$x \in \F$. \label{Cond2}
\eitem
Finally, we need to assume that the Lagrangian decays sufficiently fast
if its arguments~$x$ and~$y$ are far apart. One way of doing so is to
use the following notion first introduced in~\cite[Definition~3.3]{noncompact}
(see also~\cite[Definition~8.1.1]{intro}):
\bitem
\item[{\rm{(iii)}}] $\L$ has {\em{compact range}}: For every compact set~$K \subset \F$ there is a compact set~$K' \subset \F$ such that
\[ \L(x,y) = 0 \qquad \text{for all~$x \in K$ and~$y \not \in K'$}\:. \]
\eitem
This condition could be relaxed by demanding that~$\L$ and all its derivatives
have rapid decay. We will implicitly use this weaker assumption in some of
the examples.

The {\em{causal variational principle}} is to minimize the causal action defined by
\[ 
\Sact = \int_\F d\rho(x) \int_\F d\rho(y)\: \L(x,y) \]
under variations of the measure~$\rho$ in the class of all regular Borel measures on~$\F$, keeping the total volume fixed.
The existence of minimizers has been established in~\cite{noncompact}
(see also~\cite[Chapter~12]{intro}). It is also shown that minimizing measures
are {\em{locally finite}} in the sense that~$\rho(K)<\infty$ for any compact~$K \subset \F$.
A minimizing measure satisfies the {\em{Euler-Lagrange (EL) equations}}, which state that the function~$\ell$ defined by
\beq \label{elldef}
\ell(x) = \int_\F \L(x,y)\: d\rho(y) - \s \::\: \F \rightarrow \R
\eeq
satisfies for a suitable parameter~$\s>0$ the equation
\beq \label{EL1}
\ell|_M \equiv \inf_\F \ell = 0 \:.
\eeq
The derivation can be found in~\cite[Section~4]{noncompact} or~\cite[Chapter~7]{intro}.

We remark that causal variational principles can be regarded as a
generalization of the causal action principle, being at the heart of the physical
theory of causal fermion systems. For the purposes of the present paper, we do not need
to enter the details of the connection to physics or to spacetime geometry.
We refer the reader interested in the physical background to the text book~\cite{intro};
in particular, the connection between the causal action principle and causal variational
principles is explained in detail in~\cite[Chapter~6]{intro}).
For what follows, it suffices to note that the support of the measure~$\rho$ denoted by
\[ M := \supp \rho \]
is considered as the underlying {\em{space}} or {\em{spacetime}} (we will use the terms ``space'' and ``spacetime'' interchangeably). We consider on~$M$ the topology induced by~$\F$. We point out that we do {\em{not}} assume~$M$ to be a smooth
submanifold of~$\F$; it merely is a closed topological subspace.
Thus~$M$ corresponds to the ``non-smooth space'' mentioned in the title.
But it is important to keep in mind that~$M$ is always embedded in the smooth ambient manifold~$\F$.

We close this section with a few clarifying remarks.
We first comment on how the manifold~$\F$ arises and can be understood.
In the physical setting of causal fermion systems, the space~$\F$ is
introduced abstractly as a manifold formed of linear operators on a Hilbert space.
In applications, one often begins with the space~$M$ (being the 
physical spacetime or the topological space of interest).
Then~$\F$ can be constructed by
introducing spaces of functions on~$M$ or sections of vector bundles on~$M$
and considering suitable linear operators thereon. This is explained in detail 
and illustrated by examples in~\cite{topology}.

We next point out that our assumption that
the Lagrangian is smooth~\eqref{Lsmooth} is a mathematical idealization
which does not quite hold in the physical applications. Indeed, as is 
worked out in~\cite[Section~5.1]{banach}, the causal Lagrangian merely is
locally H\"older continuous.
In order to treat this more general case, one could combine our constructions with the so-called {\em{expedient differential calculus}} as developed in~\cite[Section~4]{banach}.
The general idea is to make use of the fact that, although the Lagrangian is
only H\"older continuous, derivatives of integrals of the Lagrangian
still exist in certain directions, making it possible to develop a differential calculus
by carefully keeping track of these differentiable directions.
As a simple alternative, one could mollify the
causal Lagrangian to obtain a smooth Lagrangian~\eqref{Lsmooth}.
However, when doing so, one should be careful because the Gau{\ss} and Stokes theorems
may no longer hold. This is because these theorems rely on the EL equations,
and it is not clear in general how to mollify the Lagrangian while
preserving the EL equations.

We finally note that all the constructions and results in this paper are
{\em{functorial}}, if for a morphism between causal variational principles we take
a smooth and closed mapping from~$\F$ to~$\F'$ which preserves the Lagrangian.
A measure~$\rho$ on~$\F$ is then mapped to a measure~$\rho'$ on~$\F'$
using the push-forward.

\subsection{Surface Layer Integrals} \label{secOSI}
In the classical theorems of Gau{\ss} and Stokes, one integrates the fluxes
or differential forms over the boundary of a subregion~$\Omega$ of a manifold~$M$.
Such boundary integrals cannot be introduced in the setting of causal variational
principles. Instead, using a concept first introduced in~\cite{noether},
one works with so-called {\em{surface layer integrals}}, being
defined as double integrals of the general form
\beq \label{intdouble}
\int_\Omega \bigg( \int_{M \setminus \Omega} (\cdots)\: \L(x,y)\: d \rho(y) \bigg)\, d \rho(x) \:,
\eeq
where~$(\cdots)$ stands for a certain differential operator acting on the Lagrangian.
The main point is that one variable is integrated over a subset~$\Omega \subset M$, 
whereas the other variable is integrated over the complement of~$\Omega$
(see Figure~\ref{figosi}).
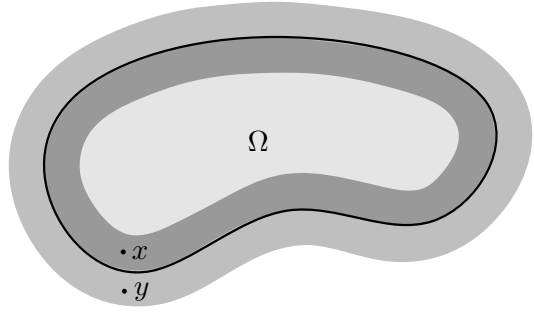
\begin{figure}
\psset{xunit=.45pt,yunit=.45pt,runit=.45pt}
\begin{pspicture}(437.88862574,254.86055149)
{
\newrgbcolor{curcolor}{0.89803922 0.89803922 0.89803922}
\pscustom[linestyle=none,fillstyle=solid,fillcolor=curcolor]
{
\newpath
\moveto(142.47850961,217.19929039)
\lineto(179.1897411,223.43327338)
\lineto(224.90561008,226.20393165)
\lineto(260.23150866,224.47226928)
\lineto(289.66976882,220.66261511)
\lineto(326.03464441,212.69696976)
\lineto(348.19991055,205.77032031)
\lineto(371.75052472,195.03402047)
\lineto(390.45245858,182.21972125)
\lineto(402.57408378,164.21044251)
\lineto(406.73009008,152.43513921)
\lineto(406.73009008,135.8111896)
\lineto(403.6131137,121.95790204)
\lineto(394.95478299,104.98761165)
\lineto(384.56482394,91.48065354)
\lineto(372.09684283,81.09068314)
\lineto(358.58986961,73.47137102)
\lineto(340.58060976,68.62272094)
\lineto(314.60569323,69.31538361)
\lineto(292.44042709,74.16403369)
\lineto(262.65582992,80.74435369)
\lineto(235.98823937,81.0906718)
\lineto(214.51563969,75.54935149)
\lineto(193.04302866,67.23738047)
\lineto(172.95576189,55.46208094)
\lineto(152.1758211,44.03311086)
\lineto(134.51287181,35.0284696)
\lineto(119.96691024,30.52615275)
\lineto(106.11361134,28.79449039)
\lineto(92.60665323,30.52615275)
\lineto(81.52402016,34.68214015)
\lineto(66.97806236,43.34044062)
\lineto(54.16376315,54.42308125)
\lineto(40.65680126,74.85668881)
\lineto(31.99849323,95.98295905)
\lineto(28.53517228,119.53356188)
\lineto(32.69116346,144.81581952)
\lineto(43.77379276,165.24943086)
\lineto(57.28075087,179.44905921)
\lineto(77.36803276,192.95601732)
\lineto(100.57230236,204.73132062)
\lineto(122.73757228,213.04328787)
\closepath
}
}
{
\newrgbcolor{curcolor}{0.89803922 0.89803922 0.89803922}
\pscustom[linewidth=0.99999871,linecolor=curcolor]
{
\newpath
\moveto(142.47850961,217.19929039)
\lineto(179.1897411,223.43327338)
\lineto(224.90561008,226.20393165)
\lineto(260.23150866,224.47226928)
\lineto(289.66976882,220.66261511)
\lineto(326.03464441,212.69696976)
\lineto(348.19991055,205.77032031)
\lineto(371.75052472,195.03402047)
\lineto(390.45245858,182.21972125)
\lineto(402.57408378,164.21044251)
\lineto(406.73009008,152.43513921)
\lineto(406.73009008,135.8111896)
\lineto(403.6131137,121.95790204)
\lineto(394.95478299,104.98761165)
\lineto(384.56482394,91.48065354)
\lineto(372.09684283,81.09068314)
\lineto(358.58986961,73.47137102)
\lineto(340.58060976,68.62272094)
\lineto(314.60569323,69.31538361)
\lineto(292.44042709,74.16403369)
\lineto(262.65582992,80.74435369)
\lineto(235.98823937,81.0906718)
\lineto(214.51563969,75.54935149)
\lineto(193.04302866,67.23738047)
\lineto(172.95576189,55.46208094)
\lineto(152.1758211,44.03311086)
\lineto(134.51287181,35.0284696)
\lineto(119.96691024,30.52615275)
\lineto(106.11361134,28.79449039)
\lineto(92.60665323,30.52615275)
\lineto(81.52402016,34.68214015)
\lineto(66.97806236,43.34044062)
\lineto(54.16376315,54.42308125)
\lineto(40.65680126,74.85668881)
\lineto(31.99849323,95.98295905)
\lineto(28.53517228,119.53356188)
\lineto(32.69116346,144.81581952)
\lineto(43.77379276,165.24943086)
\lineto(57.28075087,179.44905921)
\lineto(77.36803276,192.95601732)
\lineto(100.57230236,204.73132062)
\lineto(122.73757228,213.04328787)
\closepath
}
}
{
\newrgbcolor{curcolor}{0.60000002 0.60000002 0.60000002}
\pscustom[linewidth=30.23622047,linecolor=curcolor]
{
\newpath
\moveto(220.45024059,210.53077133)
\curveto(231.2892113,210.71679968)(242.13934295,210.68066739)(252.95921193,210.01109385)
\curveto(266.82594705,209.15295968)(280.619605,207.25480913)(294.25398988,204.58530991)
\curveto(308.7042198,201.75607653)(322.97684973,198.06168283)(337.07758673,193.82105291)
\curveto(352.03234579,189.32358519)(367.3690639,183.86562613)(378.38173791,172.79344724)
\curveto(385.58880563,165.54742771)(390.60911429,155.94284314)(391.49050012,145.76100755)
\curveto(392.13014736,138.37175905)(390.60473004,130.89620125)(387.80054736,124.02984661)
\curveto(384.99640248,117.16349196)(380.94410642,110.86709385)(376.4410639,104.97363873)
\curveto(370.165045,96.75977291)(362.6778387,89.04365952)(353.13936941,85.05951212)
\curveto(346.57391429,82.31716629)(339.3298828,81.48400724)(332.22044028,81.76974708)
\curveto(325.11099776,82.05547936)(318.09181036,83.42482487)(311.14148595,84.94760031)
\curveto(299.17523713,87.56933007)(287.31164028,90.6581981)(275.29935681,93.06004629)
\curveto(260.52799933,96.01356566)(245.25889319,97.91854582)(230.43328437,95.25064913)
\curveto(214.80527429,92.43835936)(200.56029917,84.72451369)(186.57620973,77.20195747)
\curveto(173.09841067,69.95175401)(159.57654988,62.77571999)(145.82255051,56.06423527)
\curveto(138.35584059,52.42073669)(130.77942988,48.89877637)(122.77896752,46.65836314)
\curveto(114.77850516,44.41795369)(106.28098957,43.49512535)(98.13193886,45.11370047)
\curveto(90.65937445,46.5979096)(83.70098831,50.18222834)(77.68158043,54.85222015)
\curveto(71.66216878,59.52221196)(66.53918169,65.25937228)(61.96317917,71.35053826)
\curveto(54.0834198,81.83936598)(47.65391807,93.69634109)(45.23031114,106.58945102)
\curveto(42.01564154,123.69086456)(46.32648752,142.0718485)(56.81776059,155.95441228)
\curveto(68.75548532,171.75098928)(87.34435461,180.93927023)(105.66863429,188.44006015)
\curveto(118.17636847,193.55992913)(130.91608626,198.16614141)(143.98878043,201.59559779)
\curveto(156.79757288,204.95582456)(169.89265697,207.17639527)(183.06542925,208.53075495)
\curveto(195.48322579,209.8074907)(207.96881949,210.31657417)(220.45024059,210.53077133)
\closepath
}
}
{
\newrgbcolor{curcolor}{0.74901962 0.74901962 0.74901962}
\pscustom[linewidth=30.23622047,linecolor=curcolor]
{
\newpath
\moveto(226.63445295,239.70032295)
\curveto(238.51812287,239.47453398)(250.36227405,238.30288421)(262.17893673,237.02262217)
\curveto(275.61064067,235.56737933)(289.03071878,233.9700565)(302.37842271,231.88034815)
\curveto(318.39330146,229.37306626)(334.34619594,226.14454469)(349.70714083,220.96723823)
\curveto(369.66792193,214.23957713)(388.81172035,204.00446909)(403.37643216,188.78726343)
\curveto(409.99344004,181.87383476)(415.64908728,173.90422469)(419.01989673,164.94778374)
\curveto(423.49655815,153.05298689)(423.74766996,139.79005398)(420.84627783,127.41634563)
\curveto(417.94488571,115.04264106)(412.00271248,103.51713051)(404.54188728,93.22815665)
\curveto(394.17940539,78.93760043)(380.59200382,66.71615413)(364.44726051,59.57840295)
\curveto(350.24519437,53.29953429)(334.33988413,51.10888988)(318.8993424,52.75603823)
\curveto(296.81726366,55.11168547)(275.95339846,65.05908075)(253.78098901,66.30479035)
\curveto(239.15677232,67.12642185)(224.48220098,64.08617366)(210.81953768,58.80655161)
\curveto(198.46221358,54.03134532)(186.88421295,47.46355445)(175.51318303,40.66607413)
\curveto(160.42207752,31.6447765)(145.29438618,22.04019949)(128.26178271,17.67949398)
\curveto(112.14929768,13.55435886)(94.76033012,14.43407035)(79.2227339,20.36803823)
\curveto(57.81910303,28.54230595)(40.74277799,45.85811697)(29.53703059,65.84224421)
\curveto(19.78006681,83.24263728)(14.02611405,103.21458311)(15.29473894,123.12344106)
\curveto(16.58022429,143.29686217)(25.12804539,162.77137681)(38.33628098,178.0737239)
\curveto(54.47678366,196.77323886)(76.84971972,209.00581209)(99.73542051,218.30272043)
\curveto(116.82149295,225.24363382)(134.52711311,230.78831035)(152.70481138,233.89979886)
\curveto(162.39113201,235.55781335)(172.18164665,236.52118847)(181.96794712,237.41680043)
\curveto(196.82203846,238.77620201)(211.72097768,239.98367413)(226.63445295,239.70032295)
\closepath
}
}
{
\newrgbcolor{curcolor}{0 0 0}
\pscustom[linewidth=1.88976378,linecolor=curcolor]
{
\newpath
\moveto(57.20184945,177.61895149)
\curveto(25.20856819,149.05352031)(23.8374274,111.11863275)(38.46323528,79.58214519)
\curveto(53.08904315,48.04566141)(83.71057512,22.90858676)(118.67504882,29.53615086)
\curveto(153.6395263,36.16371495)(192.94399748,74.55412251)(233.85050079,80.26732062)
\curveto(274.75700409,85.98051873)(317.26152945,59.01528661)(354.51115465,71.35606393)
\curveto(391.76077984,83.69684125)(423.75342236,135.34211243)(399.5291452,170.76343653)
\curveto(375.30486803,206.18476062)(294.86782866,225.3799644)(223.33856504,225.3799644)
\curveto(151.80930142,225.3799644)(89.19513071,206.18437889)(57.20184945,177.61895149)
\closepath
}
}
{
\newrgbcolor{curcolor}{0 0 0}
\pscustom[linestyle=none,fillstyle=solid,fillcolor=curcolor]
{
\newpath
\moveto(97.52353588,12.71748329)
\curveto(97.52353588,12.25466003)(97.14834349,11.87946765)(96.68552023,11.87946765)
\curveto(96.22269697,11.87946765)(95.84750459,12.25466003)(95.84750459,12.71748329)
\curveto(95.84750459,13.18030655)(96.22269697,13.55549893)(96.68552023,13.55549893)
\curveto(97.14834349,13.55549893)(97.52353588,13.18030655)(97.52353588,12.71748329)
\closepath
}
}
{
\newrgbcolor{curcolor}{0 0 0}
\pscustom[linewidth=2.46803149,linecolor=curcolor]
{
\newpath
\moveto(97.52353588,12.71748329)
\curveto(97.52353588,12.25466003)(97.14834349,11.87946765)(96.68552023,11.87946765)
\curveto(96.22269697,11.87946765)(95.84750459,12.25466003)(95.84750459,12.71748329)
\curveto(95.84750459,13.18030655)(96.22269697,13.55549893)(96.68552023,13.55549893)
\curveto(97.14834349,13.55549893)(97.52353588,13.18030655)(97.52353588,12.71748329)
\closepath
}
}
{
\newrgbcolor{curcolor}{0 0 0}
\pscustom[linestyle=none,fillstyle=solid,fillcolor=curcolor]
{
\newpath
\moveto(96.58561776,45.58393452)
\curveto(96.58561776,45.12111126)(96.21042538,44.74591888)(95.74760212,44.74591888)
\curveto(95.28477886,44.74591888)(94.90958648,45.12111126)(94.90958648,45.58393452)
\curveto(94.90958648,46.04675778)(95.28477886,46.42195016)(95.74760212,46.42195016)
\curveto(96.21042538,46.42195016)(96.58561776,46.04675778)(96.58561776,45.58393452)
\closepath
}
}
{
\newrgbcolor{curcolor}{0 0 0}
\pscustom[linewidth=2.46803149,linecolor=curcolor]
{
\newpath
\moveto(96.58561776,45.58393452)
\curveto(96.58561776,45.12111126)(96.21042538,44.74591888)(95.74760212,44.74591888)
\curveto(95.28477886,44.74591888)(94.90958648,45.12111126)(94.90958648,45.58393452)
\curveto(94.90958648,46.04675778)(95.28477886,46.42195016)(95.74760212,46.42195016)
\curveto(96.21042538,46.42195016)(96.58561776,46.04675778)(96.58561776,45.58393452)
\closepath
}
\rput[bl](200,130){$\Omega$}
\rput[bl](103,40){$x$}
\rput[bl](105,5){$y$}
}
\end{pspicture}
\caption{A surface layer integral.}
\label{figosi}
\end{figure}%
We are interested in the situation where the Lagrangian and its derivatives vanish unless~$x$
and~$y$ are close together. Then the integrand in~\eqref{intdouble} vanishes
unless both~$x$ and~$y$ are ``close to the boundary of~$\Omega$''
(as shown again in Figure~\ref{figosi}).
Here, on purpose, the notion in quotation mark is not given in a mathematically
precise form. This could be done by quantifying the support properties of the
Lagrangian in terms of a distance function on~$M$ (as is done in~\cite[Section~9.1]{intro}).
Here we do not want to do so, but instead consider the surface layer 
integral~\eqref{intdouble}
as a generalization of surface integrals to the setting of causal variational principles.

The concept of a surface layer integral can also be used to introduce
integrals over subsets of higher co-dimension. This was first done in~\cite{jacobson}
in order to derive a simple connection between area change and matter flux
in a four-dimensional spacetime. In this setting,
a two-dimensional surface~$S \subset M$ was described as the intersection
of the boundaries of two open sets~$\Omega, V \subset M$,
\[ S = \partial \Omega \cap \partial V \:. \]
The corresponding surface layer integral was defined by
\[ 
\int_{\Omega \cap V} \bigg( \int_{M \setminus (\Omega \cup V)} (\cdots)\: \L(x,y)\: d\rho(y) \bigg)\, d\rho(x) \:. \]
Again assuming that the Lagrangian vanishes unless~$x$ and~$y$ are close together,
we only get contributions to this surface layer integral if both~$x$
and~$y$ are close to the two-dimensional surface~$S$ (see Figure~\ref{figarea}).
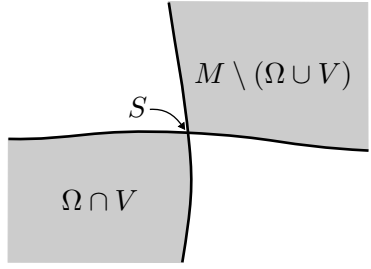
\begin{figure}
\psset{xunit=.5pt,yunit=.5pt,runit=.5pt}
\begin{pspicture}(271.1957508,195.81662156)
{
\newrgbcolor{curcolor}{0.80000001 0.80000001 0.80000001}
\pscustom[linestyle=none,fillstyle=solid,fillcolor=curcolor]
{
\newpath
\moveto(135.05108787,98.74163321)
\lineto(135.05108787,99.64047289)
\lineto(130.78158992,129.30227699)
\lineto(125.38852913,167.50308455)
\lineto(120.44489953,196.49075321)
\lineto(270.55159559,195.81662156)
\lineto(270.7762885,84.80957274)
\lineto(240.66505701,87.05668329)
\lineto(191.22871937,93.34857242)
\lineto(153.2526274,97.16865809)
\closepath
}
}
{
\newrgbcolor{curcolor}{0.80000001 0.80000001 0.80000001}
\pscustom[linewidth=0.37795276,linecolor=curcolor]
{
\newpath
\moveto(135.05108787,98.74163321)
\lineto(135.05108787,99.64047289)
\lineto(130.78158992,129.30227699)
\lineto(125.38852913,167.50308455)
\lineto(120.44489953,196.49075321)
\lineto(270.55159559,195.81662156)
\lineto(270.7762885,84.80957274)
\lineto(240.66505701,87.05668329)
\lineto(191.22871937,93.34857242)
\lineto(153.2526274,97.16865809)
\closepath
}
}
{
\newrgbcolor{curcolor}{0.80000001 0.80000001 0.80000001}
\pscustom[linestyle=none,fillstyle=solid,fillcolor=curcolor]
{
\newpath
\moveto(0.00000378,93.57328423)
\lineto(24.4934589,93.79801494)
\lineto(51.23402079,95.82040612)
\lineto(79.32284976,98.31778187)
\lineto(109.20935811,97.98015667)
\lineto(134.60166803,98.29219447)
\lineto(136.84876724,74.24815919)
\lineto(137.07346016,51.32767919)
\lineto(134.15221795,22.78943037)
\lineto(130.55685165,0.09365463)
\lineto(-0.22473071,0.31834754)
\closepath
}
}
{
\newrgbcolor{curcolor}{0.80000001 0.80000001 0.80000001}
\pscustom[linewidth=0,linecolor=curcolor]
{
\newpath
\moveto(0.00000378,93.57328423)
\lineto(24.4934589,93.79801494)
\lineto(51.23402079,95.82040612)
\lineto(79.32284976,98.31778187)
\lineto(109.20935811,97.98015667)
\lineto(134.60166803,98.29219447)
\lineto(136.84876724,74.24815919)
\lineto(137.07346016,51.32767919)
\lineto(134.15221795,22.78943037)
\lineto(130.55685165,0.09365463)
\lineto(-0.22473071,0.31834754)
\closepath
}
}
{
\newrgbcolor{curcolor}{0 0 0}
\pscustom[linewidth=1.99999997,linecolor=curcolor]
{
\newpath
\moveto(0.02221126,93.13902689)
\curveto(50.07131063,92.0273269)(60.5218105,101.32432678)(132.56010959,97.90832683)
\curveto(204.59840868,94.49232687)(202.72040871,87.97452695)(271.14280784,84.347727)
}
}
{
\newrgbcolor{curcolor}{0 0 0}
\pscustom[linewidth=1.99999997,linecolor=curcolor]
{
\newpath
\moveto(120.77424667,195.8166191)
\curveto(126.77414659,154.49331962)(131.43884654,137.99751983)(136.28644647,87.69692046)
\curveto(139.14604644,58.02412083)(137.84984645,42.66522103)(130.94594654,0.16032156)
}
}
{
\newrgbcolor{curcolor}{0 0 0}
\pscustom[linewidth=1.00157475,linecolor=curcolor]
{
\newpath
\moveto(106.51283906,115.37021652)
\curveto(117.50448756,115.98224675)(124.48166551,111.9394373)(131.00629417,102.78642313)
}
}
{
\newrgbcolor{curcolor}{0 0 0}
\pscustom[linestyle=none,fillstyle=solid,fillcolor=curcolor]
{
\newpath
\moveto(132.68355278,100.43349661)
\lineto(128.13018718,102.50741712)
\lineto(132.20804942,105.41427779)
\closepath
}
}
{
\newrgbcolor{curcolor}{0 0 0}
\pscustom[linewidth=0.66771652,linecolor=curcolor]
{
\newpath
\moveto(132.68355278,100.43349661)
\lineto(128.13018718,102.50741712)
\lineto(132.20804942,105.41427779)
\closepath
}
\rput[bl](40,40){$\Omega \cap V$}
\rput[bl](140,130){$M \setminus (\Omega \cup V)$}
\rput[bl](90,110){$S$}
}
\end{pspicture}
\caption{A surface layer integral of co-dimension two.}
\label{figarea}
\end{figure}%

In Section~\ref{secsoft} we shall elaborate on this idea, making it possible
to introduce surface layer integrals of general co-dimension.
As we shall see, these surface layer integrals harmonize with our exterior calculus,
giving rise to a general formulation of Stokes' theorem (Theorem~\ref{thmstokesgen}).

\section{Osculating Vacua} \label{secosculate}
In what follows, we assume that we are given two minimizing measures: The measure~$\rho$
describing the vacuum, and a measure~$\tilde{\rho}$ describing the interacting system.
We assume that the {\em{vacuum space}}~$M:= \supp \rho$
has the dimension of a real vector space of dimension~$k \in \N$.
This means in particular that
one point of~$M$ is distinguished as the origin; we denote it by~$\0 \in M$.
In what follows, we always identify~$M$ with its tangent space~$T_\0M$.
Next, we assume that the {\em{Lagrangian}} is {\em{translation invariant}}
on~$M$, i.e.\
\[ \L(x,y) = L[y-x] \qquad \text{for all~$x,y \in M$} \]
with a function~$L : M \rightarrow \R^+_0$. The fact that the Lagrangian is symmetric
in its two arguments implies that the function~$L$ is {\em{reflection
symmetric}}, i.e.\
\[ L[\xi] = L[-\xi] \qquad \text{for all~$\xi \in M$} \:. \]
Moreover, we we assume that, using the above vector space structure,
the measure~$\rho|_M$ is a multiple of the Lebesgue measure.
For the applications, one can think of~$M$ as Minkowski space. However, in the constructions
of this paper, we do not want to make use of the Minkowski metric. Instead, we only want to make use of
the measure~$\rho$ as specified above.

The structure of the {\em{interacting space}}~$\tilde{M} \subset \F$, however, can be arbitrary.
In particular, the interacting space does not need to have a manifold structure. Our constructions
will even cover the situation that~$\tilde{M}$ is singular or discrete.

Following~\cite[Definition~6.1]{matter} we introduce symmetry transformations of the Lagrangian.
\begin{Def} \label{defiso}
A diffeomorphism~$\Phi \in C^\infty(\F, \F)$ describes a {\bf{symmetry of the Lagrangian}} if
\[ \L\big(\Phi(\x), \Phi(\y) \big) = \L(\x,\y) \qquad \text{for all~$\x,\y \in \F$}\:. \]
\end{Def} \noindent
Such diffeomorphisms form a group, denoted by~${\mathcal{G}}$, the {\em{group of symmetries of the Lagrangian}}.
Before going on, we remark that in the physical
setting of causal fermion systems, the unitary
transformations on the underlying Hilbert space give rise to symmetries
of the Lagrangian (see~\cite[Section~8.2]{intro}). In this case, 
the symmetry group~$\G$ consists of the group of all unitary transformations
on the Hilbert space~\cite{paganini+yadav}.
With this in mind, the Lagrangian typically has many
symmetries as described by a high-dimensional compact Lie group.
For the constructions in this paper, we do not need to specify~$\G$.

For introducing our concepts, it is best to begin with the simplest setting
where we assume that the group of symmetries~$\G$ of the Lagrangian
{\em{acts transitively on~$\tilde{M}$}} in the sense that for every~$p \in \tilde{M}$
there is a symmetry transformation~$\Phi \in \G$ with~$\phi(\0)=p$.
In this case, given a point~$p \in \tilde{M}$, we consider the set of all symmetry transformations which map~$\0 \in M$ to~$p$,
\beq \label{Gpdef}
\mathcal{G}_p := \{ \Phi_p \in \G \:|\: \Phi_p(\0) = p \} \:.
\eeq
Given~$\Phi_p \in \mathcal{G}_p$, we set
\[ \rho_p := (\Phi_p)_* \rho \qquad \text{and} \qquad M_p := \supp \rho_p \:. \]
We refer to~$\Phi_p$ as the {\em{osculation}} and~$\rho_p$ as the {\em{osculating vacuum}} at~$p \in \tilde{M}$.
These notions are illustrated in Figure~\ref{figosculate-discrete} in the simple example
where~$\tilde{M}$ is discrete.
For the following constructions, we assume that for every~$p \in \tilde{M}$ we have chosen an osculation~$\Phi_p \in {\mathcal{G}}_p$.
We also assume that the osculation depends continuously on the base point,
meaning that the mapping
\beq \label{loctop}
\tilde{M} \rightarrow C^\infty(\F, \F)\:,\qquad p \mapsto \Phi_p
\eeq
is continuous, where on~$C^\infty(\F, \F)$ we take the topology
of locally uniform convergence of all derivatives.
We remark that, in the setting of causal fermion systems, this simply
corresponds to norm convergence of the unitary operators describing
the symmetry transformations.

Our constructions still apply in the case that~$\G$ does not act
transitively on~$\tilde{M}$. Clearly, in this case, the set~$\G_p$ defined by~\eqref{Gpdef}
may be empty. Instead, we choose the mapping~$\Phi_p \in \G$
in such a way that~$\Phi_p(\0)$ is approximately equal to~$p$.
For the topological questions to be studied here, we do not need
to specify what we mean by ``approximately equal'', nor is it important
how precisely the osculation~$\Phi_p$
is chosen (indeed, in Section~\ref{secindep} we will specify conditions
under which topological invariants are independent of the choice
of the osculation).
Therefore, here we only sketch how~$\Phi_p$ can be constructed,
referring for the details to~\cite[Section~3]{gauss}.
Intuitively speaking, the idea is to choose~$\Phi_p$ in such a way that the spaces~$M$ and~$\tilde{M}_p$
``agree as much as possible near~$p$''. In order to make this statement mathematically precise, one
can set up a variational principle at~$p$. In the case that~$\G$
acts transitively on~$\tilde{M}$, one chooses a basis~$(e_i)_{i=, \ldots, k}$ of~$M$.
Then~$e_i(p) := D (\Phi_p)|_{\mathbf{0}}\, e_i \in T_p \F$ is a basis of~$M_p$. We set
\beq \label{Spdef}
\Sact_p(\Phi_p) = \sum_{i=1}^k D^2 \tilde{\ell}|_p \big( e_i, e_i \big) \:,
\eeq
where~$\tilde{\ell}$ is defined by~\eqref{elldef} with~$\rho$ replaced by~$\tilde{\rho}$. Note that, in view of the EL equations~\eqref{EL1} for~$\tilde{\rho}$, the first derivative of~$\tilde{\ell}$
vanishes at~$p$. This is why the second derivatives in~\eqref{Spdef} is well-defined without the need to specify
a connection on~$\F$. Moreover, it follows from~\eqref{EL1} that the second derivatives in~\eqref{Spdef} are non-negative.
Therefore, the action~$\Sact_p$ is non-negative.
Now one chooses~$\Phi_p$ as a minimizer of~$\Sact_p$ under variations~$\Phi_p \in {\mathcal{G}}_p$. In the general case that~$\G$ does not
act transitively on~$\tilde{M}$, one replaces the constraint~$\Phi_p(\0)=p$
by a weaker conditions.
More details on this method can be found in~\cite[Section~3.2]{gauss}.

We finally comment on the name ``osculating vacua''.
The notion ``osculating'' (literally ``kissing'') can be found in the older literature
(see for example~\cite{lichnerowiczintro, laugwitz}) for a
Euclidean metric which approximates a Riemannian metric in a neighborhood of a point. In this setting, the Euclidean and Riemannian metrics are osculating
if they coincide in a first order Taylor expansion about a base point.
This is as good as possible, because the second derivatives of the Riemannian
metric involve curvature, which clearly cannot be compensated by a coordinate
transformation.
Our notion of ``osculating vacua'' is similar in the sense
that~$M_p$ should approximate~$\tilde{M}$ near~$p$ as good as possible
(as is made precise by the variational principle~\eqref{Spdef}).
In this way, one can indeed approximate the geometry, 
as is made precise in the recent paper~\cite{gauss},
where it is shown that the osculating vacuum at~$p \in \tilde{M}$
gives rise to a chart in which the Christoffel symbols vanish at~$p$ 
(see~\cite[Lemma~4.5]{gauss}).
One should keep in mind that, in contrast to the setting of Riemannian geometry,
the symmetry transformation~$\Phi_p$ not only aims at adjusting the geometry,
but also tries to adapt all the other structures
encoded in~$\tilde{\rho}$ through the variational principle.

\section{Directional Derivatives and Vector Fields} \label{secdirder}
We let~$\h := L^2(\tilde{M}, d\tilde{\rho})$ be the
Hilbert space of square-integrable real-valued functions on~$\tilde{M}$ endowed with the scalar product
\[ \la f, g \ra_\h := \int_{\tilde{M}} f\,g\: d\tilde{\rho} \:. \]
Given a function~$f$ in this Hilbert space, we want to form its derivative. Before doing so, however, we must
mollify the function with the Lagrangian; more precisely, we set
\beq \label{Lrhodef}
(\L_{\tilde{\rho}} \, f)(x) := \fint_{\tilde{M}} \L(x,y)\: f(y)\: d\tilde{\rho}(y) \:,
\eeq
where the integral sign with bar means that the integral is rescaled with
a factor~$1/\s$, i.e.\
\beq \label{fintdef}
\fint \cdots := \frac{1}{\s} \int \cdots \:.
\eeq
We note that, by rescaling the
Lagrangian, one can arrange that the parameter~$\s$ is equal to one.
Doing so would simplify many of the subsequent equations. Nevertheless, 
in order to clarify scaling dimensions and for the sake
of consistency with the literature, we decided against setting~$\s$ to one
(but we will use 
Instead, we will carry the parameter~$\s$ along in all our formulas,
using the shorthand notation~\eqref{fintdef}.

Noting that~$M_p$ is a submanifold of~$\F$, a vector~$u \in M_p$
can be identified with a tangent vector~$u \in T_\0 M_p \subset T_p \F$.
This allows us to define the {\em{directional derivative}} by
\[ D_u f := D_u (\L_{\tilde{\rho}} \, f) = \fint_{\tilde{M}} D_{1,u} \L(x,y)\:
f(y)\: d\tilde{\rho}(y) \Big|_{x=\0_p} \:\in\: \R\:, \]
where for clarity the origin in~$M_p$ (considered as a vector space) is denoted by~$\0_p$.
Viewing~$M_p$ as embedded in~$\F$, we can identify~$\0_p$ with~$p \in \tilde{M} \subset \F$.

A {\em{vector field}}~$v$ on~$\tilde{M}$ is a continuous
mapping which to every~$p \in \tilde{M}$ associates a vector
in the corresponding osculating vacuum,
\[ v : \tilde{M} \rightarrow T\F \qquad \text{with} \qquad v(p) \in T_\0M_p \simeq M_p 
\quad \text{for all~$p \in \tilde{M}$}\:. \]
We denote the space of vector fields by~$\Gamma(\tilde{M}, TM)$.
Consequently, taking the directional derivative at each point of~$\tilde{M}$ gives the mapping
\beq \label{dirdiff}
D_u \::\: C^0(\tilde{M}) \rightarrow C^0(\tilde{M}) \:,\qquad D_u f := D_u (\L_{\tilde{\rho}} \, f) = \fint_{\tilde{M}} D_{1,u} \L(x,y)\: f(y)\: d\tilde{\rho}(y) \:.
\eeq
For two compactly supported functions~$f,g \in C^0_0(\tilde{M}) \subset \h$, we can compute
the formal adjoint of the directional derivative by
\begin{align*}
\la D_u f, g \ra_\h &= \fint_{\tilde{M}} 
\bigg( \int_{\tilde{M}} D_{1,u} \L(x,y)\: f(y)\: d\tilde{\rho}(y) \bigg)\: g(y)\: d\tilde{\rho}(x) \\
&= \fint_{\tilde{M}} d\tilde{\rho}(y) \int_{\tilde{M}}  d\tilde{\rho}(x)\: f(y)\: D_{2,u} \L(y,x)\: g(x) 
= \la f, D_u^* g \ra_\h \:. \end{align*}
By extending this formal adjoint to all continuous functions, we obtain the mapping
\[ D_u^* \::\: C^0(\tilde{M}) \rightarrow C^0(\tilde{M}) \:,\qquad D_u^* f(p) := \fint_{\tilde{M}} D_{2,u} \L(x,y)\:
f(y)\: d\tilde{\rho}(y) \:. \]

\section{The Gau{\ss} Divergence Theorem} \label{secgauss}
In this section we introduce the divergence of a vector field and prove
a corresponding Gau{\ss} theorem. This generalizes concepts
introduced for smooth manifolds in~\cite{fockbosonic, noether}
(see also~\cite[Section~8.3 and Proposition~9.3.1]{intro})
to the non-smooth setting.
\begin{Def} The {\bf{divergence}} of a vector field~$v \in \Gamma(\tilde{M}, TM)$ is defined by
\[ \div \::\: \Gamma(\tilde{M}, TM) \rightarrow \R \:,\qquad
\div v(x) := - (D_v^* \tilde{1})(x) = -\fint_{\tilde{M}} D_{2,v} \L(x,y)\: d\tilde{\rho}(y) \Big|_{x=p} \]
(where~$\tilde{1}$ is the constant function one on~$\tilde{M}$).
\end{Def}

\begin{Thm} {\bf{(Gau{\ss} divergence theorem)}} \label{thmgauss}
$\qquad$ \\ For any relatively compact Borel set~$\tilde{U} \subset \tilde{M}$
and any vector field~$v \in \Gamma(\tilde{M}, TM)$,
\[ 
\int_{\tilde{U}} \div v\: d\tilde{\rho}
= \frac{1}{\s} \int_{\tilde{U}} d\tilde{\rho}(x) \int_{\tilde{M} \setminus
\tilde{U}} d\tilde{\rho}(y)\:\big( D_{1,v} - D_{2,v} \big) \L(x,y)\:. \]
\end{Thm}
\Proof A direct computation using the symmetry of the Lagrangian gives
\begin{align}
-\s &\int_{\tilde{U}} \div v\: d\tilde{\rho}
= \int_{\tilde{U}} d\tilde{\rho}(x) \int_{\tilde{M}} d\tilde{\rho}(y)\:D_{2,v} \L(x,y) \notag \\
&= \int_{\tilde{M}} d\tilde{\rho}(x) \int_{\tilde{U}} d\tilde{\rho}(y)\:D_{1,v} \L(x,y) \notag \\
&= \int_{\tilde{U}} d\tilde{\rho}(x) \int_{\tilde{U}} d\tilde{\rho}(y)\:D_{1,v} \L(x,y)
+ \int_{\tilde{M} \setminus \tilde{U}} d\tilde{\rho}(x) \int_{\tilde{U}}
d\tilde{\rho}(y)\:D_{1,v} \L(x,y)\notag \\
&= \int_{\tilde{U}} d\tilde{\rho}(x) \int_{\tilde{M}} d\tilde{\rho}(y)\:D_{1,v} \L(x,y) \label{zero} \\
&\quad\: -\int_{\tilde{U}} d\tilde{\rho}(x) \int_{\tilde{M} \setminus
\tilde{U}} d\tilde{\rho}(y)\:D_{1,v} \L(x,y)
+ \int_{\tilde{M} \setminus \tilde{U}} d\tilde{\rho}(x) \int_{\tilde{U}} d\tilde{\rho}(y)\:
D_{1,v} \L(x,y) \:.
\end{align}
The term~\eqref{zero} vanishes in view of the EL equations~\eqref{EL1}
for~$\tilde{\rho}$, giving the result.
\QED

\section{$\L$-Induced Charts and the Induced Connection} \label{secconnection}
Suppose that we have chosen an osculating vacuum~$M_p$ at~$p$.
For any~$x \in \F$, the integral
\[ K_p(x) := \fint_{M_p} \L(x,y)\: y\: d\rho_p(y) \]
gives a vector in~$M_p$ (where we again regard~$M_p \simeq T_\0M_p$ as a vector space).
We thus obtain a mapping
\beq \label{Kpdef}
K_p \::\: \F \rightarrow M_p \:,\qquad
x \mapsto \fint_{M_p} \L(x,y)\: y\: d\rho_p(y) \:.
\eeq
It is useful to restrict this mapping to~$\tilde{M}$,
\beq \label{phipdef}
\phi_p := K_p|_{\tilde{M}} \::\: \tilde{M} \rightarrow M_p \:.
\eeq
In case that~$\tilde{M}$ has a manifold structure and assuming that~$\phi_p$ is a local
homeomorphism, one can regard~$\phi_p$ as a chart of~$\tilde{M}$ around~$p$
(for details see~\cite[Section~4.1]{gauss}).
In the general situation when~$\tilde{M}$ has no manifold structure (for example, describing
a discrete space), we can still use the mapping~$\phi_p$ to describe~$\tilde{M}$
locally in terms of vectors in~$M_p$ (at least if the mapping~$\phi_p$ is locally injective).
With this general picture in mind, we refer to~$\phi_p$
as the {\em{$\L$-induced chart around~$p \in \tilde{M}$}}.

On a Riemannian manifold, the Levi-Civita connection is characterized
by the property
that in a normal coordinate system (like Gaussian coordinates or coordinates induced by the
exponential map), the covariant derivative reduces to ordinary derivatives.
We now use this relation in the reverse order and use the $\L$-induced charts to {\em{define}} a connection.
To this end, given~$q \in \tilde{M}$, we restrict the function~$K_q$ to~$M_p$,
\[ \phi_{q,p} := K_q|_{M_p} \::\: M_p \rightarrow M_q \:,\qquad
\phi_{q,p}(x) = \fint_{M_q} \L(x,y)\: y\: d\rho_q(y) \:. \]
This mapping is in general not linear. But its linearization can be used to
define the connection
\[ \nabla^\L_{q,p} := \big((\phi_{q,p})_*\big)_{\0_p} \::\: M_p \rightarrow M_q \:,\qquad
\nabla^\L_{q,p} u = \fint_{M_q} D_{1,u} \L(x,y)\: y\: d\rho_q(y) \:. \]
In the case~$p=q$, it follows from the EL equations as well as from the
symmetry of the vacuum under translations and reflections that
\[ \phi_{q,q}(x) = \fint_{M_q} \L(x,y)\: y\: d\rho_q(y) = x \:. \]
This implies that~$\nabla^\L_{q,q}$ is the identity,
\beq \label{nablaid}
\nabla^\L_{q,q} = \1  \:.
\eeq
 But we point out that, in general, 
\[ \nabla^\L_{q,p}\: \nabla^\L_{p,q} \neq \1 \:. \]

By continuity, the relation~\eqref{nablaid} implies that every point~$p \in \tilde{M}$
has an open neighborhood~$U \subset \tilde{M}$ such that the mapping
\beq \label{vectorbundle}
\nabla^\L_{q,p} \::\: M_p \rightarrow M_q \text{ is invertible
for all~$p \in U$} \:.
\eeq
These mappings give a natural identification of the osculating vacua~$M_p$
for all~$p \in U$. Using these identifications, the collection of all osculating
vacua~$(M_p)_{p \in \tilde{M}}$ is endowed with the structure of a
{\em{topological vector bundle}}
(for basic definitions see~\cite{milnor+stasheff, steenrod}; cf.\
also~\cite[Section~3.2]{topology}).

In order to clarify the geometric significance of the connection~$\nabla^\L$, the following 
consideration
seems useful. We note that~$\F$ is a manifold, but it is not endowed with a Riemannian
metric. Nevertheless, we can decompose its tangent space~$T_p \F$ canonically
into the direct sum of~$M_p$ and
\[ N_p := \ker D K_p(p) \subset T_p \F \:, \]
referred to as the {\em{normal space}} at~$p$. Regarding the decomposition
\[ T_p \F = M_p \oplus N_p \:, \]
as an orthogonal direct sum, the Lagrangian gives rise to a notion of orthogonality to the tangent space.
Likewise, we can form the {\em{orthogonal projection}}
\[ \Pi : T_p \F \rightarrow M_p \:,\qquad (u,v) \in M_p \oplus N_p \mapsto u \:. \]
In this context, the connection~$\nabla^\L$ can be understood similar
to the classical Gau{\ss} map, being obtained by taking the orthogonal projection 
of an ambient covariant derivative to~$M_p$.
More details on the geometric significance of~$\nabla^\L$ and~$\Pi$
can be found in~\cite[Section~4 and Appendix~A]{gauss}.

\section{The $\L$-Calculus for Differential Forms} \label{secLcalc}
\subsection{Weak Derivatives and Mollified Evaluation}
So far, we only considered first order derivatives. The method was to first
take the convolution with the Lagrangian and then to differentiate
tangentially to~$M_p$ (see~\eqref{dirdiff}).
Higher order derivatives could be defined iteratively. However, this method
has the shortcoming that, in contrast to partial derivatives in a chart,
these directional derivatives do in general not commute.
In other words, the resulting higher derivatives are in general not symmetric,
but violate the Schwarz theorem.

Our strategy to bypass these issues is to work with {\em{weak derivatives}}.
These concepts and our notation are inspired by and have similarities to
familiar concepts in distribution theory and partial differential equations.
However, the adaptation to our setting is rather subtle and not straightforward.
In general terms, we aim for a formula of the form
\beq \label{wevalmotivate}
\fint_{\tilde{M}} \L(x,y)\: D^\kappa f(y) \: d\tilde{\rho}(y)
:= (-1)^{|\kappa|}\: \fint_{\tilde{M}} \big( D_2^\kappa \L(x,y) \big) \,f(y)\: d\tilde{\rho}(y) \:,
\eeq
where~$\kappa$ is a multi-index. Similar to the derivative of a distribution,
the derivative~$D^\kappa f$ is {\em{defined}} by the integral on the right side.
In order to give~$D^\kappa f(p)$ a mathematical meaning independent of
the integrals and of choices of bases, it is convenient to consider it as a section in
the symmetric tensor product of~$M_p$, as we now explain in detail.
Given~$m \in \N$, for any~$p \in \tilde{M}$ we consider the symmetric
tensor product of~$M_p$,
\[ M_p^m := \underbrace{M_p \odot \cdots \odot M_p}_{\text{$m$ factors}} \:, \]
where~$\odot$ denotes the symmetrized tensor product, i.e.\
\[ 
u_1 \odot \cdots \odot u_m  := \frac{1}{m!} \sum_{\sigma \in S_m}
\psi_{\sigma(1)} \otimes \cdots \otimes \psi_{\sigma(m)} \:, \]
where~$S_m$ denotes the group of all permutations.
Next, we let~$T^m\tilde{M}$ be the vector bundle over~$\tilde{M}$ with fibres~$M_p^m$
(constructed as explained after~\eqref{vectorbundle}).
We denote the continuous sections in this bundle by~$\Gamma(\tilde{M}, T^m\tilde{M})$.
Thus, a section~$f$ in this bundle can be written as
\beq \label{fdef}
f(p) = \sum_{i_1,\ldots, i_m} h^{i_1,\ldots, i_m}(p)\; u_{i_1}(p) \odot
\cdots \odot u_{i_m}(p)
\eeq
with~$u_i(p) \in M_p$ and coefficients~$h^{i_1, \ldots, i_m} \in \R$.
In view of the glueing constructions in Sections~\ref{secext} and~\ref{secmv}, it is preferable to work
with sections which are merely in~$L^\infty_\loc$, denoted by~$\Gamma^\infty_\loc(\tilde{M}, T^m \tilde{M})$.

The purpose of these sections in the tensor bundle is to describe derivatives
acting on the Lagrangian. More precisely, we consider
derivatives of~$\L(x,y)$ on~$M_y$, viewed as multilinear mappings, i.e.
\beq \label{D2def}
D_2 \L(x,y) \::\: M_y \rightarrow \R \qquad \text{and} \qquad
D^m_2 \L(x,y) \::\: \underbrace{M_y \times \cdots \times M_y}_{\text{$m$ factors}}
\rightarrow \R \:.
\eeq
We now define
\begin{align}
& \fint_{\tilde{M}} \L(x,y) \cdot f(y) \: d\tilde{\rho}(y) \notag \\
&\quad := (-1)^m \sum_{i_1,\ldots, i_m} 
\fint_{\tilde{M}} \big( D_2^m \L(x,y) \big) \big( u_{i_1}, \ldots, u_{i_m} \big)\: h^{i_1, \ldots, i_m}(y)
\: d\tilde{\rho}(y) \:. \label{weval}
\end{align}
This gives the right side of~\eqref{wevalmotivate} a precise mathematically
meaning. ``Formally integrating by parts'' gives in analogy to the left side of~\eqref{wevalmotivate} the expression
\begin{align}
&\fint_{\tilde{M}} \L(x,y) \cdot f(y) \: d\tilde{\rho}(y) \notag \\
& \quad= \sum_{i_1,\ldots, i_m}  \fint_{\tilde{M}} \L(x,y)\:
D^m \big(h^{i_1, \ldots, i_i} \big)\big|_u \big( u_{i_1}, \cdots, u_{i_m} \big)\: d\tilde{\rho}(y) \:.
\label{formweak}
\end{align}
However, the derivatives on the right side are ill-defined, because
the coefficients~$h_{i_1, \ldots, i_m}$ are defined only on~$\tilde{M}$, where there is 
no differentiable structure. Moreover, one should keep in mind that the
integration is over~$\tilde{M}$ (and not over~$M_y$), which again shows that
it makes no mathematical sense to integrate by parts.
Nevertheless, it could be helpful for understanding the definition~\eqref{weval}
as a way to give derivatives of functions on~$\tilde{M}$ 
(as in~\eqref{formweak}) a rigorous mathematical meaning.
We refer to~\eqref{weval} as the {\em{mollified evaluation of weak
derivatives}}.

The similarity to weak derivatives becomes clearer if we work with
distinguished bases of the osculating vacua. Since~\eqref{weval} involves
the osculating vacua~$M_y$ at different points~$y$, it is important to
canonically identify the basis vectors by constructing suitable local frames.
This can be accomplished using with
the connection~$\nabla^\L$ introduced in Section~\ref{secconnection}, as we
now explain. In preparation, we need to specify the size of the neighborhoods
on which the frames should be defined. In order to ensure that in expressions
like~\eqref{weval} the points~$x$ and~$y$ lie in the same neighborhood,
it is useful to introduce the following notions.
\begin{Def} \label{defBL}
Given~$x \in \tilde{M}$, we define the {\bf{$\L$-ball}}~$B_\L(x)$ by
\[ B_\L(x) := \big\{ y \in \tilde{M} \:\big|\: \text{there is~$m \in \N_0$
with~$D_2^m \L(x,y) \neq 0$} \big\}
\overset{\text{open}}{\subset} \tilde{M} \]
(with the derivative~$D_2^m \L(x,y)$ as defined in~\eqref{D2def}).
Likewise, for a set~$U \subset \tilde{M}$, we define its {\bf{$\L$-neighborhood}}~$B_\L(U)$ by
\[ B_\L(U) := \bigcup_{x \in U} B_\L(x) =
\big\{ y \in \tilde{M} \:\big|\: \text{there is~$x \in U$ and~$m \in \N_0$ with~$D_2^m
\L(x,y) \neq 0$}
\big\} \:. \]
\end{Def} \noindent
Note that, since the Lagrangian is strictly positive on the diagonal
(see property~(ii) of the Lagrangian in the preliminaries), we know that~$U \subset B_\L(U)$.
Hence~$B_\L(U)$ really is an open neighborhood of~$U$.

Given~$x \in \tilde{M}$, we want to construct a canonical frame in~$B_\L(x)$.
To this end, we need the following assumption.
\begin{Def} The point~$x \in \tilde{M}$ is {\bf{$\L$-regular}} if for every~$y \in B_\L(x)$,
the parallel transport~$\nabla_{x,y} : M_y \rightarrow M_x$ is bijective.
\end{Def} \noindent
In what follows, we shall always make the blanket assumption that all points of~$\tilde{M}$
are $\L$-regular.
Under this assumption, we choose a basis~$(e_i)$ of~$M_x$ and
define corresponding bases~$(e_i(y))$ of~$M_y$ with~$y \in B_\L(x)$ by the conditions
\beq \label{parallelframe}
e_i = \nabla^\L_{x,y}\, e_i(y) \:.
\eeq
We refer to this frame as a {\em{parallel frame around~$x$}}.
Likewise, in order to obtain a frame of the dual spaces, we let~$(e^i)$
be the dual basis of~$M_x^*$ (defined by~$(e^i(e_j)=\delta^i_j$)
and define bases~$(e^i(y))$ of~$M_y^*$ by
\beq \label{dualframe}
e^i(y) \big( \nabla^\L_{y,x} e_i \big) = \delta^i_j \:.
\eeq
Note that~$e^i(y)$ is in general {\em{not}} the dual basis of~$e_i(y)$. The reason for
our at first sight counter-intuitive convention is the following: Our 
definitions~\eqref{parallelframe} 
and~\eqref{dualframe} ensure that the basis vectors of both the frame and the
dual frame become large if~$\L(x,y)$
and its derivatives become small, as typically happens near the boundary
of~$B_\L(x)$. Consequently, the coefficient functions will tend to zero near this boundary,
improving the analytic properties of integrals as in~\eqref{weval}.

Expressing the section~$f$ in~\eqref{fdef} in the basis vectors of the parallel frame,
we obtain locally
\beq \label{fdef1}
f(y) = \sum_{i_1,\ldots, i_m} h^{i_1,\ldots, i_m}(y)\;e_{i_1}(y) \odot
\cdots \odot e_{i_m}(y) \:.
\eeq
Using the multi-index notation
\[ D_2^\kappa \L(x,y) := \big( D_2^m \L(x,y) \big) \big( e_{i_1}, \ldots, e_{i_m} \big)
\qquad \text{with} \qquad \kappa := (i_1, \ldots, i_m) \:, \]
we can write~\eqref{weval} in the more compact and familiar form
\[ 
\fint_{\tilde{M}} \L(x,y) \cdot f(y) \: d\tilde{\rho}(y)
= (-1)^{|\kappa|} \fint_{\tilde{M}} D_{2,\kappa} \L(x,y)\: h^\kappa(y)\: d\tilde{\rho}(y) \:, \]
where~$|\kappa|:=m$ denotes the order of the multi-index. Moreover, the sum
was omitted using the Einstein summation convention.
We also abbreviate the formal weak derivative in~\eqref{formweak} by
\[ D_\kappa h^\kappa(y) := \sum_{i_1,\ldots, i_m}
D^m \big(h^{i_1, \ldots, i_i} \big)\big( u_{i_1}, \cdots, u_{i_m} \big) \:. \]
We point out that this merely is a convenient notation. Similar to the notation
used for distributional derivatives, it becomes mathematically
meaningful only after formally integrating by parts, i.e.\
\[ \fint_{\tilde{M}} \L(x,y) \cdot D_\kappa h^\kappa(y) \: d\tilde{\rho}(y)
\overset{\text{def}}{=} (-1)^{|\kappa|}
\fint_{\tilde{M}} D_{2,\kappa} \L(x,y) \: h^\kappa(y) \: d\tilde{\rho}(y) \:. \]
Despite this analogy to distribution theory, there is also the major difference that we
test exclusively with one function: the Lagrangian~$\L$.

Working in the parallel frame~$e_i(y)$, we can also introduce
{\em{weak partial derivatives}} by
\[ \fint_{\tilde{M}} D_{2,\kappa} \L(x,y) \cdot D_j h^\kappa(y) \: d\tilde{\rho}(y) 
:= -\fint_{\tilde{M}} D_{2,\kappa} D_{2,j} \L(x,y) \: h^\kappa(y) \: d\tilde{\rho}(y) \:. \]
In order to obtain an index-free notation, we can regard the partial derivatives
(for a fixed base point~$x \in \tilde{M})$ as an operator
\[ D_j \::\: \Gamma^\infty_\loc(\tilde{M}, T^m\tilde{M}) \rightarrow \Gamma^\infty_\loc(\tilde{M}, T^{m+1}\tilde{M}) \]
defined by
\[ (D_j f)(y) := \sum_{i_1,\ldots, i_m} h^{i_1,\ldots, i_m}(y)\; e_j(y) \odot
e_{i_1}(y) \odot \cdots \odot e_{i_m}(y) \]
(with~$f$ again given by~\eqref{fdef1}).
Since the Lagrangian is smooth, these weak derivatives exist and 
satisfy the Schwarz theorem,
\[ D_i \big( D_j  f)(y) = D_j \big( D_i  f)(y) \qquad \text{for all~$i,j =1,\ldots, m$}\:, \]
and similarly for higher derivatives.
This property will be crucial for the construction of differential forms in the next section.

We now proceed by defining the needed function spaces.
In words, these function spaces are obtained by
taking equivalence classes of tensor sections
which coincide after mollified evaluation to a given order~$\ell$.
\begin{Def} \label{defvanish}
A section~$f \in \Gamma^\infty_\loc \big( \tilde{M}, T^m \tilde{M} \big)$
{\bf{vanishes at~$x \in \tilde{M}$ to the order~$\ell \in \N_0$}} if
\beq \label{vanish}
\fint_{\tilde{M}} \L(x,y) \cdot (D_\lambda f)(y)\: d\tilde{\rho}(y) = 0
\qquad \text{for all multi-indices~$\lambda$ with~$|\lambda| \leq \ell$}\:.
\eeq
Likewise, $f$ vanishes to the order~$\ell$ if it vanishes to the order~$\ell$
at all~$x \in \tilde{M}$.
\end{Def} \noindent

We introduce on~$\Gamma^\infty_\loc \big( \tilde{M}, T^m \tilde{M} \big)$ the equivalence relation
\beq \label{fequiv}
f \simeq^\ell_{\tilde{M}} g \qquad \text{if} \qquad
\text{$f-g$ vanishes to the order~$\ell$}\:.
\eeq
For~$\ell, m \in \N_0$ we define the vector space
\beq \label{Cequi}
\D^{m, \ell}(\tilde{M}) :=
\bigg( \bigoplus_{r=0}^m
\Gamma^\infty_\loc \big( \tilde{M}, T^r \tilde{M} \big) \bigg) \Big/ \simeq^\ell_{\tilde{M}} \:.
\eeq
Note that the vectors in this vector space are equivalence classes of tensor sections.
Our notation is motivated by the fact that these equivalence classes involve
weak derivatives of order at most~$m$.
The parameter~$\ell$, on the other hand, is the maximal order of derivatives taken
into account when forming equivalence classes. We note that the equivalence classes
become smaller if~$\ell$ is increased, and thus there is a canonical quotient map
\[ \D^{m, \ell+1}(\tilde{M}) \rightarrow \D^{m, \ell}(\tilde{M}) \quad \text{surjective} \:. \]

We finally point out that the partial derivatives~$D_j$ (introduced at a base point~$x$)
give rise to a globally defined mapping
\[ D_j\::\: \D^{m, \ell}(M) \rightarrow \D^{m+1, \ell-1}(M) \qquad \text{if~$\ell \geq 1$}\:. \]
It is verified immediately using the implication
\[ f \simeq^\ell_{\tilde{M}} 0 \qquad \Longrightarrow \qquad D_j f \simeq^{\ell-1}_{\tilde{M}} 0 \]
that this mapping is well-defined on the equivalence classes.
This implication also explains why~$\ell$ must be at least one,
and~$D_j f$ is formed of equivalence classes which vanish to the order~$\ell-1$.

\subsection{Differential Forms and Exterior Calculus} \label{secdiffex}
Given~$r \in \N_0$ and~$s \in \{0, \ldots, k\}$,
at each point~$p \in \tilde{M}$ we define an {\em{alternating $s$-form}}~$\omega_p$
on~$M_p$ as a mapping
\[ \omega_p \::\: \underbrace{M_p \times \cdots \times M_p}_{\text{$s$ factors}} \rightarrow
M_p^r \qquad \text{totally anti-symmetric} \:. \]
We now let~$\omega$ be a continuous
mapping which to every~$p \in \tilde{M}$ associates a corresponding alternating
$s$-form on~$M_p$. We denote the vector space of such mappings
by~$\Gamma^\infty_\loc(\tilde{M}, T^{s,r}\tilde{M})$.

Similar as in the previous section (see~\eqref{fdef1}), a section~$\eta \in
\Gamma^\infty_\loc(\tilde{M}, T^{s,r}\tilde{M})$ can be written in a parallel frame 
locally as
\beq \label{eta0}
\begin{split}
\eta(y) &= \eta_{i_1 \cdots i_s}(y)\: \big( e^{i_1}(y) \wedge \cdots \wedge e^{i_s}(y)
\big) \\
\text{with} \quad \eta_{i_1 \cdots i_s}(y) &= h_{i_1 \cdots i_s}^{j_1,\ldots, j_r}(y) \:
\big( e_{i_1}(y) \odot \cdots \odot e_{i_r}(y) \big) \:,
\end{split}
\eeq
where~$e^i(y)$ are the corresponding dual basis vectors in~$M_p^*$
as introduced in~\eqref{dualframe}.
Moreover, the wedge product is the totally anti-symmetrized tensor product,
\[ 
e^{1} \wedge \cdots \wedge e^{m} := \frac{1}{m!} \sum_{\sigma \in S_m}
(-1)^{\sign \sigma}\:
e^{\sigma(1)} \otimes \cdots \otimes e^{\sigma(m)} \:. \]

We now introduce differential forms as suitable equivalence classes of
the above tensor sections: 
We say that~$\eta \in \Gamma^\infty_\loc(\tilde{M}, T^{s,r}\tilde{M})$
{\em{vanishes at~$x$ to the order~$\ell \in \N_0$}}
if in the representation~\eqref{eta0} all its component
functions~$\eta_{i_1 \cdots i_s}(x)$ vanish to the order~$\ell$ at~$x$.
We use the notation
\beq \label{Lnot}
\L_{\tilde{\rho}}(D_\lambda \eta)(x) = 0 \qquad \text{for all~$\lambda$ with~$|\lambda| \leq \ell$} \:.
\eeq
Likewise, the differential form vanishes to the order~$\ell$ if it vanishes
to the order~$\ell$ at every point~$x \in \tilde{M}$.
Similar to~\eqref{fequiv}, we denote the corresponding equivalence
relation by~$\simeq^\ell_{\tilde{M}}$. Given~$m \in \N_0$, we define the space of {\em{differential
forms}}~$\Omega^{s,m,\ell}(\tilde{M})$ in analogy to~\eqref{Cequi} by
\beq \label{omegaequi}
\Omega^{s, m, \ell}(\tilde{M}) := \bigg(
\bigoplus_{r=0}^m
\Big( \Gamma^\infty_\loc \big( \tilde{M}, T^{s,r} \tilde{M} \big) \bigg) \Big/ \simeq^\ell_{\tilde{M}} \:.
\eeq

We want to define the exterior derivative as the totally anti-symmetrized
weak derivative. Thus for the section~$\eta \in
\Gamma^\infty_\loc(\tilde{M}, T^{s,r}\tilde{M})$ in~\eqref{eta0}, we introduce the section~$d \eta \in \Gamma^\infty_\loc(\tilde{M}, T^{s+1,r+1}\tilde{M})$ 
locally by
\beq \label{ddef}
(d \eta)(y) = 
h_{i_1 \cdots i_s}^{j_1,\ldots, j_m}(y) \:
\big( e_j(y) \odot e_{i_1}(y) \odot \cdots \odot e_{i_m}(y) \big)
\: \big( e^j(y) \wedge e^{i_1}(y) \wedge \cdots \wedge e^{i_s}(y) \big) \:.
\eeq
In order for this definition to be independent of the choice of representatives of the
equivalence classes, we need to assume that~$\ell \geq m+1$.
\begin{Def}
The {\bf{exterior derivative}} is defined by the local formula~\eqref{ddef}
as the mapping between differential forms
\beq \label{dmap}
d \::\: \Omega^{s,m,\ell}(\tilde{M}) \rightarrow \Omega^{s+1, m+1, \ell-1}(\tilde{M})
\qquad \text{if~$\ell \geq 1$}\:.
\eeq
\end{Def}

\subsection{De Rham Cohomology}
\label{sec:Cohomology}
By iterating~\eqref{ddef} and again using the symmetry
of second weak derivatives, one immediately sees that
\beq
d^2=0 \:.
\eeq
We thus obtain {\em{co-chain complexes}} of real vector spaces.
For the simplest choice of the parameters~$m$ and~$\ell$, one obtains
\beq \label{simplecomplex}
0 \longrightarrow \Omega^{0, 0, k}(\tilde{M})
\overset{d}{\longrightarrow} \Omega^{1,1,k-1}(\tilde{M}) 
\overset{d}{\longrightarrow} 
\cdots 
\overset{d}{\longrightarrow} \Omega^{k,k,0}(\tilde{M})
\longrightarrow 0
\eeq
(where~$k$ again denotes the dimension of~$M$).
This makes it possible to introduce the {\em{de Rham cohomologies}}~$H^r(\tilde{M})$
as usual by
\begin{align}
H^0(\tilde{M}) &:= \ker d|_{\Omega^{0, 0, k}(\tilde{M})} \label{H0}\\
H^r(\tilde{M}) &:= \ker d|_{\Omega^{r, r, k-r}(\tilde{M})} \big/ d\big( \Omega^{r-1, r-1, k-r+1}(\tilde{M}) \big) \:,\quad r \in \{1,\ldots, k-1\} \label{Hr}\\
H^k(\tilde{M}) &:= \Omega^{k, k, 0}(\tilde{M}) \big/ d\big( \Omega^{k-1, k-1, 1} (\tilde{M}) \big)\:\label{Hk}.
\end{align}
The dimensions of the cohomologies give the {\em{Betti numbers}}
\[ b_r := \dim H^r(\tilde{M}) \:. \]

\subsection{Integration of $k$-Forms and Stokes' Theorem} \label{secintegrate}
We choose an orientation of~$M$.
This also induces an orientation on the osculating vacuum~$M_p$.
We thus have a distinguished $k$-form~$\epsilon_\flat$ characterized by the condition
that the volume form induced on~$M_p$ coincides with the measure~$\rho_p|_{M_p}$.
The symbol~$\flat$ indicates that, choosing a basis of~$M_p$ and representing~$\epsilon_\flat$
in components~$\epsilon_\flat = (\epsilon_\flat)_{i_1 \cdots i_k}$, it has~$k$ lower
(i.e.\ covariant) indices. 
Similar to the differential forms~$\Omega^{s,m,l}(\tilde{M})$, we introduce alternating dual $s$-forms denoted by~$(\Omega_p^{s,m,l})^*(\tilde{M})$.
There is a distinguished dual $k$-vector~$\epsilon^\sharp$ characterized by the condition
\beq \label{epsnorm}
(\epsilon^\sharp)^{i_1 \cdots i_k}\: (\epsilon_\flat)_{i_1 \cdots i_k} = k! \:.
\eeq
Note that the above constructions did not involve a metric on~$M_p$, but merely the
volume measure~$\rho_p|_{M_p}$.

The {\em{integral of $k$-forms}} is defined as follows. Since the vector space of $k$-forms
is one-dimensional, any~$\omega \in \Omega^{k, m, \ell}(\tilde{M})$ can be represented uniquely as
\beq \label{fom}
\omega_p = f\, \epsilon_\flat
\eeq
with~$f \in T_p^m\tilde{M}$. Since it may involve weak derivatives,
we can integrate~$f$ only after evaluating weakly according to~\eqref{weval}.
This leads us to define
\[ \int_{\tilde{U}} \omega := \int_{\tilde{U}} \bigg( \fint_{\tilde{M}}
\L(x,y) \cdot f(y) \: d\tilde{\rho}(y) \bigg) \:d\tilde{\rho}(x) \:, \] 
where~$\tilde{U}$ is any Borel subset of~$\tilde{M}$. Next, we need to
define the boundary integral over~$\partial U$
of a differential form~$\eta \in \Omega^{k-1,m, \ell}(\tilde{M})$.
Similar as in the Gau{\ss} divergence theorem (Theorem~\ref{thmgauss}),
this boundary integral is defined as a suitable {\em{surface layer integral}}.
More precisely, for a representative~$\eta$ of a differential form
in~$\Omega^{k-1,m, \ell}(\tilde{M})$ we set
\beq \label{osik}
\int_{\partial \tilde{U}} \eta := -\frac{1}{\s} \bigg(
\int_{\tilde{U}} d\tilde{\rho}(x) \int_{\tilde{M} \setminus \tilde{U}} d\tilde{\rho}(y)
- \int_{\tilde{M} \setminus \tilde{U}} d\tilde{\rho}(x) \int_{\tilde{U}} d\tilde{\rho}(y) \bigg) \:(\epsilon^\sharp \llcorner \eta)^i \: D_{2,i} \L(x,y) \:,
\eeq
where the symbol~$\llcorner$ denotes the tensor contraction, i.e.\ in a local frame
\[ (\epsilon^\sharp \llcorner \eta)^i := \frac{1}{k!}\: (\epsilon^\sharp)^{i \,i_1 \cdots i_{k-1}}
\:\eta_{i_1 \cdots k_{k-1}} \:. \]

\begin{Thm} {\bf{(Stokes' theorem)}} \label{thmstokes}
For any relatively compact Borel set~$\tilde{U} \subset \tilde{M}$
and any $(k-1)$-form~$\eta \in \Omega^{k-1, 0, \ell}(\tilde{M})$ with~$\ell \geq 1$,
\beq \label{stokes}
\int_{\tilde{U}} d \eta = \int_{\partial U} \eta \:.
\eeq
\end{Thm} \noindent
We point out that this theorem applies only if~$\eta$ contains no
weak derivatives (i.e., only if~$\Omega^{k-1, m, \ell}(\tilde{M})$
with~$m=0$). We also note that the formula~\eqref{stokes} implies that
the surface layer integral~\eqref{osik} is well-defined on~$\Omega^{k-1, 0, \ell}(\tilde{M})$
(in the sense that it does not depend on the choice of representatives).

\Proof[Proof of Theorem~\ref{thmstokes}.]
We first note that, using~\eqref{epsnorm}, we can express the
function~$f$ in~\eqref{fom} as
\[ f = \epsilon^\sharp \llcorner \omega = \frac{1}{k!}\: (\epsilon^\sharp)^{i_1 \cdots i_k}\: \omega_{i_1 \cdots i_k}\:. \]
We introduce the abbreviation~$g^i := (\epsilon^\sharp \llcorner \eta)^i$.
Then a computation similar to that in the proof of Theorem~\ref{thmgauss} gives
\begin{align}
\int_{\tilde{U}} d \eta &= -
\frac{1}{\s} \int_{\tilde{U}} d\tilde{\rho}(x) \int_{\tilde{M}} d\tilde{\rho}(y)
D_{2,i_1} \L(x,y)\: g^i(y) \notag \\
&= - \frac{1}{\s} \int_{\tilde{U}} d\tilde{\rho}(x) \int_{\tilde{U}} d\tilde{\rho}(y) \:g^i(y)\:
D_{2,i_1} \L(x,y) \notag \\
&\quad\: - \frac{1}{\s} \int_{\tilde{U}} d\tilde{\rho}(x) \int_{\tilde{M} \setminus \tilde{U}} d\tilde{\rho}(y) \:g^i(y)\:
D_{2,i_1} \L(x,y) \notag \\
&= - \frac{1}{\s} \int_{\tilde{M}} d\tilde{\rho}(x) \int_{\tilde{U}} d\tilde{\rho}(y) \:g^i(y)\:
D_{2,i_1} \L(x,y) \label{stokes1} \\
&\quad\: + \frac{1}{\s} \int_{\tilde{M} \setminus \tilde{U}} d\tilde{\rho}(x) \int_{\tilde{U}} d\tilde{\rho}(y) \:g^i(y)\:
D_{2,i_1} \L(x,y) \label{stokes2} \\
&\quad\: - \frac{1}{\s} \int_{\tilde{U}} d\tilde{\rho}(x) \int_{\tilde{M} \setminus \tilde{U}} d\tilde{\rho}(y) \:g^i(y)\: D_{2,i_1} \L(x,y) \:. \label{stokes3} 
\end{align}
Changing the order of integration in~\eqref{stokes1}, we can apply the
EL equations~\eqref{EL1} for~$\tilde{\rho}$ to obtain zero.
Rewriting the remaining terms~\eqref{stokes2} and~\eqref{stokes3} 
gives the result.
\QED

\subsection{Restrictions of Differential Forms} \label{secrestrict}
So far, all differential forms were defined globally on all of~$\tilde{M}$.
The goal of this section is to ``localize'' differential forms to
measurable subsets of~$\tilde{M}$,
making glueing constructions applicable to our setting.
In view of the fact that the mollified evaluation is non-local, it is not obvious
how our goal can be achieved and what a ``localization'' should be.
The general idea is to take equivalence classes inside subsets, as we will now
explain step-by-step.

Let~$U \subset \tilde{M}$ be a measurable subset.
Similar to the notions introduced in Definition~\ref{defvanish}, we say that a
section~$f \in \Gamma^\infty_\loc \big( \tilde{M}, T^m \tilde{M} \big)$ vanishes in~$U$
if it vanishes to the order~$\ell$ at all~$x \in U$, i.e., using
again the notation~\eqref{Lnot},
\[ 
\L_{\tilde{\rho}}(D_\lambda \eta)(x) = 0 \qquad \text{for all~$x \in U$
and~$\lambda$ with~$|\lambda| \leq \ell$} \:. \]
Similar to~\eqref{fequiv},
we denote the resulting equivalence relation by~$\simeq^\ell_U$.
Clearly, the equivalence relation~$\simeq^\ell_U$ is coarser in the sense that
\[ f \simeq^\ell_{\tilde{M}} g \qquad \Longrightarrow \qquad f \simeq^\ell_U g \:. \]
Therefore, the equivalence relation~$\simeq^\ell_U$ can be applied to
the equivalence classes with respect to~$\simeq^\ell_{\tilde{M}}$.
This makes it possible to define the {\em{restriction of differential forms to~$U$}} by
\beq \label{restrict quotient} 
\Omega^{s, m, \ell} \big( U \subset \tilde{M} \big)
:= \Omega^{s, m, \ell}(\tilde{M}) \big/ \simeq^\ell_{U}
\eeq
We also define the {\em{restriction map}}~$|_U$ by
\beq \label{restrict}
|_U \::\: \Omega^{s, m, \ell}(\tilde{M}) \rightarrow \Omega^{s, m, \ell} \big( U \subset \tilde{M} \big)\:,\qquad \eta \mapsto \la \eta \ra_{\simeq^\ell_{U}} \:.
\eeq
The next lemma explains in which sense the restriction map is a local concept.
\begin{Lemma} {\bf{(Restriction lemma)}} \label{lemmarestrict}
For any Borel subset~$U \subset \tilde{M}$,
\[ \Omega^{s, m, \ell} \big( U \subset \tilde{M} \big)
= \Omega^{s, m, \ell} \big( U \subset B_\L(U) \big) \]
(where~$B_\L(U)$ is the $\L$-neighborhood as introduced in Definition~\ref{defBL}).
\end{Lemma}
\Proof Recall that in the equivalence relation~$\simeq^\ell_U$, one evaluates
the mollification integral in Definition~\ref{defvanish} only for~$x \in U$.
For every point~$y$ which is outside~$B_\L(U)$, the integrand in~\eqref{vanish}
vanishes, because the derivatives of the Lagrangian are zero.
Therefore, we may modify the differential form outside~$\B_\L(U)$ arbitrarily
without changing~$\L_{\tilde{\rho}}(\eta)|_{U}$. In other words,
\[ \Omega^{s, m, \ell} \big( U \subset \tilde{M} \big) 
\;\ni\; \eta \simeq^\ell_U \eta|_{B_\L(U)} 
\;\in\; \Omega^{s, m, \ell} \big( U \subset B_\L(U) \big) \:. \]
This gives the result.
\QED

The restriction maps are compatible with the exterior derivative. In particular, we conclude that~$d$ is a co-chain map:
\begin{Lemma}{\bf{(Co-chain Map lemma)}}
\label{lemma:d-chain-map}
    $ $\\
    Let~$\tilde{M}\supseteq U\supseteq V$ be a chain of measurable sets. The exterior derivative~$d$ commutes with the restriction map. In other words
    \begin{equation*}
        d\circ |_V=|_V\circ d \::\:
\Omega^{s, m, \ell} \big( U \subset \tilde{M} \big)  \rightarrow
\Omega^{s+1, m+1, \ell-1} \big( V \subset \tilde{M} \big) \:.
    \end{equation*}
\end{Lemma}
\Proof
    The proof consists of a diagram chase via the following diagram:
    \begin{equation*}
        \begin{tikzcd}[column sep = small]
            \bigoplus^m_{r=0}\Gamma\big(\tilde{M},T^{s,r}\tilde{M} \big) \arrow[rd, "q_U"]\arrow[rrr,"d"]\arrow[ddd,"\text{Id}"] & & & \bigoplus^{m+1}_{r=0}\Gamma \big( \tilde{M},T^{s+1,r}\tilde{M} \big) \arrow[ddd,"\text{Id}"]\arrow[ld, "Q_U"]\\
            &\Omega^{s,m,\ell}(U)\arrow[r,"d"]\arrow[d,"|_V"] & 
            \Omega^{s+1,m+1,\ell-1}(U)\arrow[d,"|_V"]\\
            &\Omega^{s,m,\ell}(V)\arrow[r,"d"] & 
            \Omega^{s+1,m+1,\ell-1}(V)\\
            \bigoplus^m_{r=0}\Gamma \big(\tilde{M},T^{s,r}\tilde{M}\big) \arrow[ru,"q_V"]\arrow[rrr,"d"] & & & \bigoplus^{m+1}_{r=0}\Gamma \big(\tilde{M},T^{s+1,r}\tilde{M} \big) \arrow[lu, "Q_V"]
        \end{tikzcd}
    \end{equation*}
    The outer square is trivially commutative. Furthermore, the whole diagram, except the inner square, is commutative by Definitions~\ref{ddef} and~\ref{restrict}. Thus it remains to show that this inner square commutes.
    
    Let~$[f]_U\in\Omega^{s,m,\ell}(U\subseteq\tilde{M})$, with lift~$f+\phi$ via~$q_U$ with~$\phi\simeq^\ell_U0$. By the commutativity of the top quadrant, we obtain~$Q_U\circ d(f+\phi)=d[f]_U$. Once again, we lift~$d[f]_U$ via~$Q_U$ to~$df+\psi$ where~$\psi\simeq^{\ell-1}_U0$. Now by commutativity of the left quadrant we obtain~$(d[f]_U)|_V=Q_V(f+\psi)$. However, as~$V\subseteq U$, we conclude~$\psi\simeq^{\ell-1}_V0$ as~$\psi\simeq^{\ell-1}_U$. As such, $Q_V(df+\psi)=Q_V(df)$.

Next, we apply the commutativity of the left quadrant to~$f+\phi$ to conclude~$q_V(f+\phi)=[f]|_V=[f]_V$. Lifting~$[f]_V$ up via~$q_V$, we find lift~$f+\psi^\prime$ for~$\psi^\prime\simeq^{\ell}_V0$. The commutativity of the bottom quadrant finally provides~$d([f]_V)=Q_V\circ d(f+\psi^\prime)$. But again, $\psi^\prime\simeq^{\ell}_V0$ implies that~$d\psi^\prime\simeq^{\ell-1}_V0$. Therefore, $Q_V\circ d(f+\psi^\prime)=Q_v\circ df$. In total, we conclude that
    \begin{equation*}
        d\big([f]_V \big)=Q_V(df)=\big(d[f]\big)|_V \:,.
    \end{equation*}
completing the proof.
\QED

\subsection{Extensions of Differential Forms} \label{secext}
Our next goal is to introduce an analog of the Mayer-Vietoris sequence.
To this end, we consider a covering of~$\tilde{M}$ by measurable subsets~$U$ and~$V$,
\beq \label{MUV}
\tilde{M} = U \cup V \:,
\eeq
and introduce the mappings
\begin{align*}
\sigma &\::\: \Omega^{s, m, \ell}(\tilde{M}) \rightarrow \Omega^{s, m, \ell}\big( U \subset \tilde{M} \big) \oplus \Omega^{s, m, \ell}\big( V \subset \tilde{M} \big) \\
\Delta &\::\: \Omega^{s, m, \ell}\big( U \subset \tilde{M} \big) \oplus \Omega^{s, m, \ell}\big( V \subset \tilde{M} \big) \rightarrow \Omega^{s, m, \ell}\big( (U \cap V)
\end{align*}
by
\begin{align*}
\sigma(\eta) &:= \big( \eta|_U, \eta|_V \big) \\
\Delta(\eta, \tilde{\eta}) &:= \eta|_{U \cap V} - \tilde{\eta}|_{U \cap V} \:.
\end{align*}
We now consider the short sequence
\beq \label{short}
\begin{split}
0 \rightarrow \Omega^{s, m, \ell}(\tilde{M}) \overset{\sigma}{\longrightarrow} 
\Omega^{s, m, \ell}\big( U \subset &\tilde{M} \big) \oplus 
\Omega^{s, m, \ell}\big( V \subset \tilde{M} \big) \\
&\overset{\Delta}{\longrightarrow} 
\Omega^{s, m, \ell}\big( (U \cap V) \subset \tilde{M} \big) \rightarrow 0 \:.
\end{split}
\eeq
The basic question is whether or under which assumptions this short sequence is
exact. We begin with the parts which are easy to prove.
\begin{Lemma} The mapping~$\sigma$ is injective. Moreover, the mapping~$\Delta$
is surjective.
\end{Lemma}
\Proof In order to show injectivity of~$\sigma$, suppose that~$\sigma(\omega)=0$. Then the
section~$\L_{\tilde{\rho}}(D^\kappa \eta)$ vanishes
in both~$U$ and in~$V$. Hence it vanishes in all of~$\tilde{M}$, implying that~$\omega=0$.

For the proof that~$\Delta$ is surjective, let~$\eta \in \Omega^{s, m, \ell}\big( (U \cap V) \subset \tilde{M} \big)$. We choose a representative~$\hat{\eta} \in \Omega^{s, m, \ell}\big(\tilde{M} \big)$. Then
\[ \Delta( \hat{\eta}|_U, 0 ) = \eta \:, \]
concluding the proof.
\QED

In order to establish exactness of the sequence~\eqref{short}, it remains to show
that every differential form in the kernel of~$\Delta$ lies in the image of~$\sigma$.
This is not true in general, as the following example shows.

\begin{Example} {\em{
We choose the space~$\tilde{M}=\{x_1,x_2\}\subset \F =\R$. Moreover, we choose a
Lagrangian~$\L \in \C^\infty(\R\times \R, \R^+_0)$ with the properties
\[ \L(x,y) = 1 = D_2\L(x,y) \qquad \text{for all~$x,y \in \tilde{M}$} \:, \]
whereas all higher derivatives of~$\L$ vanish on~$\tilde{M} \times \tilde{M}$.
By taking~$U=\{x_1\}$ and~$V=\{x_2\}$ we obtain a covering by measurable sets as in~\eqref{MUV}, and where~$U\cap V=\varnothing$.

It follows immediately from Definition~\ref{restrict quotient} that
\[ \Omega^{0,0,1}(\varnothing\subset\tilde{M})=0 \qquad \text{and}
\qquad \Omega^{0,0,1}(\tilde{M})=\R^2/\simeq^1_{\tilde{M}} \:. \]
The last equivalence relation means that~$\phi\simeq^1_{\tilde{M}}0$ if and only if 
the following two conditions hold,
\begin{align*}
    \int_{\tilde{M}}\L(x,y)\phi(y)dy=\phi(x_1)+\phi(x_2) &=0 \\
    \int_{\tilde{M}}D_2^1\L(x,y)\phi(y)dy=\phi(x_1)+\phi(x_2) &=0 \:.
\end{align*}
The solutions of this equations form the linear subspace of~$\R^2$ spanned by~$(1,-1)^T$.
Therefore,
\begin{equation*}
    \Omega^{0,0,1}(\tilde{M})=\frac{\R^2}{(1,-1)^T\R}\cong\R.
\end{equation*}
Finally, per symmetry we find that~$\Omega^{0,0,1}(U) \simeq \Omega^{0,0,1}(V)$,
being linear subspaces of~$\R=\Omega^{0,0,1}(\tilde{M})$. Filling this data into the sequence~\eqref{short} we obtain
\begin{equation}
\label{Mayer-Two-Points}
    \begin{tikzcd}
        0\arrow[r] & \R\cong\Omega^{0,0,1}(\tilde{M})\arrow[r,"\sigma"] & \Omega^{0,0,1}(U)^2 \arrow[r,"\Delta"] & \Omega^{0,0,1}(\varnothing)=0
    \end{tikzcd}
\end{equation}
However, as~$\Omega^{0,0,1}(U)$ is either isomorphic to~$\R$ or~$0$, we
obtain that the sequence~\eqref{Mayer-Two-Points} must be of one of the following two forms:
\begin{equation}
\label{Mayer-two-points}
    \begin{tikzcd}[row sep = tiny]
        0\arrow[r] & \R\arrow[r,"\sigma"] & \R^2 \arrow[r,"\Delta"] & 0\\
        0\arrow[r] & \R\arrow[r,"\sigma"] & 0 \arrow[r,"\Delta"] & 0
    \end{tikzcd}
\end{equation}
As both of the sequences in~\eqref{Mayer-two-points} fail to be exact, we conclude that general exactness is not obtained.
}} \QEDrem
\end{Example}

In order to rule out such cases, we need to impose a suitable condition.
\begin{Def} \label{defseparate} Two measurable subsets~$A, B \subset \tilde{M}$ are
{\bf{$\L$-separating}} the differential forms in~$\Omega^{s,m.\ell}(\tilde{M})$ if
\beq \label{Lsep}
\Omega^{s,m,\ell}(\tilde{M}) = \text{\rm{span}} \big( \ker |_A, \ker |_B \big)
\eeq
(where~$|_A$ and~$|_B$ are the restriction maps~\eqref{restrict}).
\end{Def} \noindent
This condition is necessary and sufficient for the existence of extensions,
as is made precise in the following lemma.

\begin{Lemma} {\bf{(Extension lemma)}}
Assume that the sets~$U \setminus V$ and~$V \setminus U$
are $\L$-separating the differential forms~$\Omega^{s,m.\ell}(\tilde{M})$.
Let~$\eta_U \in \Omega^{s,m.\ell}(U \subset \tilde{M})$ 
and~$\eta_V \in \Omega^{s,m.\ell}(V \subset \tilde{M})$ be two differential
forms which coincide on their joint domain,
\[ \eta_U |_{U \cap V} = \eta_V |_{U \cap V} \:. \]
Then there is a differential form~$\eta \in \Omega^{s,m.\ell}(\tilde{M})$ (the
so-called {\bf{extension}}) with
\[ \eta|_U = \eta_U \qquad \text{and} \qquad \eta|_V = \eta_V \:. \]

Conversely, if every~$\eta_U$ and~$\eta_V$ as above have an extension~$\eta$,
then the sets~$U \setminus V$ and~$V \setminus U$ are $\L$-separating.
\end{Lemma}
\Proof We choose representatives~$\hat{\eta}_U, \hat{\eta}_V \in \Gamma^\infty_\loc(\tilde{M}, T^{s,r}\tilde{M})$ of~$\eta_U$
and~$\eta_V$, respectively.
Our goal is to construct a representative~$\hat{\eta} \in \Gamma^\infty_\loc(\tilde{M}, T^{s,r}\tilde{M})$ of~$\eta$. To this end, we make the ansatz
\[ \hat{\eta} = X \hat{\eta}_U + (1-X) \hat{\eta}_V \]
with~$X$ a linear operator on~$\Gamma^\infty_\loc(\tilde{M}, T^{s,r}\tilde{M})$.
This operator can be regarded as a ``cutoff operator'' which interpolates between
the identity in~$U \setminus V$ and zero in~$V \setminus U$. More precisely, we demand the following: Using the abbreviation~$K
:= \Gamma^\infty_\loc(\tilde{M}, T^{s,r}\tilde{M})$,
the condition that~$X$ vanishes on~$V \setminus U$
means that
\beq \label{c1}
X : K \rightarrow \ker |_{V \setminus U} \:.
\eeq
The fact that~$X$ is the identity on~$U \setminus V$ can be expressed by
\beq \label{c2}
X = \1 + Y \qquad \text{with} \qquad Y : K \rightarrow \ker |_{U \setminus V} \:.
\eeq
Finally, in order to ensure that the action of~$X$ on~$U \cap V$ is defined independently
of the choice of representatives, we need to impose the condition that
it leaves the kernel of the restriction map~$|_{U \cap V}$ invariant,
\beq \label{c3}
X|_{\ker |_{U \cap V}} \::\:  \ker |_{U \cap V} \rightarrow \ker |_{U \cap V} \:.
\eeq

In order to verify if there is an operator~$X$ having the properties~\eqref{c1},
\eqref{c2} and~\eqref{c3}, we construct a representation of~$X$
in a suitable direct sum decomposition of~$K$.
In preparation, we choose the subspaces
\[ 
K_A := \ker |_{U \setminus V}\:,\qquad
K_B := \ker |_{U \cap V} \:,\qquad
K_C := \ker |_{V \setminus U} \:. \]
By Definition~\ref{defseparate}, $K_A$ and~$K_C$ span the whole
space~$K$.
But clearly, these subspaces do not need to be transverse.
We want to decompose~$K$ into a direct sum of subspaces. To this end, we set
\[ S_{ABC} := K_A \cap K_B \cap K_C \:. \]
Next, we choose~$S_{AB}$ to be a subspace which is transverse to~$S_{ABC}$
such that
\[ K_A \cap K_B = S_{ABC} \oplus S_{AB} \:. \]
Choosing the spaces~$S_{BC}$ and~$S_{AC}$ similarly, we obtain
\beq \label{subspace}
\text{span} \big( K_A \cap K_B, K_B \cap K_C, K_A \cap K_C \big)
= S_{ABC} \oplus S_{AB} \oplus S_{BC} \oplus S_{AC} \:.
\eeq
Finally, we choose subspaces~$S_A \subset K_A$ and~$S_C \subset \K_C$
which are transverse to the space~\eqref{subspace} and span the corresponding
subspaces. We thus get a direct sum
decomposition
\[ K = S_{ABC} \oplus S_{AB} \oplus S_{BC} \oplus S_{AC} \oplus S_A \oplus S_C \:. \]

The conditions~\eqref{c1}, \eqref{c2} and~\eqref{c3} can be satisfied by
choosing~$X$ as the block matrix
\beq \label{Xchoice}
X = \begin{pmatrix} 
* & * & * & * & * & * \\
0 & 0 & 0 & 0 & 0 & 0 \\
0 & 0 & 1 & 0 & 0 & 0 \\
0 & 0 & 0 & * & * & * \\
0 & 0 & 0 & 0 & 0 & 0 \\
0 & 0 & 0 & 0 & 0 & 1 \\
\end{pmatrix} \:,
\eeq
where the stars stand for arbitrary operators. This is verified in detail as
follows: By construction of the direct sum decomposition, the kernel of
the restriction map to~$V \setminus U$ consists of all subspaces carrying
a lower index~$C$, i.e.\
\[ \ker |_{V \setminus U} = 
S_{ABC} \oplus \{0\} \oplus S_{BC} \oplus S_{AC} \oplus \{0\} \oplus S_C \:. \]
Therefore, the condition~\eqref{c1} means that the second and fifth
row of the matrix~$X$ must be zero.
Likewise, the kernel of
the restriction map to~$U \setminus V$ consists of all subspaces carrying
a lower index~$A$, i.e.\
\[ \ker |_{U \setminus V} = 
S_{ABC} \oplus S_{AB} \oplus \{0\} \oplus S_{AC} \oplus S_A \oplus \{0\} \:. \]
Consequently, the condition~\eqref{c2} is implemented by
the fact that the third and last row of the matrix~$X$ coincides with
those of the identity matrix.
Finally, the kernel of
the restriction map to~$U \cap V$ consists of all subspaces carrying
a lower index~$B$, i.e.\
\[ \ker |_{U \cap V}
= S_{ABC} \oplus S_{AB} \oplus S_{BC} \oplus \{0\} \oplus \{0\} \oplus \{0\} \:. \]
Therefore, the condition~\eqref{c3} means that~$X$
leaves the subspace spanned by the first three components invariant,
as is the case for~\eqref{Xchoice}.

We conclude that the above choice of~$X$ has all
the desired properties~\eqref{c1}, \eqref{c2} and~\eqref{c3}.
This shows that the condition~\eqref{Lsep} is sufficient to ensure the
existence of an extension.

In order to show that the condition~\eqref{Lsep} is also necessary,
let us assume conversely that this condition does not hold.
Then, in the above direct sum decomposition of~$K$, we get an additional
direct summand~$K_B$ consisting of vectors which are neither in the kernel of~$|_{U \setminus V}$
nor in the kernel of~$|_{V \setminus U}$.
The corresponding additional row of the block matrix~$X$
must be zero (because of~\eqref{c1}), and it must also be the
corresponding row of the identity operator (because of~\eqref{c2}).
This is a contradiction, completing the proof.
\QED
The question remains how to verify whether two measurable subsets
are $\L$-separating. Here is a simple sufficient condition:
\begin{Lemma} \label{lemmasufficient}
Assume that for two measurable subsets~$A, B \subset \tilde{M}$,
the corresponding $\L$-neighborhoods are disjoint,
\beq \label{Ldisjoint}
B_\L(A) \cap B_\L(B) = \varnothing\:.
\eeq
Then~$A$ and~$B$ are $\L$-separating the differential
forms~$\Omega^{s,m.\ell}(\tilde{M})$ for all values of~$s$, $m$ and~$\ell$.
\end{Lemma}
\Proof Using multiplication by characteristic functions, we obtain a direct sum
decomposition
\[ \Gamma^\infty_\loc \big( \tilde{M}, T^{s,r}\tilde{M} \big)
= \Gamma^\infty_\loc \big( B_\L(A), T^{s,r}\tilde{M} \big) \oplus
\Gamma^\infty_\loc \big( \tilde{M} \setminus B_\L(A), T^{s,r}\tilde{M} \big) \:. \]
Taking equivalence classes, we obtain
\begin{align*}
\Gamma^\infty_\loc \big( B_\L(A), T^{s,r}\tilde{M} \big) \Big/ \simeq^\ell_B \;&=\; \{ 0 \} \\
\Gamma^\infty_\loc \big( \tilde{M} \setminus B_\L(A), T^{s,r}\tilde{M} \big) \Big/ \simeq^\ell_A \;&=\; \{ 0 \} \:.
\end{align*}
This gives the result.
\QED

We point out that the condition~\eqref{Ldisjoint} in Lemma~\ref{lemmasufficient} is sufficient but not necessary for the sequence~\eqref{short} to be exact. This is illustrated by the next example.

\begin{Example} {\em{
We choose~$\tilde{M}=\Z \subset \F=\R$ and denote the spatial points
by~$i,j \in \Z$. We choose a Lagrangian~$\L \in C^\infty_0(\R \times \R, \R^+_0)$ with the properties
    \begin{equation*}
        \L(i,j)=D_2\L(i,j)=
        \begin{cases}
            1 & \text{if~$i=j$ or~$i=j\pm1$} \\
            0 & \text{otherwise,}
        \end{cases}
    \end{equation*}
and assume that all higher order derivatives vanish (thus the Lagrangian implements
a ``nearest-neighbor interaction'').

We restrict attention to the case~$(s,m,\ell)=(0,0,1)$. Furthermore, we choose~$U=\Z_{\geq1}$ and~$V=\Z_{\leq1}$, so that~$U\cup V=\tilde{M}$. Since~$U\cap V=\{1\}$, 
the condition~\eqref{Ldisjoint} is violated.
Let~$(\hat{\alpha},\hat{\beta})\in \text{Ker}(\Delta)$.
Choosing two representatives~$\alpha,\beta$, it follows by definition that
    \begin{equation*}
        \alpha-\beta=:\gamma\simeq^1_{\{1\}}0,
    \end{equation*}
    for some global section~$\gamma$.
\beq
\text{\em{Claim:}} \quad
\text{We can decompose~$\gamma\simeq^1_{\tilde{M}}\gamma_1+\gamma_2$ such that 
$\gamma_1\simeq^1_U0$ and~$\gamma_V^1\simeq^1_V0$} \:. \label{claim: Sum}
\eeq
Before proving the claim, we show that it implies the exactness of the sequence~\eqref{short}.
Indeed, assuming~\eqref{claim: Sum}, we obtain that
    \begin{equation*}
        \alpha-\gamma_1\simeq^1_{\tilde{M}}\beta+\gamma_2.
    \end{equation*}
Thus
    \begin{align*}
        (\alpha-\gamma_1)|_{U}&=\hat{\alpha}-0=\hat{\alpha},\\
        (\alpha-\gamma_1)|_{V}&=(\beta+\gamma_2)|_V=\hat{\beta}+0=\hat{\beta},
    \end{align*}
    which yields~$\sigma(\alpha-\gamma_1)=(\hat{\alpha},\hat{\beta})$. 
It follows that~$\text{Ker}(\Delta)=\text{Im}(\sigma)$, proving exactness of the sequence~\eqref{short}.

\Felix{@Frank: I changed this, because I did not understand your write-up. Please double check!}%
We finally prove the claim~\eqref{claim: Sum}:
Per definition of~$\gamma\simeq^1_{\{1\}}0$, we know that
        \begin{equation*}
            0=\int_{\tilde{M}}\L(0,y)\:\gamma(y)\:dy=\gamma(0)+\gamma(-1)+\gamma(1).
        \end{equation*}
We define a global section~$\theta$ by first setting
\[ \theta(x)= \gamma(x) \qquad \text{if~$x=0, \pm1$} \:. \]
We now extend~$\theta$ inductively for~$x >1$ by the relation
\[ \theta(x)=-\theta(x-1)-\theta(x-2) \]
and likewise for~$x<-1$ by
\[ \theta(x)=-\theta(x+1)-\theta(x+2) \:. \]
By construction, $\theta\simeq^1_{\tilde{M}}0$. We now introduce
\[ 
            \gamma_1(x)=
            \begin{cases}
                \gamma(x) - \theta(x) & \text{if~$x\leq 0$} \\
                0 & \text{if~$x\geq1$}
            \end{cases}
\quad \text{and} \quad
            \gamma_2(x)=
            \begin{cases}
                0 & \text{if~$x\leq0$}\\
                \gamma(x) - \theta(x) & \text{if~$x\geq 1$}\:. \\
            \end{cases} \]
By construction, $\gamma_1$ vanishes on~$U$ and~$\gamma_2$ vanishes on~$V$. 
Moreover, one verifies by direct computation that~$\gamma_1 \simeq^1_{\{1\}}0$
and~$\gamma_2 \simeq^1_{\{1\}}0$. Moreover,
        \begin{align*}
            \gamma_1+\gamma_2=\gamma-\theta\simeq^1_{\tilde{M}}\gamma \:,
        \end{align*}
proving the claim.
        \QEDrem
    }}
\end{Example}

\subsection{The Mayer-Vietoris Sequence} \label{secmv}
Now we are in the position to formulate the adaptation of the Mayer-Vietoris
sequence to our setting. The {\em{exterior derivative}} for restrictions is introduced simply by
(again under the assumption~$\ell \geq m+1$)
\[ d \::\: \Omega^{s,m,\ell}(U \subset \tilde{M}) \rightarrow \Omega^{s+1, m+1, \ell-1}(U \subset \tilde{M}) \:,\qquad d \eta := (d \hat{\eta})|_U \:, \]
where~$\hat{\eta} \in \Omega^{s,m,\ell}(\tilde{M})$ is a representative of~$\eta$.
One verifies immediately that this definition is independent of the choice of
representatives. Moreover, the restriction map is compatible with the exterior derivative
in the sense that
\[ (d \eta)|_U = d \big(\eta|_U \big) \qquad \text{for all~$\eta \in 
\Omega^{s,m,\ell}(\tilde{M})$} \:. \]

\begin{Corollary} {\bf{(Mayer-Vietoris sequence)}} \label{cormayer-viet}
Consider a covering of~$\tilde{M}$ by measurable subsets~$U$ and~$V$~\eqref{MUV}.
Assume that the sets~$U \setminus V$ and~$V \setminus U$ are
{\bf{$\L$-separating}} the differential forms~$\Omega^{s,m.\ell}(\tilde{M})$ for all values of~$s$, $m$ and~$\ell$. Then there is the long exact sequence
\begin{equation*}
    \begin{tikzcd}[column sep = small]
        \dots \arrow[r]
        & H^n(\tilde{M}) \arrow[r]
        & H^n(U\subseteq\tilde{M})\oplus H^n(V\subseteq\tilde{M}) \arrow[r]
        \arrow[d, phantom, ""{coordinate, name=Z}]
        & H^n(V\cap U\subseteq\tilde{M}) \arrow[dll,
            "\partial",
            rounded corners,
            to path={ -- ([xshift=2ex]\tikztostart.east)
            |- (Z) [near end]\tikztonodes
            -| ([xshift=-2ex]\tikztotarget.west)
            -- (\tikztotarget)}] \\
        &
        H^{n+1}(\tilde{M}) \arrow[r]
        & H^{n+1}(U\subseteq\tilde{M})\oplus H^{n+1}(V\subseteq\tilde{M}) \arrow[r]
        & \dots
    \end{tikzcd}
\end{equation*}
\end{Corollary}
\Proof Follows exactly as the classical proof as given for example in~\cite[\S1.1]{bott-tu}.
\QED
We conclude that the main difference to the classical setting is that we need to
verify that the set~$U \setminus V$ and~$V \setminus U$ are $\L$-separating the
differential forms. One way of doing so is to check directly the condition
in Definition~\ref{defseparate}. Alternatively, one can try to get by with the sufficient
condition in Lemma~\ref{lemmasufficient}, making it necessary to choose the
sets~$U$ and~$V$ in such a way that their intersection~$U \cap V$ separates the
sets~$U \setminus V$ and~$V \setminus U$
by at least two times the range of the Lagrangian (see Figure~\ref{figseparate}).
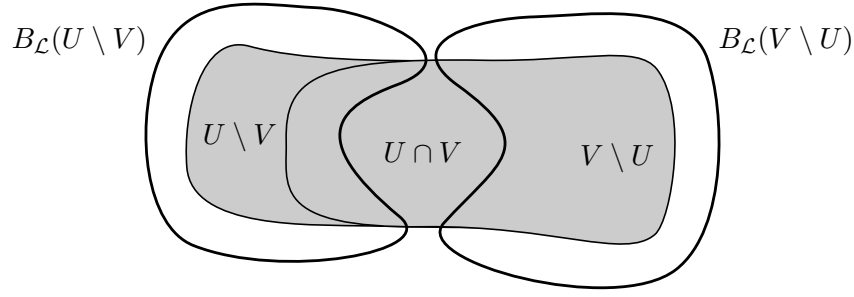
\begin{figure}[tb]
\psset{xunit=.4pt,yunit=.4pt,runit=.4pt}
\begin{pspicture}(541.98711245,269.57101416)
{
\newrgbcolor{curcolor}{0.80000001 0.80000001 0.80000001}
\pscustom[linestyle=none,fillstyle=solid,fillcolor=curcolor,opacity=0.5]
{
\newpath
\moveto(98.43421606,228.69244535)
\curveto(53.47223055,240.82811905)(35.61822614,166.00924912)(39.85003465,114.40848188)
\curveto(44.63109165,56.11044094)(155.81996598,60.70901479)(211.75915465,59.55279306)
\curveto(408.07988031,55.49499401)(371.40701102,55.4697392)(386.41986142,122.58844912)
\curveto(392.32049764,148.96872377)(424.61663622,215.19244535)(291.49574173,214.87312062)
\curveto(215.38300346,214.69053164)(146.98230047,215.58885353)(98.43421606,228.69244535)
\closepath
}
}
{
\newrgbcolor{curcolor}{0 0 0}
\pscustom[linewidth=1.51181105,linecolor=curcolor]
{
\newpath
\moveto(98.43421606,228.69244535)
\curveto(53.47223055,240.82811905)(35.61822614,166.00924912)(39.85003465,114.40848188)
\curveto(44.63109165,56.11044094)(155.81996598,60.70901479)(211.75915465,59.55279306)
\curveto(408.07988031,55.49499401)(371.40701102,55.4697392)(386.41986142,122.58844912)
\curveto(392.32049764,148.96872377)(424.61663622,215.19244535)(291.49574173,214.87312062)
\curveto(215.38300346,214.69053164)(146.98230047,215.58885353)(98.43421606,228.69244535)
\closepath
}
}
{
\newrgbcolor{curcolor}{0.80000001 0.80000001 0.80000001}
\pscustom[linestyle=none,fillstyle=solid,fillcolor=curcolor,opacity=0.5]
{
\newpath
\moveto(318.24164409,215.43427086)
\curveto(160.3176189,218.31538582)(131.96579906,197.3907118)(133.31788724,140.07754771)
\curveto(134.30343685,98.301617)(117.12365858,57.04365542)(263.26047874,58.77199558)
\curveto(385.37406614,60.21621731)(452.01337323,27.05188723)(480.37472126,51.57639495)
\curveto(496.26181417,65.3142444)(512.26508976,167.93701983)(482.25503622,206.87889826)
\curveto(463.66357795,231.00370204)(409.36042205,213.77192881)(318.24164409,215.43427086)
\closepath
}
}
{
\newrgbcolor{curcolor}{0 0 0}
\pscustom[linewidth=1.51181105,linecolor=curcolor]
{
\newpath
\moveto(318.24164409,215.43427086)
\curveto(160.3176189,218.31538582)(131.96579906,197.3907118)(133.31788724,140.07754771)
\curveto(134.30343685,98.301617)(117.12365858,57.04365542)(263.26047874,58.77199558)
\curveto(385.37406614,60.21621731)(452.01337323,27.05188723)(480.37472126,51.57639495)
\curveto(496.26181417,65.3142444)(512.26508976,167.93701983)(482.25503622,206.87889826)
\curveto(463.66357795,231.00370204)(409.36042205,213.77192881)(318.24164409,215.43427086)
\closepath
}
}
{
\newrgbcolor{curcolor}{0 0 0}
\pscustom[linewidth=2.64566925,linecolor=curcolor]
{
\newpath
\moveto(161.35985764,264.89766991)
\curveto(109.63618772,267.18005857)(35.45620535,284.93818771)(10.49506394,207.4305392)
\curveto(-8.23107402,149.28340346)(0.82479118,58.94612975)(46.07673827,41.09303243)
\curveto(129.2309178,8.28649889)(280.07382047,35.04159306)(240.2048315,68.10366046)
\curveto(148.52586331,144.13011968)(184.03896945,172.36675464)(238.12571339,192.59455558)
\curveto(312.7959874,220.52035842)(213.12696945,262.61336503)(161.35985764,264.89766991)
\closepath
}
}
{
\newrgbcolor{curcolor}{0 0 0}
\pscustom[linewidth=2.64566925,linecolor=curcolor]
{
\newpath
\moveto(367.46167559,256.30589668)
\curveto(314.06012598,255.05817637)(237.1223622,233.2567729)(294.54978898,191.96044912)
\curveto(373.97671181,134.84416062)(333.53557795,119.16863432)(281.91225827,62.53338141)
\curveto(248.39506772,25.76213857)(442.63466457,-25.89083151)(500.0224063,18.40158046)
\curveto(549.6312,56.69011464)(546.2167748,167.84090267)(529.61693858,216.9387118)
\curveto(510.37090394,273.86329511)(411.28129134,257.32973668)(367.46167559,256.30589668)
\closepath
}
\rput[bl](55,130){$U \setminus V$}
\rput[bl](410,110){$V \setminus U$}
\rput[bl](225,120){$U \cap V$}
\rput[bl](-125,220){$B_\L(U \setminus V)$}
\rput[bl](540,220){$B_\L(V \setminus U)$}
}
\end{pspicture}
\caption{Two sets which satisfy the sufficient condition of Lemma~\ref{lemmasufficient}.}
\label{figseparate}
\end{figure}%
It seems that, in many topological applications, coverings can be chosen
where this condition is satisfied for any two subsets of the covering.

\subsection{A Künneth Formula}\label{sec:Kunneth}
There are several classical computational methods for the de Rham cohomology, a particularly useful being the Künneth formula. In the classical case, this formula relates the cohomology of a product of spaces to the product of cohomologies of the component spaces. In order to construct a similar notion for causal variational principles, we begin with the notion of products of causal variational principles.

\begin{Def} \label{def:cfs-prod}
Let~$(\mathcal{F}, \L_\F, \rho)$ and~$(\mathcal{G}, \L_{\mathcal{G}}, \lambda)$ be two causal variational principles. The {\bf{product causal variational principle}} is defined by the triple~$(\mathcal{F}\times\mathcal{G},
\L_\F \cdot \L_{\mathcal{G}}, \rho\times\lambda)$.
\end{Def} \noindent
Denoting the component spaces by~$\tilde{M}:= \supp \rho$
and~$\tilde{N} := \supp \mu$, the product space is given by
\[ \supp (\rho\times\lambda) = \tilde{M} \times \tilde{N} \:. \]
We assume that the component measure both satisfy the EL equations.
In particular, for all~$x \in \tilde{M}$ and~$x' \in \tilde{N}$,
\beq \label{compEL}
\fint_{\tilde{M}} \L_\F(x,y)\:d\rho(y) = 1 \qquad \text{and} \qquad
\fint_{\tilde{N}} \L_{\mathcal{G}}(x',y')\: d\lambda(y') = 1 \:.
\eeq

As is the case for products of smooth manifolds, a product of causal variational principles comes with projection maps onto the product components.
These projection maps provide well-defined notions of pullbacks of differential forms:

\begin{Lemma} {\bf{(Pullbacks along projections)}} \label{lemma:pullback-projections}
Let~$\pi:\mathcal{F}\times\mathcal{G}\rightarrow\mathcal{F}$ be the projection map, and let~$\eta\in\Omega^{s,m,\ell}(\tilde{M})$. Then the pullback~$\pi^*(\eta)\in\Omega^{s,m,\ell}(\tilde{M}\times\tilde{N})$ is well-defined. 
\end{Lemma}

\begin{proof}
Given~$(s,m,\ell)\in\N_0^3$, we turn to the following commutative diagram,
    \begin{equation}\label{diag:Commutative-kunneth}
        \begin{tikzcd}
            \bigoplus^m_{r=0}\Gamma_{\loc}^\infty \big(\tilde{M},T^{s,r}\tilde{M} \big) \arrow[r,"\pi^*"]\arrow[d,"Q"] & 
            \bigoplus^m_{r=0}\Gamma^\infty_{\loc} \big(\tilde{M}\times\tilde{N},T^{s,r}\tilde{M}\times\tilde{N} \big) \arrow[d,"q"]\\
            \Omega^{s,m,\ell}(\tilde{M}) \arrow[r,dashed] &
            \Omega^{s,m,\ell}(\tilde{M}\times\tilde{N})
        \end{tikzcd}
    \end{equation}
    The goal of this proof is to construct the dashed arrow in this diagram so that the diagram commutes. To this end, we begin by constructing the top arrow~$\pi^*$.
    
    The projection map~$\pi$ induces a pullback operator on the level of tensor products in the following manner. In a parallel frame~$e_i(q)$ around~$p\in\tilde{M}$, with corresponding dual frame~$e^i(q)$ (as introduced in~\eqref{parallelframe} and~\eqref{dualframe}), we obtain the pullbacks of vectors
    \begin{align*}
        \pi^*(e_i)(p,q)&=(e_i\circ\pi)(p,q)=e_i(p),\\
        \pi^*(e^i)(p,q)&=(e^i\circ\pi)(p,q)=e^i(p).
    \end{align*}
    Additionally, given a continuous function on~$\tilde{M}$, the pullback along~$\pi$ is defined by pre-composition. Together, this induces the local description of the pullback of tensor products
    \begin{align}\label{eq:pullback-projection}
        \pi^*& \big( f^{i_1,\dots,i_s}_{j_1,\dots,j_m} \: e_{i_1}\odot\dots\odot e_{i_s} \:e^{j_1}\wedge\dots\wedge e^{j_m} \big)(p,q) \\
        &= f^{i_1,\dots,i_s}_{j_1,\dots,j_m}(p) \:e_{i_1}(p)\odot\dots\odot e_{i_s}(p) \:e^{j_1}(p)\wedge\dots\wedge e^{j_m}(p).
    \end{align}

In order to obtain the dashed arrow, it remains to show that~$\pi^*$ respects mollification with the Lagrangian.
    To this end, let~$\eta\in\Omega^{s,m,\ell}(\tilde{M})$ such that~$\pi^*(\eta)\simeq^\ell_{\tilde{M}\times\tilde{N}}0$. Using Einstein summation and multi-indices, we may write~$\eta$ in a local parallel frame as~$\eta^I_J e_I e^J$. Definition~\ref{defvanish} and~\eqref{eq:pullback-projection} now provide that
    \begin{align*}
        0&=\fint_{\tilde{M}} d\rho_{\tilde{M}}(y_{\tilde{M}}) \fint_{\tilde{N}}d\rho_{\tilde{N}}(y_{\tilde{N}})\;
        D_2^{\kappa}\L_{\F}(x_{\tilde{M}},y_{\tilde{M}}) \;\eta^I_J(y_{\tilde{M}})\:e_I(y_{\tilde{M}}) \:D^{\iota}_2\L_{\mathcal{G}}(x_{\tilde{N}},y_{\tilde{N}})\\
        &=\fint_{\tilde{M}}D^{\kappa+I}_2\L_{\F}(x_{\tilde{M}},y_{\tilde{M}})\:\eta^I_J(y_{\tilde{M}})\:d\rho_{\tilde{M}}\cdot\fint_{\tilde{N}}D^{\iota}_2\L_{\mathcal{G}}(x_{\tilde{N}},y_{\tilde{N}})\:d\rho_{\tilde{N}} \:,
    \end{align*}
where~$\kappa$ and~$\iota$ are arbitrary multi-indices with~$|\kappa|+|\iota|\leq\ell$.
Restricting attention to the case~$\iota=0$ and using~\eqref{compEL},
we obtain
    \begin{equation*}
        0=\fint_{\tilde{M}}D^{\kappa+I}_2\L_{\F}(x_{\tilde{M}},y_{\tilde{M}})\eta^I_J(y_{\tilde{M}})\:d\rho_{\tilde{M}}
    \end{equation*}
    for all~$|\kappa|\leq\ell$. We conclude that~$\eta\simeq^\ell_{\tilde{M}}0$. 

    As such, the universal property of quotients of vector spaces provides the existence of a unique dashed arrow, making Diagram~\ref{diag:Commutative-kunneth} commute. Abusing notation, we shall denote this arrow by~$\pi^*$, providing the required map.
\end{proof}

As in the topological setting, the graded cohomology groups of spaces can be related to the graded cohomology group of their product. This much is achieved by the following cross product.

\begin{Def}\label{def:cross-prod}
For~$\theta\in\Omega^{s,m,\ell}(\tilde{M})$ and~$\psi\in\Omega^{t,n,p}(\tilde{N})$,
the {\bf{cross product}} is defined as
    \begin{equation*}
        \theta\times\psi=\pi_{\tilde{M}}^*(\theta)\wedge\pi_{\tilde{N}}^*(\psi)
        \;\in\;\Omega^{s+t,m+t,\min (\ell, p)}(\tilde{M}\times\tilde{N}) \:,
    \end{equation*}
    where~$\wedge$ denotes the wedge product.
\end{Def} \noindent
Locally, this coincides with
\[ 
\theta^I_J e_I e^J\times\psi^K_L e_K e^L= \big(\theta^I_J \psi^K_L \big)(e_{I} \odot e_L) \:e^J\wedge e^L \:. \]

As the cross product is defined using a wedge product, the standard
properties of the cross product carry over to the wedge product.
\begin{Prp} \label{prop:cross-prod-prop}
For~$\theta\in\Omega^{s,m,\ell}(\tilde{M})$ and~$\psi\in\Omega^{t,m,p}(\tilde{N})$,
the wedge product has the following properties:
\bitem
\item[{\rm{(i)}}] $\times$ is well-defined,
\item[{\rm{(ii)}}] $\times$ is an $\R$-linear operator,
\item[{\rm{(iii)}}] $\times$ is an injective operator,
\item[{\rm{(iv)}}] $d(\theta\times\psi)=d(\theta)\times\psi+(-1)^{s}\:\theta\times d(\psi)$.
\eitem
\end{Prp}
\begin{proof}
Parts (i) and (ii) are a direct consequence of Lemma~\ref{lemma:pullback-projections} and the fact that the wedge product is well-defined and $\R$-linear.

For part~(iii) we turn to the following calculation. Suppose that~$\theta\times\psi\cong^\ell_{\tilde{M}\times\tilde{N}}0$ holds.
This is equivalent to the conditions that for all~$x\in\tilde{M}\times\tilde{N}$
and all multi-indices~$A$ and~$B$ with~$|A+ B| \leq \ell$,
    \begin{align*}
        0&=\fint_{\tilde{M}\times\tilde{N}}D^{A + B}\L(x,y)(\theta\times\psi)(y)\:d(\rho \times \lambda)(y) \\
&=\fint_{\tilde{M}\times\tilde{N}} 
D^{A+I}\L_{\F}(x_{\tilde{M}},y_{\tilde{M}})\: D^{B+K}\L_{\mathcal{G}}(x_{\tilde{N}},y_{\tilde{N}})\:\theta^I_J(y_{\tilde{M}})\: \psi^K_L(y_{\tilde{N}})
\:d(\rho \times \lambda)(y_{\tilde{M}}, y_{\tilde{N}}) \\
        &=\fint_{\tilde{M}}D^{A+I}\L_{\F}(x_{\tilde{M}},y_{\tilde{M}})\:\theta^I_J(y_{\tilde{M}})\:d\rho(y_{\tilde{M}})\cdot
        \fint_{\tilde{N}}D^{B+K}\L_{\mathcal{G}}(x_{\tilde{N}},y_{\tilde{N}})\:\psi^K_L(y_{\tilde{N}})\: d\lambda(y_{\tilde{N}}) \\
        &=\L_{\F, \rho}\theta \cdot \L_{\mathcal{G}, \lambda}\psi \:,
    \end{align*}
where in the last line we used a notation similar to~\eqref{Lrhodef}.
This implies that for every~$x\in\tilde{M}\times\tilde{N}$, either~$\L_{\F, \rho}\theta$ or~$\L_{\mathcal{G}, \lambda}\psi$ vanishes. As a consequence, either the function~$\L_{\F, \rho}\theta$ or~$\L_{\mathcal{G}, \lambda}\psi$ is identically zero.
Indeed, if conversely both functions were not identically zero.
Then there were~$x \in \tilde{M}$ and~$y \in \tilde{N}$ with~$\L_{\F, \rho}\theta(x)\neq0$ and~$\L_{\mathcal{G}, \lambda}\psi(y)\neq0$. As a consequence,
$\L_{\F, \rho}\theta(x)\cdot\L_{\mathcal{G}, \lambda}\beta(y)\neq0$, a contradiction.
We conclude that the cross product is injective.

    Part~(iv) is verified locally using Definition~\ref{ddef}. We obtain 
    \begin{align*}
        d \big(\theta^I_J \, e_I \, e^J\times\psi^K_L \, e_K \, e^L \big)&=d \Big( \big( \theta^I_J\, \psi^K_L\, \big)(e_{I}\odot e_L)\: e^J\wedge e^L \Big) \\
        &=\big(\theta^I_J \:\psi^K_L\big)(e_{I}\odot e_L\odot e_i)e^i\wedge e^J\wedge e^L\\
        &=\sum_{e^i\in M^\times}\big(\theta^I_J \:\psi^K_L\big)(e_{I}\odot e_L\odot e_i)e^i\wedge e^J\wedge e^L \\ 
        &+\sum_{e^i\in N^\times} \big( \theta^I_J\: \psi^K_L\big)(e_{I}\odot e_L\odot e_i)e^i\wedge e^J\wedge e^L \\
        &=d(\theta)\times\psi+(-1)^{|J|}\: \theta\times d(\psi) \:. 
    \end{align*}
    As~$|J|=s$, the result follows.
\end{proof}
\begin{Remark}\label{Remark:Kunneth}
    It should be noted that parts~{\rm{(iii)}} and~{\rm{(iv)}} of Proposition~\ref{prop:cross-prod-prop} provide that~$\times$ is an injective co-chain map
\[ 
        \times: \Omega^{*}(\tilde{M})\otimes\Omega^{*}(\tilde{N})\rightarrow \Omega^*(\tilde{M}\times\tilde{N}) \:. \]
\end{Remark}

In the standard smooth setting, the cross product induces a connection between the differential forms on the product and those on the component spaces. This culminates in the Eilenberg-Zilber theorem. Here, if we restrict analytical and geometric structures, a similar result is obtained. 
\begin{Lemma} \label{lemma:Eilenberg-Zilber}
    Let~$\Omega^*_c(\tilde{M}_i)$ denote the compactly supported differential forms on countable and discrete spaces~$\tilde{M}$ and~$\tilde{N}$. Then we obtain the following isomorphism,
    \begin{align*}
        \Omega^*_c(\tilde{M})\otimes\Omega^*_c(\tilde{N})&\cong\Omega^*_c(\tilde{M}\times\tilde{N}) \:, \\
        \theta\otimes\psi&\mapsto\theta\times\psi \:. 
    \end{align*}
\end{Lemma}
\begin{proof}
    First, note that when restricting to compactly supported differential forms, the morphism~$\times$ remains well-defined. Indeed, if~$\theta$ and~$\psi$ are compactly supported, then so is~$\theta\times\psi$.

    As noted in Remark~\ref{Remark:Kunneth}, the map~$\times$ is an injective chain map. As such, it remains to prove that it is also surjective. Let~$\omega\in\Omega^*(\tilde{M}\times\tilde{N})$. Using a parallel frame~$(e_i)_{i=1}^N$, we may write 
    \begin{align*}
        \omega=\alpha^i\mathds{1}_i\times\beta^j\mathds{1}_j\cdot e_{I_1}\odot e_{I_2}\cdot e^{J_1}\wedge e^{J_2} \:,
    \end{align*}
and the sums are finite because~$\omega$ is compactly supported (here~$I_1,J_1$ and~$I_2,J_2$ are multi-indices in coordinates of~$\tilde{M}$ and~$\tilde{N}$, respectively). In total, we may split~$\omega$ into a finite sum of products of two components
    \begin{equation*}
        \omega=\big(\alpha^i\mathds{1}_i\cdot e_{I_1}\cdot e^{J_1}\big)\times \big(\beta^j\mathds{1}_j\cdot e_{I_2}\cdot e^{J_2} \big).
    \end{equation*}
    We conclude that~$\omega\in\text{Im}(\times)$, and so~$\times$ is surjective.
\end{proof}

In general, an injective chain map need not induce an injective map on the level of cohomology. However, under the restrictions stated in the last lemma, we obtain the following result.

\begin{Thm} {\bf{(Künneth Formula)}} \label{lemma:kunneth-injection}
Let~$\tilde{M}$ and~$\tilde{N}$ be discrete and countable spaces, and denote by~$H_c^*(\tilde{M})$ the compactly supported differential forms. We obtain the following isomorphism of graded vector spaces
\begin{equation*}
    \Phi:H^*_c(\tilde{M})\otimes H^*_c(\tilde{N})\rightarrow H^*_c(\tilde{M}\times\tilde{N}),
\end{equation*}
defined by~$[\theta]\otimes[\psi]\mapsto[\theta\times\psi]$.
\end{Thm}

Before coming to the proof, we remark that the restrictive nature of Theorem~\ref{lemma:kunneth-injection}
comes from the restrictions in Lemma~\ref{lemma:Eilenberg-Zilber}. Indeed, the classical  Eilenberg-Zilber theorem encodes the analytic relation between the differential forms on the product, and those on the component spaces. This is the analytic part of the classical proof of Künneth's formula. However, as we have not endowed the space of differential forms with any functional analytic structure, we need to restrict to a situation where the algebraic tensor product coincides with the analytic tensor product. 
 It is likely that, by endowing the vector spaces of differential forms with
 additional structures (like a suitable topology or metric structure),
more general Künneth formulas could be proven. However, this is beyond
the scope of the present paper.

Here we merely remark that the introduction of additional functional analytic structures in order to prove a Künneth formula in the non-smooth setting is in line with other~$L^\infty$ cohomology theories for non-smooth spaces. For example, in bounded cohomology several functional analytical issues occur, such as~$\text{Im}(d)$ not being closed, which trouble a well-defined homological algebra required for a Künneth formula. A more extensive overview of several analytical issues can be found in~\cite{parkheeesook}.
Moreover, in~$L^\infty$ cohomology for pseudo-manifolds
as introduced in~\cite{cheeger},
 a Künneth formula depends strongly on additional structures
 such as assumptions on perversion functions.
 Under these additional assumptions, the~$L^\infty$ cohomology coincides with the intersection cohomology~\cite{valette}, and the perversion allows for a Künneth formula~\cite{cohen-goresky}.

In contrast to the analytic part, the algebraic part of the classical proof is still valid. This is so, as Lemma~\ref{lemma:Künneth-theorem} is a general statement holding for any co-chain complex of real vector spaces. Therefore, the failure of Lemma~\ref{lemma:Künneth-theorem} is fully encoded within the restrictive nature of Lemma~\ref{lemma:kunneth-injection}.

Before coming to the proof of Theorem~\ref{lemma:kunneth-injection},
we recall the following general result.
    \begin{Lemma} {\bf{(Künneth Theorem)}}
    \label{lemma:Künneth-theorem}
        Given two complexes of vector spaces~$C^*,D^*$ over~$\R$, there is the following isomorphism,
        \begin{align*}
            H^*(C^*)\otimes H^*(D^*)&\cong H^*(C^*\otimes D^*)\\
            [\theta]\otimes[\psi]&\mapsto[\theta\otimes\psi],
        \end{align*}
        where~$\otimes$ denotes the graded tensor product.
    \end{Lemma} \noindent
A more detailed statement and proof of this lemma can be found in~\cite[Theorem~2.1]{hilton-stammbach}.

\begin{proof}[Proof of Theorem~\ref{lemma:kunneth-injection}.]
    This proof mirrors the standard proof of the Künneth formula, consisting of two steps.
    First, Lemma~\ref{lemma:Eilenberg-Zilber} provides that
    \begin{equation*}
        \Omega^*_c(\tilde{M})\otimes\Omega^*_c(\tilde{N})\cong\Omega^*_c(\tilde{M}\times\tilde{N}).
    \end{equation*} 
Subsequently, we apply the classical Künneth theorem
as stated in Lemma~\ref{lemma:Künneth-theorem}.
    Combining these two steps, we obtain
    \begin{align*}
        H^*_c(\tilde{M}\times\tilde{N})&=H^*(\Omega^*_c(\tilde{M}\times\tilde{N}))\\
        &\cong H^*(\Omega^*_c(\tilde{M})\otimes\Omega^*_c(\tilde{N}))\\
        &\cong H^*(\Omega^*_c(\tilde{M}))\otimes H^*(\Omega^*_c(\tilde{N}))
        =H^*_c(\tilde{M})\otimes H^*_c(\tilde{N}) \:,
    \end{align*}
concluding the proof.
\end{proof} 

The need for the aforementioned additional structures on the set of differential forms already becomes clear for discrete and countable spaces,
as is illustrated by the following counter example to the Künneth formula.

\begin{Example}\label{ex:kunneth-fail} {\em{
Let~$\tilde{M} = \tilde{N} = \Z \subset \F = {\mathcal{G}} := \R$.
We choose the Lagrangian~$\L \in C^\infty(\R \times \R, \R+_0)$ such that
\[ \L(x,y) = \delta_{x,y} \quad \text{for all~$x,y \in \Z$} \:, \]
and all derivatives of~$\L$ vanish on~$\Z \times \Z$.
We also take the product
\[ \tilde{M} \times \tilde{N} = \Z^2 \subset \R^2 \:. \]
We denote the Lagrangian on the product by~$\L_{\F \times \F}$.

As~$\L_{\F \times \F}$ only relates a point~$x \in \Z^2$ with itself, we conclude that~$\cong^\ell_{\Z^2}$ is a trivial equivalence relation. Therefore, $\Omega^{0,0,\ell}(\Z^2)=C(\Z^2,\R)$. Applying the same argument to~$\tilde{M}=\Z$, provides that~$\Omega^{0,0,\ell}(\Z)=C(\Z,\R)$.

If we now assume that~$\ell=1$, we obtain the following differential complex,
\begin{equation*}
    \begin{tikzcd}
        0\arrow[r] &
        \Omega^{0,0,1}(\Z^2) \arrow[r,"d"] &
        \Omega^{1,1,0}(\Z^2) \arrow[r] & 0 \:.
    \end{tikzcd}
\end{equation*}
As~$D^\kappa\L=0$ for all multi-indices~$\kappa$, we find~$d=0$. We conclude that
\begin{equation*}
    0\not\cong^\ell_{\Z^2}\delta^n_m\in \Omega^{0,0,\ell}(\Z^2)= \text{Ker}(d)=H^0(\Z^2).
\end{equation*}
Supposing that Lemma~\ref{lemma:kunneth-injection} is true for non compactly supported differential forms, we would obtain that
\[ 
    C(\Z^2,\R)=\Omega^{0,0,1}(\Z^2)\cong\Omega^{0,0,1}(\Z)\otimes\Omega^{0,0,1}(\Z)=C(\Z,\R)\otimes C(\Z,\R) \:. \]
However, this formula is not true because
\[ \delta^n_m\notin C(\Z,\R)\otimes C(\Z,\R) \qquad \text{while} \qquad \delta^n_m\in C(\Z^2,\R) \:. \]
This can be seen in detail as follows. For a finite sum of tensor products, we may write
\begin{equation*}
    F(n,m)=\sum_{j=1}^k f_j(n)\cdot g_j(m).
\end{equation*}
When projecting this function onto the rows of~$\Z^2$, we obtain for the~$n^\text{th}$ row a function
\begin{equation}\label{eq:fin-span}
    F_n(-)=F(n,-)=\sum_{j=1}^k f_j(n)\cdot g_j(-).
\end{equation}
In turn, we conclude that~$F_n\in \text{Span}(g_1,\dots,g_k)=\Lambda$ for all~$n\in\Z$. Since~$\Lambda\subset C(\Z,\R)$ is a finitely generated linear subspace, it has a finite dimension bounded above by~$k$.

Applying the same projection procedure to~$\delta^n_m$, we obtain for the~$p^\text{th}$ row
\begin{equation}\label{eq:infinite-span}
    (\delta^n_m)_p=\delta^p_m,
\end{equation}
the Kronecker delta with~$p$ fixed. Since the Kronecker delta functions are linearly independent, these projections span an infinite dimensional subspace of~$C(\Z,\R)$.

Assuming that~$\delta^n_m\in C(\Z,\R)\otimes C(\Z,\R)$, we may write
\begin{equation*}
    \delta^n_m=\sum_{j=1}^k f_j(n)\cdot g_j(m).
\end{equation*}
Per~\eqref{eq:fin-span} the row projections of~$\delta^n_m$ span a finite dimensional subspace of~$C(\Z,\R)$. However, the relation~\eqref{eq:infinite-span} allows us to conclude that the row projections of~$\delta^n_m$ span an infinite dimensional subspace of~$C(\Z,\R)$. This is a contradiction. 

We conclude that~$\delta^n_m\notin C(\Z,\R)\otimes C(\Z,\R)$.
}} \QEDrem
\end{Example}

\subsection{A Poincar{\'e}-Type Lemma} \label{secpoincare}
We now analyze the question how and to what extent the Poincar\'e Lemma
applies in our setting.
Let~$U \subset \tilde{M}$ be an open subset. Assume that~$\omega \in 
\Omega^{s,m,\ell}(U \subset B_\L(U))$ with~$s,m > 0$ be a differential form which is
{\em{closed}} in the usual sense that~$d\omega=0$. We want to analyze
whether~$\omega$ is {\em{exact}}, meaning that there is a differential
form~$\nu \in \Omega^{s-1,m-1,\ell+1}(U \subset B_\L(U))$ with~$\omega = d\nu$.

Given~$p \in \tilde{M}$, we consider the $\L$-induced chart~$\phi_p$ defined 
in~\eqref{phipdef}. We assume that~$\phi_p|_{B_\L(U)}$ is injective. This makes it
possible to identify~$B_\L(U)$ via~$\phi_p$ with a subset of~$M_p$. For notational
clarity, we denote the resulting objects with an additional hat, i.e.\
\[ \hat{U} := \phi_p|_{B_\L(U)}^{-1}(U)\:,\qquad
B_\L(\hat{U}) := \phi_p|_{B_\L(U)}^{-1} \big( B_\L(U) \big) \:. \]
Likewise, we denote the push-forward of the measure~$\tilde{\rho}$ by
\[ \hat{\rho} := (\phi_p)^* \tilde{\rho} \]
and set~$\hat{M} := \supp \hat{\rho} \subset M_p$.

Next, we need to lift the Lagrangian as well as its derivatives to~$M_p$.
In preparation, we map the osculating vacuum at~$y \in U$ to~$M_p$.
To this end, we restrict the mapping~$K_p$ in~\eqref{Kpdef} to~$M_q$ and
take its linearization
\begin{align*}
K_p|_{M_q} &: M_q \rightarrow M_p \\
\Phi_{p,q} := D\big( K_p|_{M_q} \big) \big|_0 &: M_q \rightarrow M_p \quad \text{linear} \:.
\end{align*}
We then set
\[ \hat{\L} \::\: B_\L(\hat{U}) \times B_\L(\hat{U}) \rightarrow \R^+_0 \:,\qquad
\hat{\L} \big(\phi_p(x), \phi_p(y) \big) := \L(x,y) \]
and similarly (see~\eqref{D2def}),
\begin{align*}
&D_2^m \hat{\L} \big(\phi_p(x), \phi_p(y) \big) \::\:
\underbrace{M_p \times \cdots \times M_p}_{\text{$m$ factors}} \rightarrow \R \\
&\big( D_2^m \hat{\L} \big(\phi_p(x), \phi_p(y) \big) \big)
\big( \Phi_{p,y} u_1, \ldots \Phi_{p,y} u_m \big)
:= \big( D_2^m \L (x,y) \big) \big( u_1, \ldots, u_m \big) \:.
\end{align*}
In this way, we can lift the differential operators and differential forms to~$M_p$.
We denote the resulting identification operators by
\beq \label{phiident}
\begin{split}
\Phi &: \Gamma^\infty_\loc \big( B_\L(U), T^{s,r} \tilde{M} \big)
\rightarrow \Gamma^\infty_\loc \big( B_\L(\hat{U}), T^{s,r} \hat{M} \big) \\
\Phi &: \Omega^{s, m, \ell} \big( B_\L(U) \big) \rightarrow
\Omega^{s, m, \ell}\big( B_\L(\hat{U}) \big)\:.
\end{split}
\eeq
By construction, these mappings~$\Phi$ are invertible, making it possible
to identify differential forms on~$U \subset B_\L(U)$ with their lifts
on~$\hat{U} \subset B_\L(\hat{U})$.
The differential forms on~$\hat{U} \subset B_\L(\hat{U})$ can be written as usual
in components, i.e.\ a differential form~$\omega \in \Omega^{s,m,\ell} \big( \hat{U} \subset B_\L(\hat{U}) \big)$ can be written as
\beq \label{compU}
\omega = \sum_{|\kappa| \leq m}
D^\kappa \omega^\kappa_{i_1, \dots, i^s}(x)\: dx^{i_1} \wedge \cdots \wedge dx^{i_s} \:,
\eeq
defined for~$x \in \hat{U}$, where the component functions are equivalence classes with respect
to~$\simeq^\ell_{\hat{U}}$.
The corresponding lift of the exterior derivative is denoted by
\[ \hat{d} := \Phi \circ d \circ \Phi^{-1} \::\: \Omega^{s,m,\ell} \big( \hat{U} \subset B_\L(\hat{U}) \big)
\rightarrow \Omega^{s+1, m+1, \ell-1}\big( \hat{U} \subset B_\L(\hat{U}) \big) \:. \]
Clearly, it satisfies the relation~$\hat{d}^2=0$.

Following the standard notions, 
we say that a differential form~$\omega \in \Omega^{s,m,\ell} \big( \hat{U} \subset B_\L(\hat{U}) \big)$ is {\em{closed}} if~$\hat{d} \omega = 0$.
We want to analyze the question whether it is {\em{exact}}, meaning that
there is~$\nu \in \Omega^{s-1,m-1,\ell+1} \big( \hat{U} \subset B_\L(\hat{U}) \big)$ with~$\omega = d\nu$.
We want to adapt the proof of the Poincar{\'e} lemma to our setting.
Keeping in mind that the component functions in~\eqref{compU}
are defined only on~$\hat{U}$ (which might be a discrete subset of~$M_p$),
we first need to mollify~$\omega$. To this end, we again work with the Lagrangian,
using the notation~\eqref{Lnot},
\beq \label{mollify}
\Big(\L_{\tilde{\rho}} \big( D^\kappa \omega^\kappa_{i_1, \dots, i^s} \big)\Big)(x)
:= \int_{B_\L(\hat{U})} \big( D_x^\kappa \L(x,y) \big)\: \omega^\kappa_{i_1,
 \dots, i^s}(y)\: d\hat{\rho}(y) \:.
\eeq

Note that the mollification~\eqref{mollify} is defined for any~$x \in M_p$.
Therefore, following the standard proof of the Poincar{\'e} lemma, we can
integrate the mollified component functions along rays,
\begin{align}
&P(\omega)(x) := \sum_{|\kappa| \leq m}
\sum_{\alpha = 1}^s (-1)^{\alpha - 1} \notag \\
&\quad\; \times 
 \Big( \int_0^1 t^{s-1} 
\big( \L_{\tilde{\rho}} ( D^\kappa \omega^\kappa_{i_1 \cdots i_s}) \big)(tx) \:dt \Big) \:x^{i_\alpha}\;
dx^{i_1} \wedge \cdots \wedge \widehat{dx^{i_\alpha}} \wedge \cdots \wedge dx^{i_s} \:.
\label{Pomega}
\end{align}
In the case of the standard Poincaré lemma, the corresponding operator~$P$
has the property that~$dP(\omega) = \omega$.
In our setting, this identity will hold only approximately, i.e.\
\[ \hat{d} P(\omega) = \omega + \Delta \omega \]
with an error~$\Delta \omega \in \Omega^{s,m,\ell} \big( \hat{U} \subset B_\L(\hat{U}) \big)$. In order to quantify the error, we need to introduce a suitable norm.
In view of the mollification with~$\L$, all the 
derivatives of the component functions are bounded in terms of the
sup-norm of the functions. Therefore, all norms are equivalent to the $C^0$-norm
denoted by
\[ \|\check{\omega}\|_{C^0(\hat{U})} := \sup_{x \in B_\L(\hat{U})} \sup_{i_1, \dots, i^s}
\big| \big( \L_{\tilde{\rho}}( D^\kappa \omega^\kappa_{i_1, \dots, i^s}) \big)(x) \big| \:. \]
Nevertheless, in the applications it may be easier to work with other norms.
For this reason, it is preferable to work with an unspecified norm,
denoted by~$\|.\|_{\hat{U}}$.
Before we can formulate and prove our Poincaré-type lemma,
we need to treat the fact that differential forms are equivalence
classes of sections. So far, we worked with representatives.
This procedure is also most convenient for the analysis, because it
suffices to satisfy the necessary condition for a specific representative.
This leads to the following definition.
\begin{Def} \label{Lstarshaped}
The set~$\hat{U} \subset M_p$ is {\bf{$\L$-star-shaped}} if
for all closed differential forms~$\omega \in \Omega^{s,m,\ell} \big( \hat{U} \subset B_\L(\hat{U}) \big)$, the inequality
\beq \label{iterin}
\inf_{\check{\omega}}
\frac{\big\| \check{\omega} - \hat{d} P(\check{\omega}) \big\|_{\hat{U}}}{\|\check{\omega} \|_{\hat{U}}} < 1
\eeq
holds, where the infimum is taken over all representatives, i.e.\
\[ \check{\omega} \in \bigoplus_{r=0}^m
\Big( \Gamma^\infty_\loc \big( \tilde{M}, T^{s,r} \tilde{M} \big) \Big)
\qquad \text{and} \qquad \omega = \check{\omega} / \simeq^\ell_{\hat{U}}\:. \]
\end{Def} \noindent
We point out that this definition depends on the choice of
the norm~$\|.\|_{\hat{U}}$. This seems unavoidable because a Poincar{\'e}-type
lemma will hold only if this choice fits to the microscopic structure of~$\tilde{M}$.
This will be illustrated in Section~\ref{seclattice} in the example of a lattice system.
Here we proceed with a general abstract result.

\begin{Lemma} {\bf{(Poincaré-type lemma)}} \label{lemmapoincare}
in a $\L$-star-shaped region~$\hat{U} \subset M_p$, every closed
differential form~$\omega \in \Omega^{s,m,\ell} \big( \hat{U} \subset B_\L(\hat{U}) \big)$
is exact.
\end{Lemma}
\Proof In the proof, we shall always work with suitable representatives of
the differential forms; for ease in notation, we omit the superscripts~$\check{\:\,}$.
Starting from~$\omega^{(0)} := \omega$, we proceed inductively by setting
\[ \omega^{k+1} := \omega^{(k)} - \hat{d} P \omega^{(k)} \:. \]
Using that~$\hat{d}^2=0$, all these differential forms are closed.
Iterating the inequality~\eqref{iterin}, we find that
\[ \big\| \omega^{(k)} \big\|_{\hat{U}} \leq c^k\: \| \omega \|_{\hat{U}} \:. \]
Consequently, the series
\[ \nu := \sum_{k=0}^\infty P \omega^{(k)} \]
converges and, applying~$\hat{d}$ gives a telescopic sum
\[ \hat{d} \nu = \sum_{k=0}^\infty \big( \omega^{(k)} - \omega^{(k+1)} \big) 
= \omega^{(0)} = \omega \:. \]
Hence~$\omega$ is exact, concluding the proof.
\QED

\subsection{Computing the Cohomology, Independence of the Osculation} \label{secindep}
By combining the above constructions, one can follow the usual path for
computing the de Rham cohomology. But of course, the fact that we are in the
non-smooth setting leads to subtle modifications and requires additional assumptions,
as we now explain step by step.
We begin with a (finite or countable) covering of~$\tilde{M}$ by open sets~$(U_i)$.
For each set, we consider the open neighborhoods~$B_\L(U_i)$
(as shown on the top of Figure~\ref{figglue} in the example of two sets~$U_1$ and~$U_2$).
\begin{figure}[tb]
\psset{xunit=.3pt,yunit=.3pt,runit=.3pt}
\begin{pspicture}(717.57158895,600.34660279)
{
\newrgbcolor{curcolor}{0 0 0}
\pscustom[linewidth=2.64566925,linecolor=curcolor]
{
\newpath
\moveto(276.81236409,594.59680374)
\curveto(225.08869417,596.8791924)(139.47894047,618.31908405)(114.51779906,540.81143555)
\curveto(95.7916611,482.6642998)(116.90752252,411.15535492)(157.80874583,384.81979586)
\curveto(224.87739969,341.63550358)(338.32793953,333.36246925)(400.69464189,370.31596846)
\curveto(463.52887181,407.54646925)(471.87505134,473.95081445)(453.41999622,519.51926295)
\curveto(423.49418835,593.41072752)(328.57947591,592.31249886)(276.81236409,594.59680374)
\closepath
}
}
{
\newrgbcolor{curcolor}{0 0 0}
\pscustom[linewidth=2.64566925,linecolor=curcolor]
{
\newpath
\moveto(479.80087181,595.34501413)
\curveto(438.6273411,597.67318421)(342.0313852,600.14649571)(297.40361953,532.03734169)
\curveto(270.7422274,491.34776437)(307.63461165,410.57426767)(346.5137537,381.33586893)
\curveto(423.6337285,323.33910169)(504.99142299,330.74092846)(562.37916472,375.03334043)
\curveto(611.98795843,413.3218746)(608.57353323,481.38046232)(585.81767811,525.86124972)
\curveto(558.4497411,579.35717413)(523.56250961,592.87049697)(479.80087181,595.34501413)
\closepath
}
}
{
\newrgbcolor{curcolor}{0.80000001 0.80000001 0.80000001}
\pscustom[linestyle=none,fillstyle=solid,fillcolor=curcolor,opacity=0.5]
{
\newpath
\moveto(280.00079244,553.10757287)
\curveto(228.27712252,555.38995775)(180.07174299,566.69713634)(158.65224945,519.56166547)
\curveto(138.52668472,475.27359618)(169.38298205,423.07817319)(213.5752063,402.74301098)
\curveto(270.61316031,376.49684909)(329.99767181,374.51870484)(372.66925606,408.70099775)
\curveto(406.31192693,435.65059397)(429.01997102,470.12275775)(400.72158992,510.32565287)
\curveto(369.96269858,554.02408185)(331.76790425,550.82326799)(280.00079244,553.10757287)
\closepath
}
}
{
\newrgbcolor{curcolor}{0 0 0}
\pscustom[linewidth=1.51181105,linecolor=curcolor]
{
\newpath
\moveto(280.00079244,553.10757287)
\curveto(228.27712252,555.38995775)(180.07174299,566.69713634)(158.65224945,519.56166547)
\curveto(138.52668472,475.27359618)(169.38298205,423.07817319)(213.5752063,402.74301098)
\curveto(270.61316031,376.49684909)(329.99767181,374.51870484)(372.66925606,408.70099775)
\curveto(406.31192693,435.65059397)(429.01997102,470.12275775)(400.72158992,510.32565287)
\curveto(369.96269858,554.02408185)(331.76790425,550.82326799)(280.00079244,553.10757287)
\closepath
}
}
{
\newrgbcolor{curcolor}{0.80000001 0.80000001 0.80000001}
\pscustom[linestyle=none,fillstyle=solid,fillcolor=curcolor,opacity=0.5]
{
\newpath
\moveto(481.43185134,560.70004279)
\curveto(429.92429858,565.94576626)(353.21194583,566.20650453)(331.79246362,519.07103366)
\curveto(311.6668989,474.78296437)(351.39384189,409.39939397)(394.79850331,387.43335429)
\curveto(456.06861354,356.42604216)(517.4481411,378.29712122)(545.80948913,402.82162893)
\curveto(561.69658205,416.55947838)(582.31689071,477.62906201)(552.30683717,516.57094043)
\curveto(533.7153789,540.69574421)(519.44452535,556.82868783)(481.43185134,560.70004279)
\closepath
}
}
{
\newrgbcolor{curcolor}{0 0 0}
\pscustom[linewidth=1.51181105,linecolor=curcolor]
{
\newpath
\moveto(481.43185134,560.70004279)
\curveto(429.92429858,565.94576626)(353.21194583,566.20650453)(331.79246362,519.07103366)
\curveto(311.6668989,474.78296437)(351.39384189,409.39939397)(394.79850331,387.43335429)
\curveto(456.06861354,356.42604216)(517.4481411,378.29712122)(545.80948913,402.82162893)
\curveto(561.69658205,416.55947838)(582.31689071,477.62906201)(552.30683717,516.57094043)
\curveto(533.7153789,540.69574421)(519.44452535,556.82868783)(481.43185134,560.70004279)
\closepath
}
}
{
\newrgbcolor{curcolor}{0 0 0}
\pscustom[linewidth=2.64566925,linecolor=curcolor]
{
\newpath
\moveto(170.50825701,256.05908909)
\curveto(118.78458709,258.34150799)(33.17483339,279.78136941)(8.21369197,202.27372846)
\curveto(-10.51244712,144.12660405)(10.60341543,72.61764027)(51.50463874,46.28208122)
\curveto(118.5732926,3.09780405)(232.02383244,-5.17524162)(294.3905348,31.7782576)
\curveto(357.22478362,69.00875838)(365.57096315,135.41309224)(347.11590803,180.98157098)
\curveto(317.19008126,254.87303555)(222.27536882,253.77478043)(170.50825701,256.05908909)
\closepath
}
}
{
\newrgbcolor{curcolor}{0.80000001 0.80000001 0.80000001}
\pscustom[linestyle=none,fillstyle=solid,fillcolor=curcolor,opacity=0.5]
{
\newpath
\moveto(173.69668535,214.56985445)
\curveto(121.97301543,216.85223555)(73.76763591,228.15941035)(52.34814236,181.02393949)
\curveto(32.22257764,136.73588909)(63.07887496,84.54046232)(107.27109921,64.20531901)
\curveto(164.30905323,37.95914578)(223.69356472,35.98097886)(266.36514898,70.16329067)
\curveto(300.00781984,97.11287177)(322.71587528,131.58503555)(294.41748283,171.78794578)
\curveto(263.6585915,215.48638988)(225.46379717,212.28554578)(173.69668535,214.56985445)
\closepath
}
}
{
\newrgbcolor{curcolor}{0 0 0}
\pscustom[linewidth=1.51181105,linecolor=curcolor]
{
\newpath
\moveto(173.69668535,214.56985445)
\curveto(121.97301543,216.85223555)(73.76763591,228.15941035)(52.34814236,181.02393949)
\curveto(32.22257764,136.73588909)(63.07887496,84.54046232)(107.27109921,64.20531901)
\curveto(164.30905323,37.95914578)(223.69356472,35.98097886)(266.36514898,70.16329067)
\curveto(300.00781984,97.11287177)(322.71587528,131.58503555)(294.41748283,171.78794578)
\curveto(263.6585915,215.48638988)(225.46379717,212.28554578)(173.69668535,214.56985445)
\closepath
}
}
{
\newrgbcolor{curcolor}{0 0 0}
\pscustom[linewidth=2.64566925,linecolor=curcolor]
{
\newpath
\moveto(594.28223244,256.80732216)
\curveto(553.10870173,259.13551114)(456.51274583,261.60879618)(411.8849915,193.49966862)
\curveto(385.22359937,152.81006862)(422.11598362,72.03657571)(460.99511433,42.7981883)
\curveto(538.11508913,-15.1985865)(619.47278362,-7.79675973)(676.86052535,36.49567492)
\curveto(726.46931906,74.78417886)(723.05489386,142.84277413)(700.29903874,187.32358043)
\curveto(672.93110173,240.81950484)(638.04387024,254.33282767)(594.28223244,256.80732216)
\closepath
}
}
{
\newrgbcolor{curcolor}{0.80000001 0.80000001 0.80000001}
\pscustom[linestyle=none,fillstyle=solid,fillcolor=curcolor,opacity=0.5]
{
\newpath
\moveto(595.91321197,222.16235838)
\curveto(544.40565921,227.40807807)(467.69330646,227.66882767)(446.27381669,180.53335681)
\curveto(426.14824819,136.24526862)(465.87520252,70.86170956)(509.27986394,48.89566232)
\curveto(570.54997417,17.88834264)(631.92950173,39.75944815)(660.29084976,64.28393319)
\curveto(676.17794268,78.02179775)(696.79825134,139.09136626)(666.7881978,178.03327492)
\curveto(648.19673953,202.15807492)(633.92588598,218.2909883)(595.91321197,222.16235838)
\closepath
}
}
{
\newrgbcolor{curcolor}{0 0 0}
\pscustom[linewidth=1.51181105,linecolor=curcolor]
{
\newpath
\moveto(595.91321197,222.16235838)
\curveto(544.40565921,227.40807807)(467.69330646,227.66882767)(446.27381669,180.53335681)
\curveto(426.14824819,136.24526862)(465.87520252,70.86170956)(509.27986394,48.89566232)
\curveto(570.54997417,17.88834264)(631.92950173,39.75944815)(660.29084976,64.28393319)
\curveto(676.17794268,78.02179775)(696.79825134,139.09136626)(666.7881978,178.03327492)
\curveto(648.19673953,202.15807492)(633.92588598,218.2909883)(595.91321197,222.16235838)
\closepath
}
}
{
\newrgbcolor{curcolor}{0 0 0}
\pscustom[linewidth=3.77952756,linecolor=curcolor]
{
\newpath
\moveto(134.61470362,378.91345256)
\curveto(134.61470362,378.91345256)(110.09906268,359.48039933)(102.43516724,341.05516846)
\curveto(92.75620535,317.78538453)(97.70288504,292.78586201)(97.70288504,292.78586201)
}
}
{
\newrgbcolor{curcolor}{0 0 0}
\pscustom[linestyle=none,fillstyle=solid,fillcolor=curcolor]
{
\newpath
\moveto(97.70288504,292.78586201)
\lineto(86.59969347,300.22088011)
\lineto(102.28810187,269.61309992)
\lineto(105.13790314,303.88905357)
\closepath
}
}
{
\newrgbcolor{curcolor}{0 0 0}
\pscustom[linewidth=3.77952756,linecolor=curcolor]
{
\newpath
\moveto(596.4857348,383.64573854)
\curveto(596.4857348,383.64573854)(611.0313978,363.66486169)(610.68262299,334.42997035)
\curveto(610.47361512,316.90826201)(595.53930331,294.67877791)(595.53930331,294.67877791)
}
}
{
\newrgbcolor{curcolor}{0 0 0}
\pscustom[linestyle=none,fillstyle=solid,fillcolor=curcolor]
{
\newpath
\moveto(595.53930331,294.67877791)
\lineto(592.96536642,307.79118448)
\lineto(582.3662162,275.0708486)
\lineto(608.65170988,297.25271479)
\closepath
}
\rput[bl](620,550){$\F$}
\rput[bl](210,460){$U_1$}
\rput[bl](490,440){$U_2$}
\rput[bl](-30,460){$B_\L(U_1)$}
\rput[bl](610,460){$B_\L(U_2)$}
\rput[bl](160,120){$U_1$}
\rput[bl](550,110){$U_2$}
\rput[bl](-135,120){$B_\L(U_1)$}
\rput[bl](730,100){$B_\L(U_2)$}
}
\end{pspicture}
\caption{Glueing Construction.}
\label{figglue}
\end{figure}
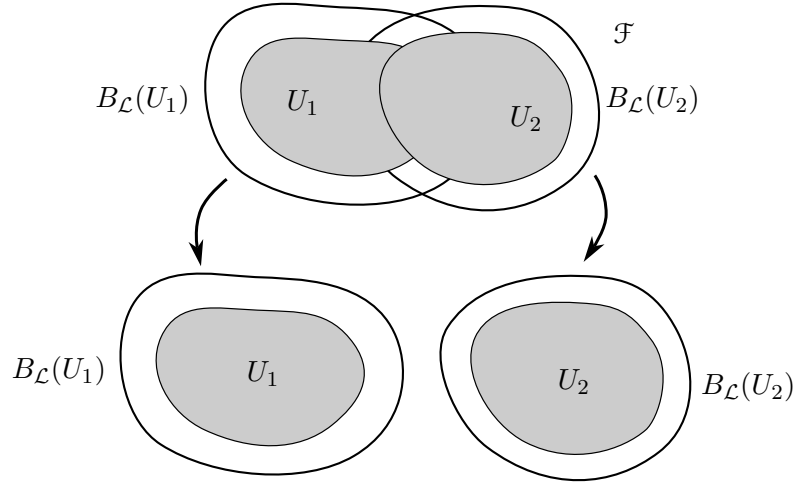%
Then Corollary~\ref{cormayer-viet} gives us the corresponding Mayer-Vietoris
sequence. Moreover, by the restriction lemma (Lemma~\ref{lemmarestrict}), 
the restriction of a differential form to~$U_i$ can be regarded separately
on~$B_\L(U_i)$, independent of the geometry of the other sets~$U_j$ with~$j \neq i$
and the topology of~$\tilde{M}$ (as indicated by the arrows in Figure~\ref{figglue}).
In order to apply the standard arguments (following for example in~\cite[Chapter~5]{bott-tu}),
we need to make sure that the cover is ``nice'' in the sense that the Poincar{\'e} lemma
applies. This will not be true in general because, as we saw in Section~\ref{secpoincare},
the Poincar{\'e} lemma holds only under additional assumptions
(intuitively speaking, if the region is star-shaped on a scale which is much larger
than the scale on which space or spacetime is non-regular).
Instead, we need to {\em{assume}} that there is a covering~$(U_i)$ such that all the sets
\beq \label{setintersect}
B_\L(U_i) \qquad \text{and} \qquad B_\L(U_i \cap U_j)
\eeq
satisfy the assumptions of Lemma~\ref{lemmapoincare}.
If this is the case, the usual Mayer-Vietoris argument goes through and allows us to
express the de Rham cohomology in terms of combinatorial properties of the covering.
In particular, we obtain the following result.
\begin{Thm} {\bf{(Independence of choice of osculations)}}
Suppose that two different choices of osculations can be joined by a
homotopy~$(M_p(\tau))_{p \in \tilde{M}}$ with~$\tau \in [0,1]$, i.e.\ a
continuous mapping
\[ \tau \in [0,1] \times \tilde{M} \rightarrow C^\infty(\F, \F)  \:,\qquad
(\tau, p) \mapsto M_p(\tau) \]
(where, as in~\eqref{loctop}, on~$C^\infty(\F, \F)$ we consider the
topology of locally uniform convergence of all derivatives).
Assume that for all~$\tau \in [0,1]$, all the sets~\eqref{setintersect}
satisfy the assumptions of Lemma~\ref{lemmapoincare}.
Then the cohomologies are independent of~$\tau$.
\end{Thm}

\section{Softened Surface Layer Integrals of Higher Co-Dimension} \label{secsoft}
\subsection{General Concepts and Basic Definitions}
The boundary integral in our Stokes' theorem (Theorem~\ref{thmstokes})
was formulated as a surface layer integral~\eqref{osik}.
Such surface layer integrals generalize surface integrals to
the setting of causal variational principles.
It is a natural question whether and how surface integrals of higher
co-dimension can be formulated as surface layer integrals.
In this section, we give a systematic procedure for doing so and
prove a corresponding Stokes' theorem (Theorem~\ref{thmstokesgen}).

Given a co-dimension~$q \in \{1,\ldots, k-1\}$, we want to formulate a surface layer integral
of dimension~$k-q$ (choosing~$q=1$ should give us back the surface layer
integral~\eqref{osik}, whereas~$q=k-1$ gives a one-dimensional surface layer
integral). In order to describe the region of integration in precise terms,
we follow the procedure explained in~\cite[Section~9.7]{intro} and choose~$q$
Borel sets
\[ \Omega_1, \ldots, \Omega_q \subset \tilde{M} \:. \]
In the case that~$\tilde{M}$ is a smooth manifold, we could assume that
these sets have smooth boundaries~$\partial \Omega_a$ which are all transversal.
In this case, we could integrate over the intersection of these boundaries,
\beq \label{smoothsurface}
\partial \Omega_1 \cap \cdots \cap \partial \Omega_q \:.
\eeq
In the general case, we cannot work with the boundaries of these sets
(for example, if~$\tilde{M}$ is a discrete space, $\partial \Omega_a=\Omega_a$
would be the whole set). Instead, we use the idea of surface layer integrals
by integrating over the Lagrangian~$\L(x,y)$ with one of the points~$x$ or~$y$
being inside the set~$\Omega_a$ and the other outside.
Moreover, in order to ensure that all weak derivatives will be well-defined,
we need to work with so-called {\em{softened surface layer integrals}}
as first introduced in~\cite[Section~3.2]{linhyp}. In our setting, this is implemented
most easily by replacing the characteristic functions of the sets~$\Omega_1, \ldots,
\Omega_q$ by mollification with the Lagrangian. Thus for a Borel subset~$\Omega
\subset \tilde{M}$ we define the function
\[ \chi_\Omega^\L \in C^\infty(\F, \R)\:,\qquad
\chi_\Omega^\L(x) := \fint_{\tilde{M}} \L(x,y)\: \chi_\Omega(y)\: d\rho(y) 
= \fint_{\Omega} \L(x,y)\: d\rho(y)\:. \]
Moreover, we introduce the abbreviation
\[ A_\Omega^\L(x,y) := \chi_\Omega^\L(x) - \chi_\Omega^\L(y)
\;=\; \chi_\Omega^\L(x)\:\big(1 - \chi_\Omega^\L(y) \big)
- \big( 1 - \chi_\Omega^\L(x) \big)\:\chi_\Omega^\L(y) \:. \]
Obviously, its~$y$-derivative is independent of~$x$,
\beq \label{indx}
D_{2,i} A_\Omega^\L(x,y) = - D_{i} \chi_\Omega^\L(y) \:;
\eeq
this simple property will be crucial for the subsequent construction.
In order to postpone the question of whether the following integrals are well-defined
on equivalence classes, we begin with a general section involving weak derivatives
up to order~$m$,
\beq \label{etasec}
\eta \in \bigoplus_{r=0}^m
\Gamma^\infty_\loc \big( \tilde{M}, T^{k-q,r} \tilde{M} \big) \:.
\eeq
In a parallel frame, this section has the form
\beq \label{etasecform}
\eta = \eta_{i_1 \cdots i_{k-q}}(y)\: e^{i_1}(y) \wedge \cdots \wedge  e^{i_{k-q}}(y)
\eeq
with coefficient functions which may contain weak derivatives up to the order~$r$.

\begin{Def} \label{defosi}
Given Borel subsets~$\Omega_1, \ldots, \Omega_q \subset \tilde{M}$ as well as a
section~$\eta$ as in~\eqref{etasec} and~\eqref{etasecform}, we
introduce the {\bf{surface layer integral of co-dimension~$q$}} by
\begin{align*}
&\int_{\partial \Omega_1 \cap (\partial \Omega_2 \wedge \cdots \wedge \partial
\Omega_q)} \eta 
:= (-1)^q \int_{\tilde{M}} d\tilde{\rho}(x) \int_{\tilde{M}} d\tilde{\rho}(y) \:
\big( \epsilon^\sharp \llcorner \eta)^{i_1 \cdots i_q}(y) \:
A^\L_{\Omega_1}(x,y) \notag \\
&\quad\: \times
\big( D_{2,i_q} A^\L_{\Omega_2}(x,y) \big) \big( D_{2,i_{q-1}} A^\L_{\Omega_3}(x,y) \big) \cdots
\big( D_{2,i_2} A^\L_{\Omega_q}(x,y) \big)\: \big( D_{2,i_1} \L(x,y) \big) \:.
\end{align*}
Moreover, given an additional Borel subset~$\tilde{U} \subset \tilde{M}$, we
introduce
\begin{align*}
&\int_{\tilde{U} \;\cap\; \partial \Omega_1 \cap (\partial \Omega_2 \wedge \cdots \wedge \partial \Omega_q)} \eta
:= (-1)^q \int_{\tilde{M}} \chi^\L_{\tilde{U}}(x)\:d\tilde{\rho}(x) \int_{\tilde{M}} d\tilde{\rho}(y) \:
\big( \epsilon^\sharp \llcorner \eta)^{i_1 \cdots i_q}(y) \: \\
&\quad\: \times\:A^\L_{\Omega_1}(x,y)\:
\big( D_{2,i_q} A^\L_{\Omega_2}(x,y) \big) \cdots
\big( D_{2,i_2} A^\L_{\Omega_q}(x,y) \big)\: \big( D_{2,i_1} \L(x,y) \big) \:.
\end{align*}
\end{Def} \noindent
These notions require a brief explanation. First of all, weak derivatives
acting on~$\eta$ are to be understood as before by integrating formally by parts
in~$M_y$. More precisely, using the product rules, we let these derivatives act on the arguments~$y$ of the factors~$A^\L_{\Omega_\ell}$ and~$\L(x,y)$.
This is unproblematic because all these factors are smooth.
Next, one should note that for each
set~$\Omega_a$, the integrand contains a function~$A^\L_{\Omega_a}(x,y)$
or a derivative thereof. This has the effect that the integrand vanishes unless
one of the points~$x$ or~$y$ lies inside the $\L$-neighborhood of~$\Omega_a$,
and the other point lies in the $\L$-neighborhood of~$\tilde{M} \setminus \Omega_a$.
In this sense, the above integrals are indeed surface layer integrals of co-dimension~$q$.
In view of the correspondence to the smooth case~\eqref{smoothsurface}, we
write the integration range on the left as the intersection of the boundaries of~$\Omega_1,
\ldots \Omega_q$. Next, we point out that the ordering of the
sets~$\Omega_1, \ldots, \Omega_q$ is
important. The set~$\Omega_1$ is distinguished in that the
function~$A^\L_{\Omega_1}(x,y)$ is not differentiated. All the other
factors~$A^\L_{\Omega_a}(x,y)$ are differentiated exactly once.
In view of the total anti-symmetry in the indices~$i_1, \ldots, i_1$, this means that
the surface layer integral is totally antisymmetric when interchanging~$\Omega_2,
\ldots \Omega_q$. This fact is emphasized by the wedges in the notation
for the integration range~$\partial \Omega_1 \cap (\partial \Omega_2 \wedge \cdots \wedge \partial \Omega_q)$.

\subsection{A Stokes Theorem in Higher Co-Dimension}
\begin{Thm} {\bf{(Stokes' theorem in higher co-dimension)}} \label{thmstokesgen}
For any relatively compact Borel set~$\tilde{U} \subset \tilde{M}$
and any section
\[ \eta \in \Gamma^\infty_\loc \big( \tilde{M}, T^{k-q-1,0} \tilde{M} \big) \:, \]
the following identity holds,
\beq \label{stokesgen}
\int_{\tilde{U} \;\cap\; \partial \Omega_1 \cap (\partial \Omega_2 \wedge \cdots \wedge \partial \Omega_q)} d \eta
= \int_{\partial {\tilde{U}} \cap (\partial \Omega_1 \wedge \partial \Omega_2 \wedge \cdots \wedge \partial \Omega_q) } \eta \:.
\eeq
\end{Thm}
\Proof First,
\[ \big( \epsilon^\sharp \llcorner (d \eta)\big)^{i_1 \cdots i_q}
= \big(\epsilon^\sharp\big)^{i_1 \cdots i_q\, i_{q+1}\, i_{q+2} \cdots i_k} \:
D_{i_{q+1}} \eta_{i_{q+2} \cdots i_k} \:.\]
Therefore, integrating by parts using the definition of the weak derivative,
\begin{align*}
&(-1)^k \int_{\tilde{U} \;\cap\; \partial \Omega_1 \cap (\partial \Omega_2 \wedge \cdots \wedge \partial \Omega_q)} d \eta \\
&= -\int_{\tilde{M}} \chi^\L_{\tilde{U}}(x)\:d\tilde{\rho}(x) \int_{\tilde{M}} d\tilde{\rho}(y) \:
\big(\epsilon^\sharp\big)^{i_1 \cdots i_q\, i_{q+1}\, i_{q+2} \cdots i_k}(y) \:
\eta_{i_{q+2} \cdots i_k}(y) \\
&\qquad\times\: D_{2,i_{q+1}} \Big( A^\L_{\Omega_1}(x,y)\:\big( D_{2,i_q} A^\L_{\Omega_2}(x,y) \big) \cdots
\big( D_{2,i_2} A^\L_{\Omega_q}(x,y) \big)\: \big( D_{2,i_1} \L(x,y) \big) \Big) \:.
\end{align*}
In view of the total anti-symmetry in the indices~$i_1, \ldots, i_{q+1}$, we only
get a contribution if the factor~$A^\L_{\Omega_1}(x,y)$ is differentiated, i.e.\
\begin{align*}
&(-1)^k \int_{\tilde{U} \;\cap\; \partial \Omega_1 \cap (\partial \Omega_2 \wedge \cdots \wedge \partial \Omega_q)} d \eta = -\int_{\tilde{M}} \chi^\L_{\tilde{U}}(x)\:d\tilde{\rho}(x) \int_{\tilde{M}} d\tilde{\rho}(y) \:
\big(\epsilon^\sharp \llcorner \eta \big)^{i_1 \cdots i_{q+1}}(y) \\
&\qquad \times\: \big( D_{2,i_{q+1}} A^\L_{\Omega_1}(x,y)\big) \:\big( D_{2,i_q} A^\L_{\Omega_2}(x,y) \big) \cdots
\big( D_{2,i_2} A^\L_{\Omega_q}(x,y) \big)\: \big( D_{2,i_1} \L(x,y) \big) \:.
\end{align*}
Next, we rewrite the factor~$\chi^\L_{\tilde{U}}(x)$ according to
\[ \chi^\L_{\tilde{U}}(x) = A^\L_{\tilde{U}}(x,y) + \chi^\L_{\tilde{U}}(y) \:. \]
Multiplying out and using Definition~\ref{defosi} we obtain
\begin{align*}
&\int_{\tilde{U} \;\cap\; \partial \Omega_1 \cap (\partial \Omega_2 \wedge \cdots \wedge \partial \Omega_q)} d \eta
- \int_{\partial {\tilde{U}} \cap (\partial \Omega_1 \wedge \partial \Omega_2 \wedge \cdots \wedge \partial \Omega_q) } \eta \\
&= - (-1)^k
\int_{\tilde{M}} d\tilde{\rho}(x) \int_{\tilde{M}} d\tilde{\rho}(y) \:\chi^\L_{\tilde{U}}(y)\:
\big(\epsilon^\sharp\big)^{i_1 \cdots i_q\, i_{q+1}\, i_{q+2} \cdots i_k}(y) \:
\eta_{i_{q+2} \cdots i_k}(y) \\
&\qquad \times\: \big( D_{2,i_{q+1}} A^\L_{\Omega_1}(x,y)\big) \:\big( D_{2,i_q} A^\L_{\Omega_2}(x,y) \big) \cdots
\big( D_{2,i_2} A^\L_{\Omega_q}(x,y) \big)\: \big( D_{2,i_1} \L(x,y) \big) \:.
\end{align*}
Using~\eqref{indx}, we can write the last expression in the short form
\begin{align}
&\int_{\tilde{U} \;\cap\; \partial \Omega_1 \cap (\partial \Omega_2 \wedge \cdots \wedge \partial \Omega_q)} d \eta
- \int_{\partial {\tilde{U}} \cap (\partial \Omega_1 \wedge \partial \Omega_2 \wedge \cdots \wedge \partial \Omega_q) } \eta \notag \\
&=\int_{\tilde{M}} d\tilde{\rho}(x) \int_{\tilde{M}} d\tilde{\rho}(y)\: f^k(y)\: \big( D_{2,k} \L(x,y) \big) \label{stokgen}
\end{align}
with the abbreviation
\begin{align*}
f^k(y) &= (-1)^{q+1} \chi^\L_{\tilde{U}}(y)\:
\big(\epsilon^\sharp\big)^{k\,i_2 \cdots i_{q+1}\, i_{q+2} \cdots i_k}(y) \:
\eta_{i_{q+2} \cdots i_k}(y) \\
&\qquad \times\: \big( D_{i_{q+1}} \chi^\L_{\Omega_1}(y)\big) \:\big( D_{i_q} \chi^\L_{\Omega_2}(y) \big) \cdots\big( D_{i_2} \chi^\L_{\Omega_q}(y) \big) \:.
\end{align*}
As in the proof of Theorem~\ref{thmstokes}, the right side of~\eqref{stokgen}
vanishes in view of the EL equations. This concludes the proof.
\QED

\subsection{Differential Forms in Softened Surface Layers} \label{secOosi}
The above analysis was carried out for a global section~$\eta$ of the form~\eqref{etasec}
and~\eqref{etasecform}. This is not quite satisfying, because it should suffice to
define the section in a neighborhood of the surface layer.
Moreover, we want to disregard properties of the sections which do not enter the
surface layer integrals. Following the strategy in Sections~\ref{secdiffex}
and~\ref{secrestrict}, this can be accomplished by forming suitable equivalence
classes of sections. Similar to~\eqref{etasec}, we now consider a more general section 
\[ 
\eta \in \bigoplus_{r=0}^m
\Gamma^\infty_\loc \big( \tilde{M}, T^{s,r} \tilde{M} \big) \:. \]
In a parallel frame, we consider the integral
\[ I(x) := \int_{\tilde{M}} \eta_{i_1 \cdots i_s}(y)\: A^\L_{\Omega_1}(x,y)\: D_{j_2} \chi_{\Omega_2}(y) \cdots  D_{j_q} \chi_{\Omega_q}(y)\:d\tilde{\rho}(y) \:. \]

\begin{Def} Given a Borel subset~$\tilde{U} \subset \tilde{M}$,
the section~$\eta$ vanishes in~$\tilde{U} \cap \partial \Omega_1 \cap
(\partial \Omega_2 \wedge \cdots \wedge \partial \Omega_q)$ to the order~$\ell$ if
the integral~$I(x)$ vanishes to the order~$\ell$ for any~$x \in \tilde{U}$ and for any choice of the indices~${i_1 \cdots i_s}$ and~$j_2, \ldots j_q$. The corresponding equivalence relation is
denoted by~$\simeq^\ell_{\tilde{U} \cap \partial \Omega_1 \cap
(\partial \Omega_2 \wedge \cdots \wedge \partial \Omega_q)}$.

The {\bf{differential forms}} in the surface layer~$\tilde{U} \cap \partial \Omega_1 \cap
(\partial \Omega_2 \wedge \cdots \wedge \partial \Omega_q)$
of class~$s,m,\ell$ are denoted and defined by
\[ \Omega^{s,m,\ell}\big( \tilde{U} \cap \partial \Omega_1 \cap
(\partial \Omega_2 \wedge \cdots \wedge \partial \Omega_q)\big)
:= 
\Big( \bigoplus_{r=0}^m
\Gamma^\infty_\loc \big( \tilde{M}, T^{s,r} \tilde{M} \big) \Big)
\big/ \simeq^\ell_{\tilde{U} \cap \partial \Omega_1 \cap
(\partial \Omega_2 \wedge \cdots \wedge \partial \Omega_q)} \:. \]
\end{Def} \noindent
As is immediately verified, these differential forms depend only on
their behavior in the surface layer, meaning that they may be modified
arbitrarily on a set~$\Omega$ which does not intersect the surface layer in the sense that
\[ \overline{\Omega} \;\cap\; \supp \big(
A^\L_{\Omega_1}(x,.)\: \chi_{\Omega_2}(.) \cdots \chi_{\Omega_q}(.) \big)
= \varnothing
\qquad \text{for all~$x \in \tilde{U}$}\:.
\]

In analogy to~\eqref{dmap}, the exterior derivative is a mapping
\[ d \::\: \Omega^{s,m,\ell}\big( \partial \Omega_1 \cap
(\partial \Omega_2 \wedge \cdots \wedge \partial \Omega_q)\big)
\rightarrow
\Omega^{s+1,m+1,\ell-1}\big( \partial \Omega_1 \cap
(\partial \Omega_2 \wedge \cdots \wedge \partial \Omega_q)\big) \]
for any~$\ell \geq 1$. Localizing to~$\tilde{U}$, we have accordingly
\[ d \::\: \Omega^{s,m,\ell}\big( \tilde{U} \cap \partial \Omega_1 \cap
(\partial \Omega_2 \wedge \cdots \wedge \partial \Omega_q)\big)
\rightarrow
\Omega^{s+1,m+1,\ell-1}\big( \tilde{U} \cap \partial \Omega_1 \cap
(\partial \Omega_2 \wedge \cdots \wedge \partial \Omega_q)\big) \:. \]
In particular, Theorem~\ref{thmstokesgen} applies to the 
exterior derivative acting on the differential forms
\[ d \::\: \Omega^{k-q-1,0,\ell}\big( \tilde{U} \cap \partial \Omega_1 \cap
(\partial \Omega_2 \wedge \cdots \wedge \partial \Omega_q)\big)
\rightarrow
\Omega^{k-q,1,\ell-1}\big( \tilde{U} \cap \partial \Omega_1 \cap
(\partial \Omega_2 \wedge \cdots \wedge \partial \Omega_q)\big) \:. \]
The statement of this theorem shows in particular that the
integral on the right of~\eqref{stokesgen} is well-defined on~$\Omega^{k-q-1, 0, \ell}(\tilde{U} \cap \partial \Omega_1 \cap
(\partial \Omega_2 \wedge \cdots \wedge \partial \Omega_q))$ with~$\ell \geq 1$
(in the sense that this integral does not depend on the choice of representatives).

\section{General Tensor Calculus} \label{sectensor}
In the previous sections we considered vector fields
(Section~\ref{secdirder}) and differential forms (Section~\ref{secLcalc}).
But clearly, one can also introduce a general tensor field~$t$
of degree~$(r,s)$ as a multilinear mapping
\[ t_p \::\: 
\underbrace{M^*_p \times \cdots \times M^*_p}_{\text{$r$ factors}} \times
\underbrace{M_p \times \cdots \times M_p}_{\text{$s$ factors}} 
\rightarrow \R \:. \]
In a local parallel frame, it can be represented as usual in components
as
\[  t_p(x) = t^{i_1 \cdots i_r}_{j_1 \cdots j_s}\: 
e_{i_1}(x) \otimes \cdots \otimes e_{i_r}(x) \otimes
e_{j_1}(x) \otimes \cdots \otimes e_{j_s}(x) \:. \]
Note that no symmetry is required in the indices.
With the tensor product, the tensors form an algebra.
Depending on the application in mind, one can quotient out
equivalence relations, as was done for differential forms
in Section~\ref{secdiffex}. Moreover, one can apply the
operations~$\epsilon_\sharp$ and~$\epsilon^\flat$
(introduced in Section~\ref{secintegrate}) to raise and lower indices.

We remark that, in order to get to the setting of differential
tensor calculus, one needs additional structures.
These structure were introduced in the smooth setting in~\cite{gauss}.
We now remark how they can be obtained in the non-smooth setting,
referring for details to the upcoming paper~\cite{nonsmooth}.
In order to get into the position to arbitrarily raise and lower indices,
one needs a {\em{metric}}. A canonical tensor of degree~$(0,2)$
(which generalizes the inverse of the Riemannian metric to our setting)
can be defined by (for details see~\cite[Section~4.3]{gauss})
\[ g^*_p := \int_{M_p} \L(p,y)\: y \otimes y\: d\rho_p(y) \::\:
M_p^* \times M_p^* \rightarrow \R \:. \]
Using this tensor, one can raise and lower indices.
Moreover, this metric makes it possible to
define the Hodge star as an operator which maps $s$-forms to $(k-s)$-forms.
Moreover, one can define the Hodge Laplacian.
We finally note that the $\L$-induced connection~$\nabla^\L$
introduced in Section~\ref{secconnection} also gives rise to
notions of covariant derivative and curvature. This is worked out
in the smooth setting in~\cite{gauss}; the non-smooth generalizations
will be explored in~\cite{nonsmooth}.

\section{Illustrating Examples} \label{secex}
In this section we study various examples. The goal is not to be exhaustive,
but to illustrate the previous constructions and results in a simple and
concrete way.

\subsection{One Point} \label{sec:onepoint}
Applying the constructions of Section~\ref{sec:Cohomology}, we are now able to compute the cohomologies of various spaces. We start with the example of a single point.
Choose~$\tilde{M}=\{*\}\subset \F = \R$ and a Lagrangian~$\L \in C^\infty(\R \times \R, \R^+_0)$ with the properties
\begin{equation*}
    \L(*,*)=1\:, \qquad D^p_y\L(*,*)=c_p \quad \text{with~$p \in \N$}\:,
\end{equation*}
where only finitely many~$c_p$ are non-zero. Furthermore, we let~$m \in \N$ be the maximal index for which~$c_m\neq0$.

Next, we need to specify the dimension of~$M$. The simplest and natural choice
is to choose~$\dim M=0$. With this choice, only zero forms exist.
Thus we only need to consider the spaces~$\Omega^{s,m,\ell}$ for~$s=0$.
Next, no derivatives can be taken, implying that also~$m=0$ and~$\ell=0$.
The space~$\Omega^{0,0,0}(\tilde{M})$ consists of all real-valued functions
on~$\tilde{M}$, so that
\[ \Omega^{0,0,0}(\tilde{M}) = \R \:. \]
Thus, just as in the algebraic topological situation, the
only non-trivial cohomology group is
\[ H^{0}(\tilde{M}) = \R \:. \]

In contrast to the situation in algebraic topology, now we can also choose
the osculating vacuum as~$M_\ast=M=\F$ and thus~$k= \dim M = 1$. In this case, the cohomology has no classical counterparts. It can be determined in detail as follows.
We denote a basis vector of~$M$ by~$e_1$. In a trivial way, this basis vector gives rise to a corresponding parallel frame, again denoted by~$e_1$.

We begin by calculating the objects in the differential complex
\[ 
    \begin{tikzcd}
        0\arrow[r] &
        \Omega^{0,0,\ell}(\tilde{M}) \arrow[r,"d"] &
        \Omega^{1,1,\ell-1}(\tilde{M}) \arrow[r] & 0\:.
    \end{tikzcd} \]
Combining the definitions~\eqref{Cequi} and~\eqref{omegaequi} for~$m=0$, we obtain that~$\Omega^{0,0,\ell}(\tilde{M})=\mathcal{D}^{0,\ell}(\tilde{M})$. Furthermore, using that~$T^{0,0}\tilde{M}\cong\R$, we conclude that
\begin{equation*}
    \Gamma^\infty_{\loc}(\tilde{M},T^{0,0}\tilde{M})\cong L^\infty(\{*\},\R)\cong\R.
\end{equation*}

The next step in computing the $(0,0,\ell)$-forms
is to study the equivalence relation~$\simeq^\ell_{\tilde{M}}$.
For simplicity we restrict attention to the case~$\ell>m>0$
(the case~$\ell \leq m$ and~$m=0$ are worked out in~\cite{vandertop}).
Definition~\ref{defvanish} states that~$f\simeq^\ell_*0$ for a section
$f\in\Omega^{0,0,\ell}(*)$ if and only if 
\begin{align*}
    \int_M\L(x,y)f(y)dy&=\L(x,*)f(*)=0,\\[-1em]
    &\vdots\\[-1em]
    \int_MD^\ell_y\L(x,y)f(y)dy&= D^\ell_y \L(x,*)f(*)=0,
\end{align*}
for all~$x\in\tilde{M}$. Evaluating at~$x=*$ gives the conditions
\begin{align}\label{eq:matrix-form1}
    c_p\cdot f(*)&=0 \qquad \text{for all~$p=0,\ldots, \ell$} \:.
\end{align}
As~$c_0\neq0$, we conclude that~$\R\ni f\simeq^\ell_{\tilde{M}}0$ if and only if~$f=0$. Thus~$\simeq^\ell_{\tilde{M}}$ is a trivial equivalence relation. Combining
the above results, we find
\begin{equation*}
    \Omega^{0,0,\ell}(\tilde{M})=\mathcal{D}^{0,\ell}(\tilde{M})=\frac{\Gamma^\infty_{\loc}(\tilde{M},T^{0,0}\tilde{M})}{\simeq^\ell_{\tilde{M}}}\cong\frac{\R}{\simeq^\ell_{\tilde{M}}}\cong\R.
\end{equation*}
Now we turn to~$\Omega^{1,1,\ell-1}$. Using again~\eqref{omegaequi}, we find
\begin{equation*}
    \Omega^{1,1,\ell-1}(\tilde{M})
    =
    \frac{\Gamma(\tilde{M},T^{1,0}\tilde{M})\;\oplus\;\Gamma(\tilde{M},T^{1,1}\tilde{M})}
         {\simeq^{\ell-1}_{\tilde{M}}}
    =
    \frac{e_1\mathbb{R}\oplus e_1 \otimes e^1\mathbb{R}}
         {\simeq^{\ell-1}_{\tilde{M}}}.
\end{equation*}
Again using Definition~\ref{defvanish}, we obtain that a section~$f=(f^0+f^1e_1)e^1$ vanishes everywhere to the order~$\ell-1$
if and only if for every~$p=0,\ldots, \ell-1$,
\begin{gather}\label{eq:point-vanish-1-form}
    \int_MD_y^p\L(x,y)(f^0+f^1e_1)(y)dy=D_y^p\L(x,*)f^0+D_y^{p+1}\L(x,*)f^1=0\\
    \Longleftrightarrow \qquad c_pf^0=-c_{p+1}f^1,\label{eq:point-vanish-1-form2}
\end{gather}
where in the last step we evaluated~\eqref{eq:point-vanish-1-form} at~$x=*$.
Collecting the equations~\eqref{eq:point-vanish-1-form2} for~$p=0, \ldots, \ell-1$
and using that~$c_{m+1} =\cdots =c_{\ell-1}=0$, the conditions can be
rewritten as the matrix equation
\begin{equation}\label{eq:point-vanish-1-form-matrix}
\begin{pmatrix}
    c_0 & c_1 \\
    c_1 & c_2 \\
    \vdots & \vdots \\
    c_{m-1} & c_m\\
    c_m & 0
\end{pmatrix}
\begin{pmatrix}
    f^0\\
    f^1
\end{pmatrix}
=0 \:.
\end{equation}
It follows that~$c_mf^0=0$,
which, as~$c_m\neq0$ per definition, implies~$f^0=0$. 
Using that~$m>0$, it follows that also~$f^1=0$.
We conclude that~$f\simeq^{\ell-1}_{\tilde{M}}0$ if and only if~$f=0$.
Hence
\[ \Omega^{1,1,\ell-1} \simeq \R^2 \:. \]

Since~$e^1\wedge e^1=0$, we obtain the following differential complex,
\begin{equation} \label{exact1}
    \begin{tikzcd}
        0\arrow[r] & 
        \Omega^{0,0,\ell}(\tilde{M}) \simeq \R \arrow[r,"d"] & 
        \Omega^{1,1,\ell-1}(\tilde{M}) \simeq \R^2 \arrow[r] & 
        0
    \end{tikzcd}
\end{equation}
It remains to determine the behavior of the operator~$d$. Following~\eqref{ddef}, we find for~$f\in\Omega^{0,0,\ell}(\tilde{M})$
\[ 
    df=f\;e_1\, e^1 \:. \]
Since the equivalence relation~$\simeq^{\ell-1}_{\tilde{M}}$ is trivial,
$df$ vanishes if and only if~$f=0$. Thus~$d$ is an injective operator.
Using this in~\eqref{exact1} we obtain the cohomologies
\begin{equation*}
    H^p(*)=
    \begin{cases}
        0& \text{if~$p\neq1$} \\
        \R & \text{if~$p=1$}\:.
    \end{cases}
\end{equation*}

\subsection{Two Points} \label{sec:Two-points}
We consider the space~$\tilde{M}=\{x_1,x_2\}\subset\F = \R$. 
Choosing~$k=0$ brings us back to the standard setting of algebraic topology.
It is more interesting to consider the case~$k=1$.
Thus we choose the osculations
\[ M_{x_1}=M_{x_2} = \F \:. \]
As in~\eqref{eq:matrix-form1} and~\eqref{eq:point-vanish-1-form-matrix} we will define the Lagrangian on~$\tilde{M}$ in matrix form as
\begin{equation}\label{eq:two-point-largangian}
    \L=
    \begin{pmatrix}
        1 & a_0 \\
        a_0 & 1
    \end{pmatrix},
    \;\;
    D^p_2\L=
    \begin{pmatrix}
        b_p & a_p \\
        a_[ & b_p
    \end{pmatrix} \quad \text{with~$p \in \N$}\:,
\end{equation}
where we assume that for only finitely many~$p$, the matrix~$D^p_2\L$ is non-zero. Furthermore, as in Section~\ref{sec:onepoint}, we will assume that~$m$ is the maximal value for which~$D^m_2\L\neq0$.

Furthermore, we let~$U=\{x_1\}$ and~$V=\{x_2\}$ so that~$U\cup V=\tilde{M}$. There are two methods of computation at play here. First, in case that we can find an $\L$-separating cover, we may apply Corollary~\ref{cormayer-viet} to compute the cohomology. Otherwise, a direct computation can be applied.  

We begin by applying the Mayer-Vietoris sequence in Corollary~\ref{cormayer-viet}.
In order to ensure that~$U$ and~$V$ form an $\L$-separating cover, we assume
for simplicity that
\[ a_1 = \cdots = a_m = 0 \:. \]
This implies 
\begin{equation*}
    B_\L(U)\cap B_\L(V)=U\cap V=\varnothing,
\end{equation*}
which by Lemma~\ref{lemmasufficient} provides that~$U$ and~$V$
are $\L$-separating the differential forms.

We now apply the Mayer-Vietoris long exact sequence in combination with the calculations in Section~\ref{sec:onepoint} to obtain
\begin{equation*}
    \begin{tikzcd}[column sep = small]
        0 \arrow[r]
        & H^0(\tilde{M}) \arrow[r]
        & H^0(U\subseteq\tilde{M})\oplus H^0(V\subseteq\tilde{M}) \arrow[r]
        \arrow[d, phantom, ""{coordinate, name=Z}]
        & H^0(\varnothing)=0 \arrow[dll,
            "\partial",
            rounded corners,
            to path={ -- ([xshift=2ex]\tikztostart.east)
            |- (Z) [near end]\tikztonodes
            -| ([xshift=-2ex]\tikztotarget.west)
            -- (\tikztotarget)}] \\
        &
        H^{1}(\tilde{M}) \arrow[r]
        & H^{1}(U\subseteq\tilde{M})\oplus H^{1}(V\subseteq \tilde{M}) \arrow[r]
        & H^1(\varnothing)=0
    \end{tikzcd}
\end{equation*}
We conclude that
\begin{equation*}
    H^p(\tilde{M})=H^p(U\subseteq\tilde{M})\oplus H^p(V\subseteq\tilde{M}) \:.
\end{equation*}

In the case that some of the off-diagonal matrix elements in~\eqref{eq:two-point-largangian},
we cannot use the Mayer-Vietoris sequence.
But the cohomology groups can still be computed directly as follows. Using~\eqref{omegaequi}, we obtain the following differential complex
\[ 
    \begin{tikzcd}
        0 \arrow[r] &
        \Omega^{0,0,\ell}(\tilde{M}) \simeq \displaystyle \frac{\R^2}{\simeq^\ell_{\tilde{M}}} \arrow[r, "d"] &
        \Omega^{1,1,\ell-1}(\tilde{M}) \simeq \displaystyle \frac{\R^4}{\simeq^{\ell-1}_{\tilde{M}}} \arrow[r] &
        0 \:.
    \end{tikzcd} \]

Having written the Lagrangian and its derivatives in matrix form, we can give a linear algebraic description of~$\simeq^a_{\tilde{M}}$ for all~$a\in\N$. In particular we may write
\[ 
    \text{Ker} \big( \simeq^{\ell}_{\tilde{M}} \big) =\bigcap_{p=0}^\ell\text{Ker}(D^p\L)\:,\qquad
    \text{Ker} \big( \simeq^{\ell-1}_{\tilde{M}} \big) =\bigcap_{p=0}^{\ell-1} \text{Ker}(D^p\L\boxplus D^{p+1}\L) \:, \]
where~$D^j\L\boxplus D^{j+1}\L$ is the $(2\times4)$-matrix with~$D^j\L$ and~$D^{j+1}\L$ as block matrices. This linear algebraic formulation is analogous to~\eqref{eq:point-vanish-1-form-matrix}. Now these kernels can be worked out in a case-by-case analysis
depending on the values of the matrix entries in~\eqref{eq:two-point-largangian}.
We here omit the details, which can be found in~\cite{vandertop}.

\subsection{The Discrete Line} \label{sec:discline}
In this section we calculate the cohomology groups of the discrete line embedded in~$\R$. Unlike in Sections~\ref{sec:onepoint} and~\ref{sec:Two-points}, the differential complex will not be computed. Instead, we focus directly on calculating the cohomology groups associated to this complex.

We consider the space~$\tilde{M} = \Z \subset \F = \R$ and again
choose~$k=1$ as well as the osculations
\[ M_j = \F \qquad \text{for all~$j \in \Z$}\:. \]
Furthermore, we set~$\ell=1$ and choose a Lagrangian~$\L \in C^\infty(\F \times \F, \R^+_0)$ with the properties
\begin{equation*}
    \L(i,j)=
    \begin{cases}
        1 & \text{if~$j=i$} \\
        \alpha & \text{if~$j=i\pm1$} \\
        0 & \text{otherwise}
    \end{cases}
\qquad \text{and} \qquad
    D_2\L(i,j)=
    \begin{cases}
        \beta & \text{if~$j=i\pm1$} \\
        0 & \text{otherwise} \:,
    \end{cases}
\end{equation*}
where~$\alpha, \beta$ are real parameter. Moreover, we assume that
all higher order derivatives of~$\L$ vanish on~$\Z \times \Z$.
For simplicity, we restrict attention to the differential complex
\[ 
    \begin{tikzcd}
        0\arrow[r] &
        \Omega^{0,0,1}(\tilde{M}) \arrow[r,"d"] &
        \Omega^{1,1,0}(\tilde{M}) \arrow[r] & 0 \:.
    \end{tikzcd} \]

Using~\eqref{H0}, we find that
\begin{equation*}
    H^0(\Z)=\frac{\text{Ker}(D\L)}{\text{Ker}(\L)\cap \text{Ker}(D\L)},
\end{equation*}
where~$\phi\in \text{Ker}(\L)\cap \text{Ker}(D\L)$ if and only if for all~$i\in\Z$
the following conditions hold,
\begin{align}
\phi(i)+\alpha\: \big(\phi(i-1)+\phi(i+1) \big) &=0 \label{eq:0-vanish1}\\
\beta\: \big( \phi(i-1)+\phi(i+1) \big) &=0 \:. \label{eq:0-vanish2}
\end{align}
Additionally, any~$(1,1,0)$-form, $\phi^0e^1+\phi^1e_1e^1$, is trivial if and only if~$\phi^0e^1+\phi^1e_1e^1\in \text{Ker}(\L\boxplus D\L)$. In algebraic terms this coincides can be written as
\begin{equation}
\label{eq:1-vanish}
    \phi^0(i)+\alpha \big(\phi^0(i-1)+\phi^0(i+1)\big)+\beta \big( \phi^1(i-1)+\phi^1(i+1) \big)=0 \:,
\end{equation}
to be satisfied for all~$i\in\Z$.

Analyzing the conditions~\eqref{eq:0-vanish1}--{eq:1-vanish} systematically
gives rise to the following cohomology.
\begin{Prp} Depending on the values of the parameters~$\alpha$ and~$\beta$,
the cohomology groups are given by
\begin{align}
H^p(\Z) &\cong \begin{cases}
\R^2 & \text{if~$p=0$} \\ 0& \text{otherwise}.\end{cases} && \text{if~$\beta \neq 0$}
\label{Hp1} \\
H^0(\Z) &=\Omega^{0,0,1}(\Z) \text{ and }
H^1(\Z) =\Omega^{1,1,0}(\Z)
 && \text{if~$\beta=0$\:.} \label{eq:beta=0-cohom}
\end{align}
\end{Prp}
\Proof
We begin with the case~$\beta\neq0$. Multiplying~\eqref{eq:0-vanish2}
by~$\alpha/\beta$ and subtracting~\eqref{eq:0-vanish1}, one finds
that the equations~\eqref{eq:0-vanish1} and~~\eqref{eq:0-vanish2}
only admit trivial solutions, and thus~$\text{Ker}(\L)\cap \text{Ker}(D\L)=0$. Furthermore, as~$\beta\neq0$, we conclude that~$\phi\in \text{Ker}(D\L)$ if and only if
\[ 
    \phi(i-1)=-\phi(i+1) \qquad \text{for~$i \in Z$}\:. \]
These equations have a two-dimensional solution space, which can be
parametrized by~$\phi(0)$ and~$\phi(1)$. We thus obtain
an isomorphism~$H^0(\Z)\cong\R^2$ given by
\[ 
    \phi\mapsto
    \begin{pmatrix}
        \phi(0)\\
        \phi(1)
    \end{pmatrix} . \]

Next we show that~$d$ is a surjective operator. By definition of~$d$, its image coincides with the $(1,0,0)$-forms. As any $(1,1,0)$-form is a sum of $(1,0,0)$ and $(0,1,0)$-forms, it suffices to show that all $(1,0,0)$-forms are in the image of~$d$. 
To this end, let~$\phi^0e^1$ be a $(1,0,0)$-form. We will construct a form~$\psi$ such that~$d(\psi)=\phi^0e_1$. Define~$\psi(-1),\psi(0),\psi(1)$ and~$\psi(2)$ by any values making the following relations hold,
\begin{align*}
    \phi^0(0)+\alpha\:\big( \phi^0(-1)+\phi^0(1) \big)
    +\beta\: \big(\psi(1)+\psi(-1)\big ) &=0 \\
    \phi^0(1)+\alpha\: \big(\phi^0(0)+\phi^0(2)\big)+\beta\:\big(\psi(2)+\psi(0) \big) &=0.
\end{align*}
Comparing this with~\eqref{eq:1-vanish}, we see that~$d(\psi)(i)=\phi^0e_1(i)$ for~$i=0,1$.

The remaining task is to extend~$\psi$ to be a global $(0,0,1)$-form in such
a way that the relation~$d(\psi)(i)=\phi^0e_1(i)$ holds for all~$i \in \Z$.
To this end, we construct~$\psi(l)$ for~$l>1$ inductively by the relation
\[ 
    \psi(l)=-\psi(l-2)-\frac{1}{\beta}\:\phi^0(l-1)-\frac{\alpha}{\beta}\: \big(\phi^0(l)+\phi^0(l-2)\big). \]
Similarly, we construct~$\phi(l)$ for~$l<-1$ inductively by
\[ 
    \psi(l)=-\psi(l+2)-\frac{1}{\beta}\:\phi^0(l+1)-\frac{\alpha}{\beta}\:\big(\phi^0(l+2)+\phi^0(l)\big).
\]
In this way, using the equivalence relation~\eqref{eq:1-vanish} and by construction of~$\psi$, we obtain
\begin{equation*}
    d(-\psi)=-\psi \:e_1e^1\simeq^0_\Z\phi^0e^1 \:.
\end{equation*}
Hence~$d$ is indeed surjective.

The surjectivity of~$d$ implies that~$H^1(\Z)=0$, proving~\eqref{Hp1}.

Next we assume~$\beta=0$. In this case, the operator~$D\L$ vanishes. However, from~\eqref{ddef} that~$d$ itself is the zero operator. Using the~\eqref{H0} and~\eqref{Hk}, we obtain~\eqref{eq:beta=0-cohom}.
\QED

So far, we considered global sections without any decay assumptions
at infinity. Imposing such conditions may have an influence on the
cohomology, as we now illustrate by considering compactly supported
sections.
\begin{Prp} \label{prpdisccmp}
Restricting attention to compactly supported differential forms,
the cohomology groups are given by
\begin{align}
H^0_c(\Z) &=
    \begin{cases}
        0 & \text{if~$\beta\neq0$} \\
       C_c(\Z,\R) & \text{if~$\beta=0$}
    \end{cases} \label{H0c} \\
    H^1_c(\Z) &\cong \begin{cases}
        \Omega^{1,1,0}(\Z) & \text{if~$\beta\neq0$} \\
        C_c(\Z,\R) & \text{if~$\beta=0$\:.}
    \end{cases} \label{eq:1st-Cohom-compact}
\end{align}
\end{Prp}
\Proof We begin by calculating the $(0,0,1)$-forms. To this end, we note that
a compactly supported form~$\omega\neq0$ can be represented as a finite sum
\begin{equation*}
    \omega=\sum_{i=1}^N\alpha^i \:\mathds{1}_{x_i} \:,
\end{equation*}
for~$x_i\in\Z$ and~$\alpha^i\in\R$ (and~$\mathds{1}_{x_i}$ is the
function which is one at~$x_i$ and zero otherwise).
Assume that~$\omega \in \Omega^{0,0,1}(\tilde{M})$ is non-zero
but~$\omega \simeq^1_\Z$. Then we may choose~$j \in \Z$
such that~$\omega(j)\neq0$ but~$\omega(j+1)=0=\omega(j+2)$.
Next, according to the relations~\eqref{eq:0-vanish1} and~\eqref{eq:0-vanish2},
\begin{align}
0 &= \omega(j+1)+\alpha\: \big(\omega(j)+\omega(j+2) \big)=\alpha\cdot\omega(j)\label{eq:alpha-compact-vanish}\\
0 &= \beta\: \big( \omega(j)+\omega(j+2) \big)=\beta\cdot\omega(j) \:.\label{eq:beta-compact-vanish} 
\end{align}
In the case~$\alpha\neq0$ or~$\beta\neq0$, it follows that~$\omega(j)=0$, in contradiction to our choice~$j$. As such, if~$\alpha\neq0$ or~$\beta\neq0$, then~$\omega\simeq^1_\Z0$ if and only if~$\omega=0$.
On the other hand, in the case~$\alpha=0$ and~$\beta=0$, \eqref{eq:0-vanish1} again 
shows that~$\omega(j)$ vanishes, which is again a contradiction.
We conclude that~$\omega\simeq^1_\Z0$ if and only if~$\omega=0$.

The zeroth cohomology group is again calculated by regarding~$\text{Ker}(d)$. If~$\beta=0$, again we find~$d=D\L=0$. In the remaining case~$\beta \neq 0$.
we let~$\omega \in \Omega^{0,0,1}(\tilde{M})$ be non-zero.
We again choose~$j$ such that~$\omega(j) \neq 0$. Assume that~$d \omega =0$.
Then relation~\eqref{eq:1-vanish} yields
\begin{equation*}
    \beta(\omega(j)+\omega(j+2))=\beta\cdot\omega(j)=0.
\end{equation*}
Note that this coincides with~\eqref{eq:beta-compact-vanish}. However, 
by assumption~$\beta\neq0$ and~$\omega(j)\neq0$, a contradiction.
Therefore,  $d\omega$ must be non-zero, proving that~$d$ is an injective operator.
This concludes the proof of~\eqref{H0c}.

We next compute the first cohomology groups. Using~\eqref{Hk}, we may write
\begin{equation}\label{eq:1st-cohom-compact}
    H^1(\Z)
    =
    \frac{\text{Ker}(d)}{\text{Im}(d)}
    =
    \frac{\Omega^{1,1,0}(\Z)}{\text{Im}(d)}.
\end{equation}
In the case~$\beta=0$, we already showed that~$d$ vanishes,
showing that~$H^1(\Z)=\Omega^{1,1,0}(\Z)$.
In the remaining case~$\beta \neq 0$, we know that~$d$ is an injective operator.
Hence~\eqref{eq:1st-cohom-compact} reduces to
\begin{equation*}
    \frac{\Omega^{1,1,0}(\Z)}{\text{Im}(d)}
    =
    \frac{\displaystyle \frac{C_c(\Z,\R)\;e^1\;\oplus\;C_c(\Z,\R)\;e_1e^1}{\simeq^0_\Z}}{\displaystyle \frac{C_c(\Z,\R)\;e_1e^1}{\simeq^0_\Z}}\cong\frac{C_c(\Z,\R)}{\simeq^0_\Z}.
\end{equation*}
It remains to determine the equivalence relation~$\simeq^0_\Z$.
The condition~$\omega \simeq^0_\Z 0$ implies that~\eqref{eq:alpha-compact-vanish}
holds for all~$j$. Arguing similar as for the equivalence relation~$\simeq^1_\Z$ above,
it follows that~$\omega=0$. Hence the equivalence relation~$\simeq^1_\Z$ is trivial. This concludes the proof.
\QED

\subsection{The Discrete Circle}\label{sec:disc-circ}
In this section, the cohomology groups of a discrete circle are computed. Here our space consists of four distinct points. As the discrete circle is to model the geometry of a circle, we again choose~$k=1$. Thus we choose the osculating vacuum one dimensional, reminiscent of the standard tangent bundle of a circle. We will denote the parallel frame of the osculating vacuum by~$e_1$. 

In the same vein as~\eqref{eq:point-vanish-1-form} we will define a Lagrangian on~$\tilde{M}$ in a matrix form. For~$a,c\in\R$ and~$b\in\R \setminus \{0\}$, we define the Lagrangian and its derivatives by
\begin{equation}\label{eq:lagrangian-circle}
    \L=
    \begin{pmatrix}
        1 & c & 0 & c \\
        c & 1 & c & 0 \\
        0 & c & 1 & c \\
        c & 0 & c & 1
    \end{pmatrix},\;\;
    D_2\L=
    \begin{pmatrix}
        a & b & 0 & b \\
        b & a & b & 0 \\
        0 & b & a & b \\
        b & 0 & b & a
    \end{pmatrix}.
\end{equation}
Once again we assume all higher derivatives of~$\L$ vanish. 
We note that~\eqref{eq:lagrangian-circle} takes the form of a Toeplitz matrix encoding the geometry of the discrete circle.

Our first object of study is~$\Omega^{0,0,1}(\tilde{M})$. Once again we turn to~\eqref{omegaequi} to obtain
\[ \Omega^{0,0,1}(\tilde{M})=\frac{\R^4}{\simeq^1_{\tilde{M}}}=\frac{\R^4}{\text{Ker}(\L)\cap\text{Ker}(D\L)} \:. \]
As in Section~\ref{sec:discline} we turn to the matrix form of~$\L$ to determine the equivalence relation~$\simeq^1,_{\tilde{M}}$. In particular, we analyze the characteristic polynomials
\begin{align}
    \det(xI-\L)&=(x-1)^2(x-(1-2c))(x-(1+2c)),\label{eq:char-pol-L-S}\\
    \det(xI-D_2\L)&=(x-a)^2(x-(a-2b))(x-(a+2b)) \:.\label{eq:char-pol-DL-s}
\end{align}

This provides several cases to be considered. 
Firstly, if~$c\neq\pm\frac{1}{2}$ or~$a\neq0,\pm2b$, all eigenvalues are non-zero. As such, $\simeq^1_{\tilde{M}}$ is trivial.

This leaves us two different cases to be considered:
In the case~$c=\pm\frac{1}{2}$ and~$a=\pm2b$, we see from~\eqref{eq:lagrangian-circle} that~$D\L=b\cdot\L$. Therefore, $\text{Ker}(\L)=\text{Ker}(D\L)$, implying that
\begin{equation}\label{eq:c-whole}
    \text{Ker}(\L)\cap\text{Ker}(D\L)=\text{Ker}(\L)=(\pm1,1,\pm1,1)^T\cdot\R\cong\R.
\end{equation}
Finally, in the case~$c=\pm\frac{1}{2}$ and~$a=\mp2b$, a direct computation provides
\begin{equation*}
    \text{Ker}(\L)\cap\text{Ker}(D\L)=(\pm1,1,\pm1,1)^T\cdot\R\cap(\mp1,1,\mp1,1)^T\cdot\R=0.
\end{equation*}
Having computed~$\simeq^1_{\tilde{M}}$, the definition~\eqref{omegaequi} provides the~$(0,0,1)-$forms. The results of these computations are summarized in the first column of Table~\ref{tab:1-forms-square}.
\begin{table}[tb]
    \begin{tabular}{c|c||c|c|c|c}
    $c$ & $a$ & $\Omega^{0,0,1}(\tilde{M})$ & $\Omega^{1,1,0}(\tilde{M})$ & $H^0(\tilde{M})$ & $H^1(\tilde{M})$ \\
    \hline
    $\neq\pm 1/2$ & $\neq0,\pm2b$ & $\R^4$ & $\R^4$ & $0$ & $0$\\
    $\pm 1/2$ & $0$ & $\R^4$ & $\R^4$ & $\R^2$ & $\R^2$\\
    $\pm 1/2$ & $\pm2b$ & $\R^3$ & $\R^3$ & $0$ & $0$ \\
    $\pm 1/2$ & $\mp2b$ & $\R^4$ & $\R^4$ & $\R$ & $\R$
    \end{tabular}
\vspace*{0.7em}
    \caption{The differential forms and cohomology of the discrete circle.}
    \label{tab:1-forms-square}
\end{table}

We continue by calculating the~$(1,1,0)$-forms. Just as in Section~\ref{sec:discline}, we write~$\phi\in\Omega^{1,1,0}(\tilde{M})$ in a parallel frame as
\begin{equation*}
    \phi=\phi^0e^1+\phi^1e_1e^1,
\end{equation*}
where~$\phi^i\in C(\tilde{M},\R)$. 
As in Section~\ref{sec:discline}, we obtain using the matrix representations in~\eqref{eq:lagrangian-circle} that~$\phi\simeq^0_{\tilde{M}}0$ if and only if
\begin{equation*}
    0=\int_M\L(x,y)\: \big[\phi^0+\phi^1\:e_1 \big](y)\:dy=(\L\boxplus D\L) \phi.
\end{equation*}
Thus, we turn to to calculating~$\text{Ker}(\L\boxplus D\L)$. Again this devolves into several cases. 

Firstly, we assume that~$c\neq\pm\frac{1}{2}$ or~$a\neq0,\pm2b$. Note that for these values~$\L$ and~$D\L$ is a surjective operator as no non-zero eigenvalues occur in~\eqref{eq:char-pol-L-S} and~\eqref{eq:char-pol-DL-s}. As such, we conclude that the matrix~$\L\boxplus D\L$ has rank 4. Per the rank nullity theorem, this provides that 
\begin{equation*}
    \text{dim} \big( \text{Ker}(\L\boxplus D\L) \big)=\text{dim}(\R^8)-\text{dim} \big(\text{Im}(\L\boxplus D\L) \big)=8-4=4 \:.
\end{equation*}

The remaining cases are treated using a direct row reduction to determine the rank, and thus the nullity, of~$\L\boxplus D\L$. This leads to the following
\begin{equation*}
    \text{dim} \big( \text{Ker}(\L\boxplus D\L) \big)=
    \begin{cases}
        5 & \displaystyle \text{if } c=\pm\frac{1}{2}\;\wedge\;a=\pm2b \\[0.7em]
        4 & \displaystyle \text{if } c=\pm\frac{1}{2}\;\wedge\;a=\mp2b \\[0.7em]
        4 & \displaystyle \text{if } c=\pm\frac{1}{2}\;\wedge\;a=0 \:.
    \end{cases}
\end{equation*}
Definition~\ref{omegaequi} now provides~$\Omega^{0,0,1}(\tilde{M})$ up to isomorphism, as summarized in the second column of Table~\ref{tab:1-forms-square}. 

Finally, we compute the cohomology groups. Given~\eqref{eq:c-whole} and~\eqref{H0}, we conclude that
\begin{equation*}
    H^0(\tilde{M})=\frac{\text{Ker}(D\L)}{\text{Ker}(\L)\cap \text{Ker}(D\L)}\cong
    \begin{cases}
        \R^2 & \displaystyle \text{if } c=\frac{1}{2}\wedge a=0 \\[0.7em]
        \R & \displaystyle \text{if } c=\pm\frac{1}{2}\wedge a=\mp 2b \\[0.7em]
        0 & \text{otherwise} \:.
    \end{cases}
\end{equation*}
Given these cohomology groups, we can use the following exact sequence of vector spaces to compute~$H^1(\tilde{M})$,
\begin{equation}\label{eq:exact-seq-H1}
    \begin{tikzcd}
        0 \arrow[r] &
        H^0(\tilde{M}) \arrow[r, "\iota"] &
        \Omega^{0,0,1}(\tilde{M}) \arrow[r,"d"] &
        \Omega^{1,1,0}(\tilde{M}) \arrow[r,"\pi"] & 
        H^1(\tilde{M}) \arrow[r] &
        0
    \end{tikzcd}
\end{equation}
As an example, for~$c=\frac{1}{2}$ and~$a=0$ we find the sequence~\eqref{eq:exact-seq-H1} becomes
\[ 
    \begin{tikzcd}
        0 \arrow[r] &
        \R^2 \arrow[r, "\iota"] &
        \R^4 \arrow[r,"d"] &
        \R^4 \arrow[r,"\pi"] & 
        H^1(\tilde{M}) \arrow[r] & 0 \:,
    \end{tikzcd} \]
showing that~$H^1(\tilde{M})\cong\R^2$. The exact same computation for all different values of~$c$ and~$a$ is shown in the fourth column of Table~\ref{tab:1-forms-square}.

\subsection{The $k$-Dimensional Lattice} \label{seclattice}
In this section the cohomology of a $k$-dimensional lattice
will be computed using several methods. To this end, we choose~$\tilde{M} = (\varepsilon \Z)^k \subset \F = \R^k$ as a square lattice of
spacing~$\varepsilon>0$. Moreover, we choose~$\tilde{\rho}$ as the counting
measure on~$\tilde{M}$, i.e.\
\[ \int_{\tilde{M}} f(x)\: d\tilde{\rho}(x) = \sum_{x \in \tilde{M}} f(x) \:. \]
We assume that the Lagrangian~$\L : \F \times \F \rightarrow \R^+_0$
is translation invariant, i.e.\
\beq \label{LL}
\L(x,y) = L(x-y) \:.
\eeq
We again choose the osculations as
\[ M_x = \F \qquad \text{for all~$x \in \tilde{M}$}\:. \]
Two main methods of computation will be applied, those being the Künneth formula for compactly supported differential forms, and a direct calculation
using the discrete Fourier transform.

\subsubsection{Products of Discrete Lines}\label{subsec:lines-Kunneth}
Clearly, the $k$-dimensional lattice can be regarded as
the $k$-fold product of the discrete line as considered in
Section~\ref{sec:discline} (the lattice spacing~$\varepsilon$ is irrelevant
and can be set to one). The cohomology of the resulting lattice system
can be computed using the Künneth formula of Section~\ref{sec:Kunneth}.

More precisely, we want to apply Theorem~\ref{lemma:kunneth-injection}.
Therefore, we must restrict attention to compactly supported differential forms
as considered in Proposition~\ref{prpdisccmp}. We denote the product time space and Lagrangian by 
\begin{equation*}
    \tilde{M}=\prod_{i=1}^k \:,\qquad \L=\prod_{i=1}^k\L_i \:,
\end{equation*}
Then the Künneth formula in Theorem~\ref{lemma:kunneth-injection} provides the following isomorphism of graded groups,
\[ 
    \Phi:\bigotimes_{i=1}^k H^*_c(\tilde{M}_i) \rightarrow H^*_c(\tilde{M}) \:. \]
Now the individual cohomology groups are given in Proposition~\ref{prpdisccmp}.
For example, in the case that~$\beta \neq 0$ for all discrete lines, we obtain
\[ H^0_c(\tilde{M}) = 0 \qquad \text{and} \qquad
H^1_c(\tilde{M}) = C_c(\Z, \R)^k\:. \]

\subsubsection{Discrete Fourier Transform} \label{subsec:Lattic-fourier}
As an alternative to applying the Künneth formula, we now want to use
the discrete Fourier transform to compute the cohomology of the lattice system.

We first recall the basics on the discrete Fourier transform.
For simplicity, we begin with a function~$\phi : \tilde{M} \rightarrow \C$
with compact support.
We introduce its Fourier transform, denoted by~$\hat{\phi}$ by
\beq \label{hatphi}
\hat{\phi}(p) = \varepsilon^k \sum_{x \in \tilde{M}} \phi(x)\: e^{i p x} \:.
\eeq
This function is periodic, which leads us to regard the
momenta~$p$ on a torus of side length~$2 \pi/\varepsilon$,
\[ p \in \hat{M} := \R^k \Big/ \Big( \frac{2 \pi}{\varepsilon}\: \Z \Big)^k \:. \]
It is often convenient to restrict attention to the fundamental domain
of the torus, which as in in solid state physics we refer to as the
{\em{first Brillouin zone}} denoted by~$\hat{M}_0$,
\[ p \in \hat{M}_0 := \Big( -\frac{\pi}{\varepsilon}, \frac{\pi}{\varepsilon} \Big)^k \:. \]
Then the inverse Fourier transform takes the form
\beq \label{phix}
\phi(x) = \int_{\hat{M}_0} \hat{\phi}(p)\: e^{-i p x}\: \frac{d^kp}{(2 \pi)^k} \qquad
\text{with~$x \in \tilde{M}$}\:.
\eeq
The normalization constants are verified by the computation
\[ \int_{\hat{M}_0} \hat{\phi}(p)\: e^{-i p x}\: d^kp = 
\sum_{x \in \tilde{M}} \int_{\hat{M}_0} \phi(y)\: e^{i p (y-x)}\:d^kp
= \delta_{xy}\:\phi(x)\: \Big( \frac{2 \pi}{\varepsilon} \Big)^k \:. \]

According to~\eqref{LL}, the Lagrangian is described by a function~$L$
on~$\F=\R^k$. 
For technical simplicity, we will assume that this function is of the Schwartz class.
We denote its (ordinary continuous) the Fourier transform also with
an additional hat, i.e.\
\beq \label{Lxi}
L(\xi) = \int_{\hat{\F}} \hat{L}(p)\: e^{-i p \xi}\:  \frac{d^kp}{(2 \pi)^k} \qquad
\text{with~$\hat{\F} := \R^k$ and~$\xi \in \F$}\:.
\eeq

We are interested in the situation when the discretization length~$\varepsilon$ is very small.
If we assume that the function~$\hat{L}$ is compactly supported, we can arrange that
its support lies in the first Brillouin zone. This assumption simplifies the analysis as follows.
\begin{Lemma}\label{lemma:Fourier-vanish-support}
Assuming that
\beq \label{hatLcompact}
\supp \hat{L} \subset \hat{M}_0 \:,
\eeq
the following implication holds,
\beq \label{Limply}
(\L_{\tilde{\rho}} \, f)|_{\tilde{M}}=0 \qquad \Longrightarrow \qquad (\L_{\tilde{\rho}} \, f)=0 \text{ on~$\F$} \:.
\eeq
\end{Lemma}
\Proof Note that, as a consequence of~\eqref{hatLcompact},
the formulas for the inverse Fourier transforms~\eqref{Lxi}
and~\eqref{phix} coincide. In other words, the discrete inverse Fourier transform 
coincides with the ordinary Fourier transform restricted to the lattice. Moreover,
\begin{align}
(\L_{\tilde{\rho}} \, \phi)(x) &= \sum_{y \in \tilde{M}} \L(x,y)\: \phi(y) \notag \\
&= \sum_{y \in \tilde{M}} 
\bigg( \int_{\hat{M}_0} \hat{L}(p)\: e^{-i p (x-y)}\:  \frac{d^kp}{(2 \pi)^k} \bigg)
\bigg( \int_{\hat{M}_0} \hat{\phi}(q)\: e^{-i q y}\:  \frac{d^kq}{(2 \pi)^k} \bigg) \notag \\
&= \frac{1}{\varepsilon^k} \int_{\hat{M}_0} \hat{L}(q)\: \hat{\phi}(q)\: e^{-i q x}\:
\frac{d^kq}{(2 \pi)^k} \:, \label{back}
\end{align}
where in the last step we combined~\eqref{phix} and~\eqref{hatphi}.

Assume that the function~$(\L_{\tilde{\rho}} \, f)|_{\tilde{M}}$ vanishes.
Then its discrete Fourier transform
(as defined by~\eqref{hatphi}) also vanishes. In particular, its restriction to the first Brillouin
zone is zero. Therefore, the integrand in~\eqref{back} is zero, giving the result.
\QED

This result simplifies the equivalence classes of sections as follows.
\begin{Corollary} \label{corsim}
Under the assumption~\eqref{hatLcompact}, for all sections~$\eta \in \Gamma^\infty_\loc(\tilde{M}, T^{s,r}\tilde{M})$
\[ \eta \simeq^0_{\tilde{M}} 0 \qquad \Longrightarrow \qquad
\eta \simeq^\ell_{\tilde{M}} 0 \quad \text{for all~$\ell \in \N_0$}\:. \]
\end{Corollary}
\Proof The relation~$\eta \simeq^0_{\tilde{M}} 0$ means by definition
that the function~$\L_{\tilde{\rho}} \, \eta$ vanishes on~$\tilde{M}$. In view of~\eqref{Limply},
this function vanishes identically. Therefore, also all its derivatives vanish on~$\tilde{M}$,
giving the result.
\QED

Our goal is to show that the sequence~\eqref{simplecomplex} is exact.
Thus let~$\omega \in \Omega^{s,s,k-s}$ with~$s \in \{1, \ldots, k\}$ be closed.
Our task is to show that~$\omega$ is exact, i.e.\ that there
is~$\eta \in \Omega^{s-1,s-1,k-s+1}$ with~$d\eta = \omega$.
First, we need to take into account that~$\omega$ may involve weak derivatives
to the order at most~$s$. We introduce the short notation
\[ \L_{\tilde{\rho}}\, \omega = h \qquad \text{with} \qquad
h(x) := \sum_{|\kappa| \leq s} \big( L* D^\kappa \omega^\kappa \big)(x)\:, \]
and the derivatives act on~$L$ after formal integration by parts.
By definition of the exterior derivative, the closedness of~$\omega$ implies
that the mollification~$(\L_{\tilde{\rho}} \, \omega)$ (being a smooth form on~$\F$)
is closed in the sense that~$d (\L_{\tilde{\rho}} \, \omega)$ vanishes on~$\tilde{M}$
(where~$d$ is the standard exterior derivative in~$\R^k$).
By Lemma~\ref{lemma:Fourier-vanish-support}, this function vanishes on all of~$\F$,
\[ d (\L_{\tilde{\rho}} \, \omega) = 0 \quad \text{on~$\F$} \:. \]
Therefore, taking the line integrals as in the usual proof of Poincar{\'e}'s lemma~\eqref{Pomega}, we obtain a form whose exterior derivative given by
\beq \label{dPo}
d P(\omega) = \L_{\tilde{\rho}} \, \omega \:.
\eeq
This procedure looks promising, but there are still two issues:
\bitem
\item[(i)] The differential forms obtained by taking line integrals~\eqref{Pomega}
are in general no longer compactly supported.
\item[(ii)] The right side in~\eqref{dPo} does not give~$\omega$ as desired,
but we get~$\L_{\tilde{\rho}} \, \omega$ instead.
\eitem
The first issue can be addressed by working with {\em{distributional Fourier transforms}}.
More precisely, on the lattice we consider {\em{bounded functions}},
$\phi \in L^\infty(\tilde{M}, d\rho)$. We can work consistently within this class
of functions, because taking convolutions and derivatives and line integrals
thereof, we again get bounded functions.
Consequently, the Fourier transform~$\hat{\phi}$ in~\eqref{hatphi}
is defined in the distributional sense, giving a subclass of tempered distributions
on~$\hat{M}$ (which we do not need to specify in detail in momentum space).
This resolves issue~(i).

The second issue is more serious. Using that the exterior derivative commutes
with convolutions, the desired differential form~$\eta$ with~$d\eta = \omega$
can be obtained on the formal level by
\beq \label{etaformal}
\eta = \L_{\tilde{\rho}}^{-1} \big( dP(\omega) \big) \:.
\eeq
However, it is far from obvious if or under which assumptions the inverse exists.
In order to work this out in detail, we need to delve deeper into the description
in momentum space. In preparation, we rewrite the operator~$P$ in~\eqref{Pomega}
in momentum space. Interestingly, we again get line integrals, but they 
extend along a ray from~$p$ to infinity.
\begin{Lemma} \label{lemmaPwmom}
Transforming the operator in~\eqref{Pomega} to momentum
space, we obtain
\begin{align}
\widehat{P(\omega)}(p)
= -i \sum_{|\kappa|\leq m}\sum^s_{\alpha=1}(-1)^{\alpha-1}
\frac{\partial}{\partial p^{i_\alpha}} & \int_0^1 t^{s-1-k} \:
\hat{h}_{i_1 \cdots i_s}
\Big( \frac{p}{t} \Big) \:dt \notag \\
&\times dx^{i_1} \wedge \cdots \wedge \widehat{dx^{i_\alpha}} \wedge
\cdots \wedge dx^{i_s} \:. \label{Pwmom}
\end{align}
\end{Lemma}
\Proof We write the differential form~$\omega$ as a Fourier integral,
\[ \big( \L_{\tilde{\rho}}(\hat{\omega} \big)(x) = \int_{M_0}
\hat{h}(p)\: e^{-i p x}\: \frac{d^kp}{(2 \pi)^k} \:. \]
We now apply the line integrals in~\eqref{Pomega} to obtain
\begin{align*}
P(\omega)(x) &=
\sum_{|\kappa|\leq m}\sum^s_{\alpha=1}(-1)^{\alpha-1}
\int_0^1 t^{s-1}
\bigg( \int_{M_0}
\hat{h}_{i_1 \cdots i_s}(p)\: e^{-i p t x}\: \frac{d^kp}{(2 \pi)^k} \bigg)
\:x^{i_\alpha}\,dt \\
&\qquad\qquad\qquad\qquad\qquad\qquad\qquad\qquad\qquad
\times dx^{i_1} \wedge \cdots \wedge \widehat{dx^{i_\alpha}} \wedge \cdots \wedge dx^{i_s} \:.
\end{align*}
The factor~$x^{i_\alpha}$ can be rewritten as a $p$-derivative acting on the
plane wave. Integration by parts gives
\begin{align*}
P(\omega)(x) &=
\sum_{|\kappa|\leq m}\sum^s_{\alpha=1}(-1)^{\alpha-1}
\int_0^1 \frac{t^{s-1}}{i t}
\bigg( \int_{M_0}
\bigg( \frac{\partial}{\partial p^{i_k}} \hat{h}_{i_1 \cdots i_s}(p) \bigg)\: e^{-i p t x}\: \frac{d^kp}{(2 \pi)^k} \bigg)\,dt \\
&\qquad\qquad\qquad\qquad\qquad\qquad\qquad\qquad\qquad
\qquad
\times dx^{i_1} \wedge \cdots \wedge \widehat{dx^{i_\alpha}} \wedge \cdots \wedge dx^{i_s} \:.
\end{align*}
Interchanging the integrals and transforming to the new
integration variable~$p'=pt$, we obtain
\[ P(\omega)(x) = \int_{M_0} \widehat{P(\omega)}(p')\: e^{-i p' x}\:
\frac{d^kp'}{(2 \pi)^k} \]
with~$\widehat{P(\omega)}$ according to~\eqref{Pwmom}.
\QED

Now we can formulate sufficient conditions on the Lagrangian which ensure
that the cohomology is trivial. 
\begin{Prp} \label{prppl} Assume that the Fourier transform of the Lagrangian
(see~\eqref{LL} and~\eqref{Lxi}) is supported in the first Brillouin zone~\eqref{hatLcompact}, 
Moreover, assume that there is a constant~$C>0$ such that for all~$p \in \hat{M}_0$,
the following inequalities hold,
\begin{align}
\Big| \hat{L} \Big( \frac{p}{t} \Big) \Big| &\leq \frac{C}{1+\|p\|^s}\: \big| \hat{L}(p) \big| 
&& \hspace*{-1.5cm} \text{for all~$t \in (0,1)$} \label{Lc1} \\
\big| v^j \partial_j \hat{L}(p) \big| &\leq C\: \|v\|\: \big|\hat{L}(p) \big| &&
\hspace*{-1.5cm} \text{for all~$v \perp p$}\:. \label{Lc2}
\end{align}
Then the cohomology of the differential complex~\eqref{simplecomplex} is trivial,
\begin{equation*}
    H^*(\tilde{M})=0 \:.
\end{equation*}
\end{Prp} \noindent
Before coming to the proof, we point out that the conditions~\eqref{Lc1} and~\eqref{Lc2}
are quite strong. We expect that they could be weakened with more technical effort,
but for brevity we shall not go into this here.
But we note that the above conditions are satisfied if~$\L$ is chosen as a smooth and
compactly supported function which depends only on the radial variable~$\|p\|$
and is monotone decreasing.
Namely, in this case, the left side of~\eqref{Lc2} vanishes.
Moreover, the inequality~\eqref{Lc1} is satisfied with~$C=1$.
In this way, one sees that our proposition allows for non-trivial examples.

\Proof[Proof of Proposition~\ref{prppl}.] Our task is to give the formal
expression~\eqref{etaformal} a precise mathematical meaning.
Since the operator~$\L_{\tilde{\rho}}$ corresponds in momentum space to multiplying by the
function~$\hat{L}$, its inverse in~\eqref{etaformal} corresponds to dividing by this function.
Using the formula in Lemma~\ref{lemmaPwmom}, we are led to analyzing the function
\[ A(p) := \frac{1}{\hat{L}(p)} \frac{\partial}{\partial p^{i_\alpha}} \int_0^1 t^{s-1-k} \:
 \big( \hat{h}_{i_1 \cdots i_s} \big)
\Big( \frac{p}{t} \Big) \:dt \:. \]
We first pull the $p$-derivative outside,
\begin{align}
A &= \frac{\partial}{\partial p^{i_\alpha}} \bigg( \frac{1}{\hat{L}(p)} \int_0^1 t^{s-1-k} \:
\hat{h}_{i_1 \cdots i_s}
\Big( \frac{p}{t} \Big) \:dt \bigg) \label{p1} \\
&\quad- \frac{\partial_{i^\alpha} \hat{L}(p)}{\hat{L}(p)^2}
\int_0^1 t^{s-1-k} \:
\hat{h}_{i_1 \cdots i_s} \Big( \frac{p}{t} \Big) \:dt \:. \label{p2}
\end{align}
In~\eqref{p1} we can apply~\eqref{Lc1} to obtain a well-defined distribution.
In~\eqref{p2}, on the other hand, we decompose the derivative into its
radial and angular parts,
\begin{align*}
\partial_{i^\alpha} \hat{L}(p) &= 
\frac{p_{i \alpha} p^j}{p^2}\: \partial_j \hat{L}(p) + \Big( \delta^j_{i_\alpha}
- \frac{p_{i \alpha} p^j}{p^2} \Big)\: \partial_j \hat{L}(p) \:.
\end{align*}
Estimating the angular derivative with the help of~\eqref{Lc2} again gives
a well-defined distribution. Therefore, it remains to consider the radial
derivative in~\eqref{p2}. We write it as follows,
\begin{align}
&-\frac{p_{i \alpha} p^j}{p^2}\:  \frac{\partial_j \hat{L}(p)}{\hat{L}(p)^2}
\int_0^1 t^{s-1-k} \: \hat{h}_{i_1 \cdots i_s}
\Big( \frac{p}{t} \Big) \:dt \notag \\
&= \frac{p^{i \alpha} p^j}{p^2} \: \bigg( \frac{\partial}{\partial p^j} \frac{1}{\hat{L}(p)} 
\bigg) \int_0^1 t^{s-1-k} \:
\big( \hat{L}\:\hat{\omega}_{i_1 \cdots i_s} \big) \Big( \frac{p}{t} \Big) \:dt \notag \\
&= \frac{\partial}{\partial p^j}  \bigg( \frac{p^{i \alpha} p^j}{p^2} \:\frac{1}{\hat{L}(p)} 
\int_0^1 t^{s-1-k} \:
\hat{h}_{i_1 \cdots i_s} \Big( \frac{p}{t} \Big) \:dt \bigg) \label{p3} \\
&\quad\:- \bigg( \frac{\partial}{\partial p^j} \frac{p^{i \alpha} p^j}{p^2} \bigg)\:
\frac{1}{\hat{L}(p)} \int_0^1 t^{s-1-k} \:
\big( \hat{L}\:\hat{\omega}_{i_1 \cdots i_s} \big) \Big( \frac{p}{t} \Big) \:dt \label{p4} \\
&\quad\:- \frac{p^{i \alpha} p^j}{p^2} \: \frac{1}{\hat{L}(p)} 
\frac{\partial}{\partial p^j}  \int_0^1 t^{s-1-k} \:
\hat{h}_{i_1 \cdots i_s} \Big( \frac{p}{t} \Big) \:dt \label{p5}
\end{align}
The summands~\eqref{p3} and~\eqref{p4} can again be estimated with the
help of~\eqref{Lc1}. The last summand~\eqref{p5}, on the other hand, can be rewritten
using integration-by-parts as follows,
\begin{align*}
&\frac{1}{\hat{L}(p)} \: p^j
\frac{\partial}{\partial p^j}  \int_0^1 t^{s-1-k} \:
\big( \hat{L}\:\hat{\omega}_{i_1 \cdots i_s} \big) \Big( \frac{p}{t} \Big) \:dt \\
&= -\frac{1}{\hat{L}(p)} 
\int_0^1 t^{s-k} \: \frac{d}{dt}
\hat{h}_{i_1 \cdots i_s} \Big( \frac{p}{t} \Big) \:dt \\
&= -\frac{1}{\hat{L}(p)} \: \big( \hat{L}\:\hat{\omega}_{i_1 \cdots i_s} \big)(p)
+ \frac{s-k}{\hat{L}(p)}
\int_0^1 t^{s-1-k}\:\hat{h}_{i_1 \cdots i_s} \Big( \frac{p}{t} \Big) \:dt \:.
\end{align*}
Now we can again use~\eqref{Lc1}, concluding the proof.
\QED

\subsection{$\L$-Star-Shaped Subsets of the $k$-Dimensional Lattice}
\label{sec-lattice-subset}
Next, we consider the Poincar{\'e}-type lemma for a subset~$U \subset \tilde{M}$
of the $k$-dimensional lattice of the previous section.
In order to keep the setting as simple as possible, we restrict attention to
one-forms and the simplest non-trivial case~$m=1$ and~$\ell=0$. 
Begin with the case~$U=\tilde{M}$. Thus we let~$\omega \in \Omega^{1,1,0}(\tilde{M})$
be a globally defined one-form. Then the formula~\eqref{Pomega} for the integration
along rays simplifies to
\[ 
P(\omega)(x) := \sum_{|\kappa| \leq 1} \Big( \int_0^1
L * D^\kappa \omega^\kappa_i(tx) \:dt \Big) \:x^i\:. \]
For the function in the integrand we use the short notation
\beq \label{hidef}
h_i(x) := \sum_{|\kappa| \leq 1} \big( L * D^\kappa \omega^\kappa_i\big)(x) \:.
\eeq
The Fourier transform of~$h_i$ has the simple form
\beq \label{hath}
\hat{h}_i(p) = \sum_{|\kappa| \leq 1} (-i p)^\kappa\:
\hat{L}(p)\:  \hat{\omega}^\kappa_i(p) \:.
\eeq
We note for clarity that the functions~$\hat{\omega}^\kappa_i$ are periodic
(being the Fourier transforms of functions defined on the lattice). The function~$\hat{L}$,
however, is not periodic (being the Fourier transform of the test function on~$\F$).
Consequently, the product in~\eqref{hath} is {\em{not}} periodic, corresponding to the
fact that the mollification in position space is defined on all of~$\F$, not only on the lattice.

Having a bounded domain in mind, we assume that the set~$U$ is finite.
Using the embedding~$\tilde{M} \subset M$, the set~$U$ is a finite
subset of~$M$.
By a translation we may assume that the origin~$\0 \in M$ lies in~$U$.
We want to analyze what the definition of~$U$ being~$\L$-star-shaped
(see Definition~\ref{Lstarshaped}) means for the set~$U$.
In preparation we introduce the star-shaped region generated by~$U$,
\[ \sshape(U) := \big\{ \lambda x \:\big|\: x \in U, \lambda \in [0,1] \big\} \:. \]
An expected minimal requirement is that every point of~$\sshape(U)$
should be~$\varepsilon$-close to a point of~$U$,
\[ \dist(x, U)< \varepsilon \qquad \text{for all~$x \in \sshape(U)$} \:.\]
For the following construction, we need that this
property holds even in the~$\L$-neighbor\-hood of~$U$:
\begin{Def} \label{defepsstarshaped}
The region~$U \subset \tilde{M}$ is called {\bf{$\varepsilon$-star-shaped}} if
\beq \label{epsclose}
\dist \big(x, B_\L(U) \big)) < \varepsilon \qquad \text{for all~$x \in \sshape
\big(B_\L(U) \big)$} \:,
\eeq
\end{Def}

We shall assume in what follows that this condition
holds. As we will see, more conditions will be needed in order to ensure that~$U$
is~$\L$-star-shaped; we will collect all the additional conditions on the way.

We let~$\omega \in \Omega^{1,1,\ell}(U \subset \B_\L(U))$ be closed. This
means that
\[ 
\big(\L_{\tilde{\rho}}(d \omega)\big)(x) =0 \qquad \text{for all~$x \in U$}\:. \]
Note that, at this stage, this equation only holds on the finite set~$U$.
We want to extend this equation to the boundary layer~$B_\L(U) \setminus U$.
The idea is to arrange this by a suitable choice of representatives.
\begin{Def} \label{defbep}
The set~$U \subset \tilde{M}$ has the
{\bf{boundary extension property}} if every 
closed differential form~$\omega \in \Omega^{1,1,\ell}(U \subset \B_\L(U))$
can be represented by a
section~$\omega \in \Gamma^\infty_\loc(\tilde{M}, T^{1,1}(\tilde{M}))$
with the property
\beq \label{bextend}
x \,\llcorner \big(\L_{\tilde{\rho}}(d \omega)\big)(x) =0 \qquad \text{for all~$x \in B_\L(U)$}\:.
\eeq
\end{Def} \noindent
Let us discuss how to construct boundary extensions. We first note that
we are free to choose~$\omega$ arbitrarily on the set~$\tilde{M}
\setminus B_\L(U)$. In order to determine~$\omega$ on this set, we consider
the mapping
\beq \label{Phimap}
\begin{split}
&\Phi \::\: \Gamma^\infty_\loc(\tilde{M}, T^{1,0}\big( \tilde{M} \setminus B_\L(U) \big)
\rightarrow C^\infty \big(B_\L(U) \setminus U, \R^k \big) \\
&\big( \Phi(\alpha) \big)_j(x) := \int_{\tilde{M} \setminus B_\L(U)}
x^i D_{1,[i} \L(x,y)\: \omega_{j]}(y)\: d\tilde{\rho}(y)
\end{split}
\eeq
This is a linear mapping between finite-dimensional vector spaces.
Due to the anti-symmetrization, we know that~$x^j ( \Phi(\alpha))_j(x)=0$,
which reduces the maximal dimension of the image by a factor~$(k-1)/k$.
Therefore, it is a sensible to assume that~$\Phi$ is surjective.
For brevity, we will not verify in detail which subsets~$U \subset \tilde{M}$
of the lattice have the boundary extension property. Instead, we simply
take it as an additional technical assumption which ensures that~$U$ is
$\L$-star-shaped. The surjectivity of the map~$\Phi$ in~\eqref{Phimap}
implies that there is a constant~$C$ such that for all boundary extensions,
\beq \label{qext}
\|\omega\|_{L^2 \big( \tilde{M} \setminus B_\L(U), d\tilde{\rho} \big)} \leq C\:
\|\omega\|_{L^2 \big( \tilde{M}, d\tilde{\rho} \big)} \:.
\eeq

It is important to note that the equation~\eqref{bextend} holds only
at lattice points. Next, we show that
this relation also holds on the star-shaped hull~$\sshape(B_\L(U))$, up to an error
term which is linear in~$\varepsilon$.
\begin{Lemma} \label{lemmaerror}
For any~$x \in \sshape(B_\L(U))$,
\beq \label{dsmall}
\big\| x \,\llcorner \big(\L_{\tilde{\rho}}(d \omega)\big)(x) \big\| \leq c \varepsilon\,\|x\|
\: \|\omega\|_{L^2(\tilde{M}, d\tilde{\rho})}
\eeq
with the constant~$c$ given by
\[ c = \|L\|_{L^2(M)} + \|D L\|_{L^2(M)} +\|D^2L\|_{L^2(M)} \:. \]
\end{Lemma}
\Proof Given~$y \in \sshape(B_\L(U))$, we know from~\eqref{epsclose} that
there is an~$\varepsilon$-close lattice point~$z \in B_\L(U)$.
Therefore, using~\eqref{bextend}, we obtain
\[ \big\| y \,\llcorner \big(\L_{\tilde{\rho}}(d \omega)\big)(y) \big\|
= \big\| y \,\llcorner \big(\L_{\tilde{\rho}}(d \omega)\big)(y) 
- x \,\llcorner \big(\L_{\tilde{\rho}}(d \omega)\big)(x) \big\| \:. \]
The last difference can be estimated with the mean value inequality.
Using the Schwarz inequality, we obtain
\[ \big\| y \,\llcorner \big(\L_{\tilde{\rho}}(d \omega)\big)(y) \big\|
\leq \varepsilon \|x\|\, \Big( \|L\|_{L^2(M)} + \|D L\|_{L^2(M)} +\|D^2L\|_{L^2(M)} \Big)
\|\omega\|_{L^2(\tilde{M}, d\tilde{\rho})} \:, \]
concluding the proof.
\QED

Again using the notation~\eqref{hidef},
\begin{align*}
\big( \partial_j P(\omega) \big)(x)
&= \partial_j \bigg( \int_0^1 x^k h_k(tx)\: dt \bigg)
= \int_0^1 h_j(tx)\: dt + \int_0^1 t x^k \:\partial_k h_j(tx)\: dt \\
&= \int_0^1 h_j(tx)\: dt + \int_0^1 t x^k \:\partial_{[k} h_{j]}(tx)\: dt
+ \int_0^1 t \: \bigg( \frac{d}{dt} h_j(tx) \bigg)\: dt  \\
&=\int_0^1 t x^k \:\partial_{[k} h_{j]}(tx)\: dt + h_j(x) \:,
\end{align*}
where in the last line we integrated by parts.
From Lemma~\ref{lemmaerror} we know that the obtained integrand
vanishes up to small errors. More precisely,
\[ d P(\omega) = \L_{\tilde{\rho}}(\omega) + \Delta \omega \qquad
\text{in~$\sshape(B_\L(U))$} \]
with an error term~$\Delta \omega$, which according to~\eqref{dsmall} and~\eqref{qext}
can be estimated by
\begin{align}
\|\Delta \omega(x)\| &\leq \varepsilon\: c\, \text{diam}(B_\L(U))\: \|\omega\|_{L^2(\tilde{M}, d\tilde{\rho})} \notag \\
&\leq \varepsilon\: c\, \text{diam}(B_\L(U))\: (1+C)\:\|\omega\|_{L^2(B_\L(U), d\tilde{\rho})} \:.
\label{Delomega}
\end{align}
Restricting to the lattice points, we
conclude that
\beq \label{Lrho2}
d P(\omega) = \L_{\tilde{\rho}}(\omega) + \Delta \omega \qquad
\text{on~$B_\L(U)$} \:.
\eeq

It remains to show that this equation can be treated iteratively.
It is most convenient to work with the individual tensor components,
being complex-valued functions in the Hilbert space~$\H_U :=L^2(B_\L(U), d\tilde{\rho})$.
Note that this is a
finite-dimensional vector space,
\[ \dim \H = \# \big( B_\L(U) \big) < \infty \:. \]
Next, similar to~\eqref{Lrhodef} we introduce the operator~$\L_{\tilde{\rho}}|_U$ on~$\H$ by
\[ \L_{\tilde{\rho}}|_U : \H_U \rightarrow \H_U \:, \qquad
(\L_{\tilde{\rho}}|_U \psi)(x) := \frac{1}{\s} \sum_{y \in B_\L(U)} \L(x,y)\, \psi(y) \:. \]
Using that~$\tilde{\rho}$ is a minimizer of the causal variational principle,
this operator has the following useful properties
\begin{Lemma} The operator~$\L_{\tilde{\rho}}|_U$ is symmetric. Moreover, it is
positive and bounded from above by one,
\[ 0 \leq \L_{\tilde{\rho}}|_U \leq 1 \:. \]
\end{Lemma}
\Proof The positivity of~$\L_{\tilde{\rho}}|_U$ follows by considering second variations;
for details see~\cite[Lemma~3.5]{support} or~\cite{positive}. Next, using that the Lagrangian
is non-negative, $\L(x,y) \geq 0$, the Perron-Frobenius theorem shows that the
eigenvector corresponding to the largest eigenvalue can be chosen to be everywhere
positive. This also implies that the largest eigenvalue is bounded from above by one
(for details see~\cite[Section~1.5]{continuum}).
\QED

This lemma does not rule out that~$\L_{\tilde{\rho}}|_U$ has a kernel, But we can use the fact the
the kernel of~$\L$ is modded out when taking equivalence classes with respect to~$\simeq^0_U$ (at least, provided that similar as shown on the whole lattice
in Corollary~\ref{corsim}, the equivalence relation~$\simeq^0_U$
is equivalent to~$\simeq^\ell_U$).
With this in mind, we may restrict attention to the orthogonal complement of the kernel
of~$\L_{\tilde{\rho}}|_U$. Using that we are in finite dimensions, there is~$\delta$ with
\beq \label{gap}
0 < \delta \leq \L_{\tilde{\rho}}|_U \leq 1\:.
\eeq
Using this estimate in~\eqref{Lrho2}, we obtain
\begin{align*}
\big\| \omega - d P(\omega) \big\| &= \big\| (1- \L_{\tilde{\rho}}|_U)\: \omega - \Delta \omega \big\|
\leq (1-\delta)\: \|\omega\| + \| \Delta \omega \| \\
&\leq \bigg( 1-\delta+
\varepsilon\: c\, \text{diam}(B_\L(U))\: (1+C) \:\sqrt{\tilde{\rho} \big( B_\L(U) \big)} \bigg)\: \|\omega\| \:,
\end{align*}
where in the last step we applied~\eqref{Delomega} and the H\"older inequality. 
By choosing~$\varepsilon$
sufficiently small, we can arrange that the condition~\eqref{iterin} is satisfied,
so that Lemma~\ref{lemmapoincare} applies.

Our findings are summarized as follows.
\begin{Prp} Assume that the region~$U \subset \tilde{M}$ is~$\varepsilon$-star-shaped
(see Definition~\ref{defepsstarshaped}). Moreover, assume that the boundary
extension property (see Definition~\ref{defbep}) holds in the quantitative
form~\eqref{qext}. Finally, assume that the operator~$\L_{\tilde{\rho}}|_U$ 
on the Hilbert space~$\H_U:=L^2(B_\L(U), d\tilde{\rho})$ has a spectral gap~\eqref{gap}.
Then, for a sufficiently small lattice spacing~$\varepsilon$, every closed
one form~$\omega \in \Omega^{1,1,\ell}(U \subset \tilde{M})$ is exact.
\end{Prp} \noindent
We finally point out that this result relies crucially on the assumptions that
the constants~$C$ and~$\delta$ in~\eqref{qext} and~\eqref{gap} can be chosen
uniformly in the lattice spacing. These assumptions could be
verified for specific choices of the Lagrangian. For brevity, we do not enter the details
of the resulting analysis.

\Thanks{{{\em{Acknowledgments:}}
F.v.d.T.\ would like to thank the ALGANT consortium for support and the opportunity
to benefit from their study program.
We are grateful to the ``Universit\"atsstiftung Hans Vielberth'' for support.
N.K.'s research was also supported by the NSERC grant RGPIN~105490-2025.


\bibliographystyle{amsplain}
\providecommand{\bysame}{\leavevmode\hbox to3em{\hrulefill}\thinspace}
\providecommand{\MR}{\relax\ifhmode\unskip\space\fi MR }
\providecommand{\MRhref}[2]{%
  \href{http://www.ams.org/mathscinet-getitem?mr=#1}{#2}
}
\providecommand{\href}[2]{#2}

\end{document}